\documentclass[
nofootinbib,
a4paper,
accepted=2025-09-05,
]{quantumarticle}
\pdfoutput=1
\usepackage[utf8]{inputenc}
\usepackage{physics}
\usepackage{bbold}
\usepackage{amsmath}
\usepackage{amssymb}
\usepackage{latexsym}
\usepackage{graphicx}
\usepackage{amsthm}
\usepackage{hyperref}
\usepackage{lipsum}
\usepackage[english]{babel}
\usepackage{slashed}
\usepackage[compat=1.1.0]{tikz-feynman}
\usepackage{appendix}
\usepackage{subfig}
\usepackage{multirow}
\usepackage{mathtools}
\usepackage{subfig}
\usepackage[normalem]{ulem}
\usepackage{comment}
\usepackage{ragged2e} 
\usepackage{csquotes}
\usepackage[numbers,sort&compress]{natbib}

\theoremstyle{plain}
\newtheorem{proposition}{Proposition}
\newtheorem{axiom}{Axiom}
\theoremstyle{definition}
\newtheorem{definition}{Definition}
\newtheorem{theorem}{Theorem}

\newtheorem{conjecture}{Conjecture}
\newtheorem{algorithm}{Algorithm}
\newtheorem{lemma}{Lemma}
\newtheorem{notation}{Notation}

\DeclareFontFamily{U}{rcjhbltx}{}
\DeclareFontShape{U}{rcjhbltx}{m}{n}{<->rcjhbltx}{}
\DeclareSymbolFont{hebrewletters}{U}{rcjhbltx}{m}{n}
\newcommand{\lamed}{\mathfrak{t}}
\newcommand{\tfrak}{\mathfrak{t}}
\newcommand{\ayin}{\mathfrak{l}}
\newcommand{\lfrak}{\mathfrak{l}}
\newcommand{\tsadi}{\mathfrak{s}}
\newcommand{\sfrak}{\mathfrak{s}}
\newcommand{\shin}{\mathfrak{s}_{\mathrm{tot}}}
\newcommand{\stotfrak}{\mathfrak{s}_{\mathrm{tot}}}

\providecommand\sgn{\text{sign}}

\newcommand{\acd}[1]{{\color{black} #1}}
\newcommand{\vik}[1]{{\color{black} #1}}
\newcommand{\sam}[1]{{\color{black} #1}}

\def\pprec{\mathrel{\scalebox{.9}[1]{$\prec$}\mkern-5mu%
  \scalebox{.4}[1]{$\prec$}\mkern-5.5mu\scalebox{.4}[1]{$\prec$}}}

\begin{document}

\title{Knot invariants and indefinite causal order}
\author{Samuel Fedida}
\affiliation{Centre for Quantum Information and Foundations, DAMTP, Centre for Mathematical Sciences, University of Cambridge, Wilberforce Road, Cambridge CB3 0WA, UK}
\author{Anne-Catherine de la Hamette}
\affiliation{University of Vienna, Faculty of Physics, Vienna Doctoral School in Physics, and Vienna Center for Quantum Science and Technology (VCQ), Boltzmanngasse 5, A-1090 Vienna, Austria}
\affiliation{Institute for Quantum Optics and Quantum Information (IQOQI), Austrian Academy of Sciences, Boltzmanngasse 3, A-1090 Vienna, Austria}
\author{Viktoria Kabel}
\affiliation{Institute for Quantum Optics and Quantum Information (IQOQI), Austrian Academy of Sciences, Boltzmanngasse 3, A-1090 Vienna, Austria}
\affiliation{University of Vienna, Faculty of Physics, Vienna Doctoral School in Physics, and Vienna Center for Quantum Science and Technology (VCQ), Boltzmanngasse 5, A-1090 Vienna, Austria}
\author{\v{C}aslav Brukner}
\affiliation{University of Vienna, Faculty of Physics, Vienna Doctoral School in Physics, and Vienna Center for Quantum Science and Technology (VCQ), Boltzmanngasse 5, A-1090 Vienna, Austria}
\affiliation{Institute for Quantum Optics and Quantum Information (IQOQI), Austrian Academy of Sciences, Boltzmanngasse 3, A-1090 Vienna, Austria}

\newcommand{\BH}{\mathbf{\mathcal{B}}(\mathcal{H})}
\newcommand{\snabla}{\slashed{\nabla}}
\newcommand{\LH}{\mathbf{\mathcal{L}}(\mathcal{H})}
\newcommand{\R}{\mathbb{R}}
\newcommand{\C}{\mathbb{C}}
\newcommand{\supp}{\text{supp}}
\newcommand{\Wlog}{\text{Without loss of generality }}
\newcommand{\Counter}{\text{Counter}}
\newcommand{\Diff}{\text{Diff}}

\begin{abstract}
   We explore indefinite causal order between events in the context of quasiclassical spacetimes in superposition. We introduce several new quantifiers to measure the degree of indefiniteness of the causal order for an arbitrary finite number of events and spacetime configurations in superposition. By constructing diagrammatic and knot-theoretic representations of the causal order between events, we find that the definiteness or maximal indefiniteness of the causal order is topologically invariant. This reveals an intriguing connection between the field of quantum causality and knot theory. Furthermore, we provide an operational encoding of indefinite causal order and discuss how to incorporate a measure of quantum coherence into our classification.
\end{abstract}

\maketitle

\section{Introduction}
Combining the indefiniteness of quantum theory with the dynamical causal structure of general relativity gives rise to the concept of \emph{indefinite causal order}. This idea was first explored in \cite{Hardy_2007} and later developed into a concrete process by \cite{Chiribella_2013, Oreshkov_2012}. Processes that exhibit indefinite causal order are those to which one cannot even assign a probabilistic, classical order. The most studied process with indefinite causal order is the \emph{quantum switch}, a process wherein two operations are applied in a quantum controlled superposition of opposite orders. The quantum switch has been implemented in optical experiments by preparing a photon in a superposition of paths \cite{procopio_experimental_2015,Rubino2017,rubino_experimental_2022} or polarizations \cite{Goswami2018}, which function as a control for the order of gate applications. On the theoretical side, several causal witnesses \cite{Araujo_2015, Branciard_2016} and causal inequalities \cite{Oreshkov_2012, Branciard_2015, Oreshkov_Giarmatzi_2016, Abbott_2016} have since been developed to certify different types of indefiniteness of the causal order. The process matrix formalism (e.g.~\cite{Oreshkov_2012, Araujo_2015}), which offers an abstract representation of processes in terms of supermaps, as well as several other diagrammatic representations such as \cite{Kissinger_2017a, Kissinger_2017b, Pinzani_2020,Barrett2021rev, Vanrietvelde2021routedquantum, Wechs2021, Vilasini_2022PRL, Vilasini_2022PRA, Ormrod2023causalstructurein, Vanrietvelde2023,Pinzani_2023} have proven to be useful tools to capture these ideas on a formal level.

While these works focus on the information theoretic aspects of indefinite causal structures, more recent work has taken a closer look at how such structures might arise in the context of indefinite gravitational fields \cite{zych_2019, SMoller2024gravitational, Paunkovic_2020} and to what extent they can be embedded in a fixed spacetime background \cite{Vilasini_2022PRL, Vilasini_2022PRA, Ormrod2023causalstructurein, Delahamette2022}. In \cite{Delahamette2022}, some of us formalised indefinite causal order between two events in a general relativistic language, extending the notion of event to superpositions of two quasiclassical spacetimes and defining indefinite causal order in terms of proper time differences between two such events.

In the present work, we develop a powerful method for quantifying the causal order between an arbitrary finite number of events in a superposition of an arbitrary finite number of quasiclassical spacetimes. Notably, we establish an intriguing connection between indefinite causal order and knot theory. This includes a representation of indefinite causal order in terms of different types of knots and their Alexander-Conway polynomials. We find that some of the quantifiers of indefinite causal order can be related to knot invariants. \vik{More concretely, }\sam{we show that this holds for extremal cases, and conjecture with support from explicit examples that this generalises}. The topological protection of these quantities then immediately implies the invariance of indefinite causal order under quantum-controlled topology-preserving transformations, reproducing and connecting to various results from the literature \cite{CastroRuiz2018, Delahamette2022}. More generally, we believe that this connection provides a promising tool for analysing and categorising further structural properties of indefinite causal order using the well-studied framework of knot theory.

This paper is organised as follows. In Sec.~\ref{sec: events causal order}, following \cite{Delahamette2022}, we review the notions of quasiclassical spacetimes in superposition, events and their localisation, and causal order as defined through proper time differences, extending these concepts to $N$ events and $M$ spacetimes in superposition. Next, we introduce several different quantifiers for causal order in Sec.~\ref{sec:quantifiers}. Sec.~\ref{sec:knots} establishes the connection of causal order to knot theory by providing a diagrammatic representation of indefinite causal order between events. In Sec.~\ref{sec:operational}, we extend the operational encoding of causal order first presented in \cite{Delahamette2022} to the general case and in Sec.~~\ref{sec:quantumcoherence}, we introduce several quantum mechanical quantifiers  that take measures of coherence into account. Finally, in Sec.~\ref{sec:discussion} we discuss our results, connect them to existing literature, and provide an outlook on open questions and future research directions.

Readers mainly interested in the connection between causal order and knot theory rather than the physical understanding behind causal order in a superposition of spacetimes, may choose to skip most of Sec.~\ref{sec: events causal order}, referring only to definitions \ref{def:causalorder} and \ref{def:orderedcollevents}, as well as Sec.~\ref{sec:quantumcoherence}.
 
\section{Events and Causal Order in a Superposition of Spacetimes} \label{sec: events causal order}

\subsection{Superposition of quasiclassical spacetimes}

We work with superpositions of \textit{quasiclassical spacetimes}. Quasiclassical spacetimes, sometimes called semiclassical spacetimes in the literature, are spacetimes that are assigned a quantum state of the metric field peaked around a classical solution. We would expect such states to arise, at least in an approximation, in an appropriate limit to a full theory of quantum gravity. More generally, their existence is implied by two assumptions, the linearity of quantum theory and the validity of general relativity in the regime under investigation, since they can be regarded as linear superpositions of quasiclassical solutions of Einstein's field equations. We do not, however, want to make any concrete assumptions about the underlying theory but rather take the existence of such states as a starting point whose consequences we explore below.

This requires us to start from the following physical axioms that ought to be verified in the appropriate regime of generic theories of quantum gravity:
\begin{axiom}
    \label{Classical gravitational fields}
    Classical gravitational fields are described by general relativity (GR), i.e.~$\exists$ pseudo-Riemannian manifolds $\{(\mathcal{M}_i,g^{(i)})\}_{i=1}^{N}$ for $N \in \mathbb{N} \cup \{\infty\}$ such that each $g^{(i)}$ is a solution of the Einstein Field Equations (EFE).
\end{axiom}
\begin{axiom}
    \label{Quasiclassical states}
    There exist quasiclassical quantum states $\ket{g^{(i)}}$, associated to the classical metrics $g^{(i)}$.
\end{axiom}
We can now formally work with quantum superpositions of quasiclassical gravitational fields. Such a scenario could conceivably arise, for example, when the gravitational source is placed in a quantum superposition of macroscopically distinct locations \cite{Bose2017,Marletto2017, Christodoulou_2019, delaHamette2021falling, aspelmeyer2021, Foo_2021}. Taking this as our underlying assumption, we can formalise the superposition of different gravitational fields as follows.
\begin{axiom}
    For a quantum system $M$ with Hilbert space $\mathcal{H}_M$ and a position basis labelled $\{\ket{x^{(i)}}_M\}$, which sources a gravitational field, the combined state of metric and quantum system takes the form
\begin{equation}
    \label{eq: Gravitational superposition}
    \ket{\Phi} = \sum_i \alpha_i \ket{\tilde{g}^{(i)}} = \sum_i \alpha_i \ket{x^{(i)}}_M \ket{g^{(i)}}
\end{equation}
with $\alpha_i \in \mathbb{C}$, such that $g^{(i)}$ is fully fixed by the matter distribution localised at $x^{(i)}$.    
\end{axiom}
The states $\ket{x^{(i)}}$ of the position basis of the source $M$ should be regarded as coherent states, to which one can approximately assign not only the position $x^{(i)}$ but also fixed, say zero, momentum. To simplify the notation, we use $\ket{x^{(i)}}$ for these states, however, they should not mistaken for states in which the position is infinitely sharp. The states \eqref{eq: Gravitational superposition} are, of course, not the most general case for a quantum superposition of gravitational fields. For example, one might want to consider superpositions of different vacuum solutions as well. Nevertheless, for the purpose of this work, it will be enough to restrict ourselves to states of the form \eqref{eq: Gravitational superposition}, which allows us to state more clearly the following axiom.
\begin{axiom}
    \label{Orthogonal gravitational fields}
    Macroscopically distinguishable gravitational fields are assigned orthogonal quantum states, that is,
    \begin{equation}
        \braket{\tilde{g}^{(i)}}{\tilde{g}^{(j)}} = 0 \text{ if } x^{(i)} \neq x^{(j)}.
    \end{equation} 
\end{axiom}
This axiom draws its relevance from the fact that different classical solutions of the Einstein field equations are macroscopically distinguishable.

\subsection{Events and causal order in a superposition of quasiclassical spacetimes}

Let us next consider the notion of event in such a superposition of quasiclassical spacetimes. In order to define events in a coordinate-independent way, we follow \cite{Delahamette2022} and define them in terms of worldline coincidences. More precisely, let us consider two worldlines on a spacetime $(\mathcal{M},g)$. We assume that these worldlines cross exactly once: this defines an event $\mathcal{E}$.

\begin{definition}[Event]
    The crossing of two worldlines within a spacetime $(\mathcal{M},g)$ at a spacetime point $p_{\mathcal{E}}\in \mathcal{M}$ defines an \textit{event} $\mathcal{E}$.
\end{definition}

Next, let us consider a superposition of two spacetimes $(\mathcal{M}_{\mathcal{A}},g_{\mathcal{A}})$ and  $(\mathcal{M}_{\mathcal{B}},g_{\mathcal{B}})$. We denote by $\mathcal{E}^{\mathcal{A}}$ and $\mathcal{E}^{\mathcal{B}}$ the same physical event -- that is, the crossing of the same worldlines -- in the respective spacetimes. {This definition of event is consistent with (although not equivalent to) the notion of an event as a CP map in the context of the process matrix formalism \cite{Oreshkov_2012, Araujo_2015}. To fully establish the connection, one would have to understand the worldlines as the worldlines of laboratories, at the intersections of which a CP map is implemented in a way that does not distinguish between the different branches of the superposition. Moreover,} note that this operational identification of events is independent of their location on the spacetime manifold. Nevertheless, one may want to find a representation  such that the same events are also located at the same spacetime points $p_\mathcal{E}^\mathcal{A}\in \mathcal{M}_\mathcal{A}$ and $p_\mathcal{E}^\mathcal{B}\in \mathcal{M}_\mathcal{B}$ -- that they are \emph{localised} across the branches of the superposition. In the presence of diffeomorphism-invariance, however, the question of whether two points across the branches of a superposition are \enquote{the same} or \enquote{different} is non-trivial. In this case, spacetime points do not have any intrinsic physical meaning {\cite{sepholearg, anandan1997classical,Hardy2020,Adlam_2022spacetime,Kabel2024}}. In particular, they can be arbitrarily reshuffled by applying, in general, different diffeomorphisms in each branch of the superposition. As a consequence, whether two events are associated to the same or different manifold points across the superposition can change under a such a transformation. To resolve this issue, we follow \cite{Kabel2024} and identify spacetime points across a superposition by physical field coincidences. Let us therefore assume the existence of sufficiently many functionally independent physical fields in the following axiom.

\begin{axiom}
    \label{Field axiom}
    For each $d$-dimensional spacetime $(\mathcal{M}_\mathcal{A},g^{\mathcal{A}})$ $\exists$ d functionally independent fields $\chi_{(A)}^{\mathcal{A}}$ for $A = 1, ..., d$ such that they define a bijective map $\mathcal{U}_{\mathcal{A}} \to \mathbb{R}^d$ in some open subregion $\mathcal{U}_\mathcal{A} \subset \mathcal{M}_{\mathcal{A}}$ of interest. 
\end{axiom}

The bijectivity condition of axiom \ref{Field axiom} requires that these fields do not repeat their values, i.e.~that they are neither homogeneous in space nor periodic in time within the region of interest. We can now use the $d$ independent fields to construct a map between two different spacetimes and thereby identify spacetime points in a coordinate-independent manner.
\begin{definition}[Comparison map]
    The \textit{comparison map} $C_\chi^{\mathcal{A}\mathcal{B}}: \mathcal{U}_\mathcal{A} \to \mathcal{U}_\mathcal{B}$ associates to each point $p \in \mathcal{U}_\mathcal{A} \subset \mathcal{M}_\mathcal{A}$ a unique counterpart $q = C^{\mathcal{A}\mathcal{B}}_\chi(p) \in \mathcal{U}_\mathcal{B} \subset \mathcal{M}_\mathcal{B}$ where
    \begin{equation}
        C^{\mathcal{A}\mathcal{B}}_\chi := (\chi^{\mathcal{B}})^{-1} \circ \chi^{\mathcal{A}}.
    \end{equation}
\end{definition}
Given a superposition of an arbitrary number $M$ of spacetimes, we can define a comparison map between any pair of spacetimes. The relation between these comparison maps is depicted in Fig.~\ref{fig:Comparison map}. In general, the $\chi$-fields can be chosen freely as long as they satisfy the above bijectivity condition. Here, we choose the $\chi$-fields such that for each event, the points associated to the spacetimes are identified with one another. This way, we ensure that the same events, in the operational sense, are also associated to the same spacetime points across the branches in superposition.
\begin{figure*}
    \centering
    \begin{tikzpicture}[x=0.75pt,y=0.75pt,yscale=-1,xscale=1]
\draw  [fill={rgb, 255:red, 155; green, 155; blue, 155 }  ,fill opacity=0.2 ] (313.1,12.33) -- (460.33,12.33) -- (397.23,64.67) -- (250,64.67) -- cycle ;
\draw  [fill={rgb, 255:red, 155; green, 155; blue, 155 }  ,fill opacity=0.2 ] (313.1,111) -- (460.33,111) -- (397.23,163.33) -- (250,163.33) -- cycle ;
\draw  [fill={rgb, 255:red, 155; green, 155; blue, 155 }  ,fill opacity=0.2 ] (313.1,213) -- (460.33,213) -- (397.23,265.33) -- (250,265.33) -- cycle ;
\draw  [fill={rgb, 255:red, 0; green, 0; blue, 0 }  ,fill opacity=1 ] (356.67,45.5) .. controls (356.67,44.3) and (357.64,43.33) .. (358.83,43.33) .. controls (360.03,43.33) and (361,44.3) .. (361,45.5) .. controls (361,46.7) and (360.03,47.67) .. (358.83,47.67) .. controls (357.64,47.67) and (356.67,46.7) .. (356.67,45.5) -- cycle ;
\draw  [fill={rgb, 255:red, 0; green, 0; blue, 0 }  ,fill opacity=1 ] (332.67,129.5) .. controls (332.67,128.3) and (333.64,127.33) .. (334.83,127.33) .. controls (336.03,127.33) and (337,128.3) .. (337,129.5) .. controls (337,130.7) and (336.03,131.67) .. (334.83,131.67) .. controls (333.64,131.67) and (332.67,130.7) .. (332.67,129.5) -- cycle ;
\draw  [fill={rgb, 255:red, 0; green, 0; blue, 0 }  ,fill opacity=1 ] (346,242.83) .. controls (346,241.64) and (346.97,240.67) .. (348.17,240.67) .. controls (349.36,240.67) and (350.33,241.64) .. (350.33,242.83) .. controls (350.33,244.03) and (349.36,245) .. (348.17,245) .. controls (346.97,245) and (346,244.03) .. (346,242.83) -- cycle ;
\draw    (447.33,39.33) .. controls (489.03,34.41) and (614.49,28.51) .. (618.25,92.69) ;
\draw [shift={(618.33,95.67)}, rotate = 270] [fill={rgb, 255:red, 0; green, 0; blue, 0 }  ][line width=0.08]  [draw opacity=0] (8.93,-4.29) -- (0,0) -- (8.93,4.29) -- cycle    ;
\draw    (447.33,261.33) .. controls (485.28,262.98) and (619.61,278.68) .. (621.65,205.9) ;
\draw [shift={(621.67,203.67)}, rotate = 89.49] [fill={rgb, 255:red, 0; green, 0; blue, 0 }  ][line width=0.08]  [draw opacity=0] (8.93,-4.29) -- (0,0) -- (8.93,4.29) -- cycle    ;
\draw    (450.67,145.33) -- (524.67,145.01) ;
\draw [shift={(527.67,145)}, rotate = 179.75] [fill={rgb, 255:red, 0; green, 0; blue, 0 }  ][line width=0.08]  [draw opacity=0] (8.93,-4.29) -- (0,0) -- (8.93,4.29) -- cycle    ;
\draw  [draw opacity=0] (547.33,101) -- (647.67,101) -- (647.67,201.67) -- (547.33,201.67) -- cycle ; \draw   (547.33,101) -- (547.33,201.67)(567.33,101) -- (567.33,201.67)(587.33,101) -- (587.33,201.67)(607.33,101) -- (607.33,201.67)(627.33,101) -- (627.33,201.67)(647.33,101) -- (647.33,201.67) ; \draw   (547.33,101) -- (647.67,101)(547.33,121) -- (647.67,121)(547.33,141) -- (647.67,141)(547.33,161) -- (647.67,161)(547.33,181) -- (647.67,181)(547.33,201) -- (647.67,201) ; \draw    ;
\draw  [fill={rgb, 255:red, 0; green, 0; blue, 0 }  ,fill opacity=1 ] (595.33,146.17) .. controls (595.33,144.97) and (596.3,144) .. (597.5,144) .. controls (598.7,144) and (599.67,144.97) .. (599.67,146.17) .. controls (599.67,147.36) and (598.7,148.33) .. (597.5,148.33) .. controls (596.3,148.33) and (595.33,147.36) .. (595.33,146.17) -- cycle ;
\draw    (247,46.33) .. controls (179.36,46.99) and (179.64,133.03) .. (244.01,135.61) ;
\draw [shift={(247,135.67)}, rotate = 180] [fill={rgb, 255:red, 0; green, 0; blue, 0 }  ][line width=0.08]  [draw opacity=0] (8.93,-4.29) -- (0,0) -- (8.93,4.29) -- cycle    ;
\draw    (245.67,145.67) .. controls (178.03,146.32) and (178.31,232.36) .. (242.68,234.94) ;
\draw [shift={(245.67,235)}, rotate = 180] [fill={rgb, 255:red, 0; green, 0; blue, 0 }  ][line width=0.08]  [draw opacity=0] (8.93,-4.29) -- (0,0) -- (8.93,4.29) -- cycle    ;
\draw    (243,21) .. controls (17.67,20.33) and (17.67,260.33) .. (241.67,261.67) ;
\draw [shift={(241.67,261.67)}, rotate = 180.34] [fill={rgb, 255:red, 0; green, 0; blue, 0 }  ][line width=0.08]  [draw opacity=0] (8.93,-4.29) -- (0,0) -- (8.93,4.29) -- cycle    ;

\draw (307.33,71.07) node [anchor=north west][inner sep=0.75pt]  {$\left(\mathcal{M}_{\mathcal{A}} ,g^{\mathcal{A}} ,\chi _{( A)}^{\mathcal{A}}\right)$};
\draw (313.33,169.73) node [anchor=north west][inner sep=0.75pt]  {$\left(\mathcal{M}_{\mathcal{B}} ,g^{\mathcal{B}} ,\chi _{( A)}^{\mathcal{B}}\right)$};
\draw (312.67,270.4) node [anchor=north west][inner sep=0.75pt]  {$\left(\mathcal{M}_{\mathcal{C}} ,g^{\mathcal{C}} ,\chi _{( A)}^{\mathcal{C}}\right)$};
\draw (355.33,27.33) node [anchor=north west][inner sep=0.75pt] [align=left] {p};
\draw (320.67,115.33) node [anchor=north west][inner sep=0.75pt]  [align=left] {q};
\draw (348.67,223.33) node [anchor=north west][inner sep=0.75pt] [align=left] {r};
\draw (568.67,71.73) node [anchor=north west][inner sep=0.75pt]  {$\mathbb{R}^{d}$};
\draw (593.33,125.73) node [anchor=north west][inner sep=0.75pt] {$\chi $};
\draw (488,128.4) node [anchor=north west][inner sep=0.75pt]  {$\chi ^{\mathcal{B}}$};
\draw (537.33,243.73) node [anchor=north west][inner sep=0.75pt]  {$\chi ^{\mathcal{C}}$};
\draw (535.33,45.07) node [anchor=north west][inner sep=0.75pt]  {$\chi ^{\mathcal{A}}$};
\draw (203.33,84.73) node [anchor=north west][inner sep=0.75pt]   {$C_{\chi }^{\mathcal{AB}}$};
\draw (200,183.4) node [anchor=north west][inner sep=0.75pt]  {$C_{\chi }^{\mathcal{BC}}$};
\draw (75.33,132.73) node [anchor=north west][inner sep=0.75pt]  {$C_{\chi }^{\mathcal{AC}} =C_{\chi }^{\mathcal{BC}} \circ C_{\chi }^{\mathcal{AB}}$};
\end{tikzpicture}
    \caption{\justifying A strategy to identify points $p\in\mathcal{M_A}$, $q \in \mathcal{M_B}$, and  $r\in\mathcal{M_C}$ across a superposition of three spacetimes. Following  \cite{Kabel2024}, we identify those points at which a chosen set of reference fields $\chi$ takes on the same values, that is, points for which $\chi^\mathcal{A}(p) = \chi^\mathcal{B}(q) = \chi^\mathcal{C}(r) = \chi \in \mathbb{R}^d$.}
    \label{fig:Comparison map}
\end{figure*}

Next, let us define the causal order between two events in a given spacetime. To this end, consider now three timelike future-directed worldlines $\gamma_0, \gamma_a, a = 1,2$, where $\gamma_0$ can be seen as the worldline of some test particle and intersects $\gamma_1$ and $\gamma_2$ respectively exactly once. We denote the event defined by the crossing of $\gamma_0$ with $\gamma_a$ by $\mathcal{E}_a$. We parameterise $\gamma_0$ by the proper time of the test particle and define the proper time $\tau_a$ associated to the event $\mathcal{E}_a$ by
\begin{equation}
    \tau_a = \int_{p_{\mathrm{init}}}^{p_{\mathcal{E}_a}} \frac{ds}{c} = \int_{p_{\mathrm{init}}}^{p_{\mathcal{E}_a}} \frac{\sqrt{-g_{\mu\nu}dx^\mu dx^\nu}}{c},
\end{equation}
where $p_{\mathrm{init}}$ denotes an arbitrary initial point of the test particle's worldline.

For two distinct points $p_{\mathcal{E}_1}$ and $p_{\mathcal{E}_2}$ we have that $\tau_1 \neq \tau_2$. Let
\begin{equation}
    \Delta \tau = \tau_2 - \tau_1 = s\abs{\tau_2 - \tau_1}
\end{equation}
so that 
\begin{equation}
    s = \sgn(\tau_2 - \tau_1)
\end{equation}
is the causal order of the events, with $s=1$ if $\mathcal{E}_2$ lies in the future of $\mathcal{E}_1$ and $s=-1$ if it lies in the past.\\

Given $N$ timelike future-oriented worldlines $\gamma_a$, $a=1,\dots,N$ that each cross $\gamma_0$ exactly once, defining $N$ events $\mathcal{E}_a$, this notion of causal order straightforwardly extends to all pairs of any finite number of events.

\begin{definition}[Causal order between two events]
    The \textit{causal order} between events $\mathcal{E}_a$ and $\mathcal{E}_b$ of a spacetime $(\mathcal{M},g)$ is 
    \begin{equation}
        s^\mathcal{M}_{ab} := \sgn(\tau_b - \tau_a)
    \end{equation}
    where $\tau_a$ and $\tau_b$ are the proper times corresponding to the crossing points $\gamma_0(\tau_a)=p_{\mathcal{E}_a}$ and $\gamma_0(\tau_b)=p_{\mathcal{E}_b}$, respectively. \label{def:causalorder}
\end{definition}

\begin{proposition}
    \begin{equation}
        s^\mathcal{M}_{ab} = -s^\mathcal{M}_{ba} = s^\mathcal{M}_{[ba]}
    \end{equation}
    where brackets denote antisymmetrisation.
\end{proposition}

\begin{proof}
    \begin{equation}
        s^\mathcal{M}_{ab} = \sgn(\tau_b - \tau_a) = - \sgn(\tau_a - \tau_b) = - s^\mathcal{M}_{ba}
    \end{equation}
\end{proof}

This can be extended straightforwardly to an arbitrary finite number $M$ of spacetimes in superposition. We denote these events by $\mathcal{E}_a^\mathcal{X}$ where $\mathcal{X} = \mathcal{A},\mathcal{B}, \dots$ denotes the different branches of the superposition, while $a = 1,\dots,N$ labels the different events in each of the branches respectively. Note that for given $a$, all instances $\mathcal{E}_a^{\mathcal{A}}, \mathcal{E}_a^{\mathcal{B}}, \dots $ refer to the \emph{same} physical event -- that is, the crossing of the worldlines $\gamma_0$ with $\gamma_a$ -- in different branches of the superposition. Given the causal order between any two events in a given spacetime, we can arrange the events in an ordered set.

\begin{definition}[Ordered collection of events]
    The \textit{ordered collection of events of a spacetime} $(\mathcal{M}_{\mathcal{X}},g_\mathcal{X})$ is the totally ordered (countable) set 
    \begin{equation}
        \mathcal{S}_\mathcal{X} := \{\mathcal{E}_1^\mathcal{X} \prec \mathcal{E}_2^\mathcal{X} \prec ...\}\label{eq: orderedCollection}
    \end{equation}
    where the total order $\prec$ is given by $\mathcal{E}_a^\mathcal{X} \prec \mathcal{E}_b^\mathcal{X}$ if $s_{ab}^{\mathcal{X}} = 1$. 
    \label{def:orderedcollevents}
\end{definition}

The information of causal orderings is thus entirely encapsulated in the collection of all ordered collections of events over all spacetimes, that is, in $\mathcal{S}:= \{\mathcal{S}_\mathcal{X}\}_\mathcal{X}$.
In general, $S_\mathcal{X}$ can differ across different spacetimes $(\mathcal{M}_\mathcal{X}, g^\mathcal{X})$ -- although the number of events contained in this set will always be the same, their ordering can vary. This will, in general, give rise to an \emph{indefinite} causal order. In the next section, we look at different ways of quantifying this indefiniteness of causal order. {Let us note that the following discussion, while motivated by the understanding of events as crossings of worldlines in a superposition of quasiclassical spacetimes, also applies more generally to processes with quantum-controlled causal order, which give rise to a superposition of ordered collections of events $\mathcal{S}_\mathcal{X}$. In this case, the index $\mathcal{X}$ could refer to any other (ancillary) degree of freedom that distinguishes the branches of the superposition, such as the path \cite{procopio_experimental_2015,Rubino2017,rubino_experimental_2022} or polarization \cite{Goswami2018} of a photon.}

\section{Causal order quantifiers}
\label{sec:quantifiers}

\subsection{$M=2$ spacetimes, $N<\infty$ events}

For now, we consider a superposition of two spacetimes $(\mathcal{M}_\mathcal{A},g_{\mathcal{A}})$ and $(\mathcal{M}_\mathcal{B},g_{\mathcal{B}})$, i.e.
\begin{equation}
    \ket{\Psi} = c_1 \ket{\Tilde{g}^\mathcal{A}} + c_2 \ket{\Tilde{g}^\mathcal{B}}
\end{equation}
where $c_1,c_2 \in \mathbb{C}, \abs{c_1}^2 + \abs{c_2}^2 = 1$. We start with the case of finitely many events, the same number $N$ in both spacetimes.

\begin{definition}[Pairwise causal order]
   For two spacetimes $(\mathcal{M}_\mathcal{A},g_{\mathcal{A}})$ and $(\mathcal{M}_\mathcal{B},g_{\mathcal{B}})$, we define the \textit{pairwise causal order} between two events $\mathcal{E}_a$ and $\mathcal{E}_b$ to be
\begin{equation}
    \tsadi^{\mathcal{A}\mathcal{B}}_{a b} := s^{\mathcal{A}}_{a  b} s^{\mathcal{B}}_{a b}
\end{equation}
where the identification of events across the different spacetimes has been made through the construction of Sec.~\ref{sec: events causal order}. We say that the pairwise causal order between two events $\mathcal{E}_a$ and $\mathcal{E}_b$ is \textit{definite} if $\tsadi^{\mathcal{A}\mathcal{B}}_{a b} = 1$. In that case, the order of $\mathcal{E}_a$ and $\mathcal{E}_b$ is the same in both spacetimes. We say that the pairwise causal order is \textit{indefinite} if $\tsadi^{\mathcal{A}\mathcal{B}}_{a b} = -1$. In that case, the order of $\mathcal{E}_a$ and $\mathcal{E}_b$ is opposite across both spacetimes.
\end{definition}

\begin{proposition}
    \begin{equation}
    \tsadi^{\mathcal{A}\mathcal{B}}_{a b} = \tsadi^{\mathcal{B}\mathcal{A}}_{a b} = \tsadi^{(\mathcal{A}\mathcal{B})}_{a b}
    \end{equation}
    where parentheses denote symmetrisation.
\end{proposition}

\begin{proof}
    The causal order between two events is a real number and thus commutes under multiplication.
\end{proof}

\begin{proposition}
\begin{equation}
    \tsadi^{\mathcal{A}\mathcal{B}}_{ab} = \tsadi^{\mathcal{A}\mathcal{B}}_{ba} = \tsadi^{\mathcal{A}\mathcal{B}}_{(ab)}
\end{equation}
\end{proposition}
\begin{proof}
    \begin{equation}
        \tsadi^{\mathcal{A}\mathcal{B}}_{a b} = s^{\mathcal{A}}_{a  b} s^{\mathcal{B}}_{a b} = - s^{\mathcal{A}}_{b  a} s^{\mathcal{B}}_{a b} = + s^{\mathcal{A}}_{b  a} s^{\mathcal{B}}_{b a} = \tsadi^{\mathcal{A}\mathcal{B}}_{b a}
    \end{equation}
\end{proof}

We can naturally generalise the pairwise causal order between two events to the case of an arbitrary number $N$ of events as
\begin{align}
\begin{cases}
    \tsadi^{\mathcal{A}\mathcal{B}}_{12} \equiv s^{\mathcal{A}}_{1  2} s^{\mathcal{B}}_{12} &= \pm 1 \\
    \tsadi^{\mathcal{A}\mathcal{B}}_{23} \equiv s^{\mathcal{A}}_{23} s^{\mathcal{B}}_{23} &= \pm 1 \\
    \tsadi^{\mathcal{A}\mathcal{B}}_{13} \equiv s^{\mathcal{A}}_{13} s^{\mathcal{B}}_{13}&= \pm 1 \\
    &\;\;\vdots \nonumber \\
    \tsadi^{\mathcal{A}\mathcal{B}}_{N-1 N} \equiv s^{\mathcal{A}}_{N-1 N} s^{\mathcal{B}}_{N-1 N} &= \pm 1
\end{cases}
\end{align}

\begin{definition}[Longitudinal causal order]
    The \textit{longitudinal causal order} between two spacetimes $(\mathcal{M}_\mathcal{A},g_{\mathcal{A}})$ and $(\mathcal{M}_\mathcal{B},g_{\mathcal{B}})$ is 
\begin{equation}
    \ayin^{\mathcal{A}\mathcal{B}} := \sum_{1 \leq i < j}^N \tsadi^{\mathcal{A}\mathcal{B}}_{ij}
\end{equation}
with, once again,
\begin{equation}
    \ayin^{\mathcal{A}\mathcal{B}} = \ayin^{\mathcal{B}\mathcal{A}} = \ayin^{(\mathcal{A}\mathcal{B})}.
\end{equation}
\end{definition}

A visual representation of pairwise and longitudinal causal orders as measures of indefinite causal order is given in Fig.~\ref{fig:Causal orders}.

\begin{figure}[b!]
    \centering
    \includegraphics[scale=0.18]{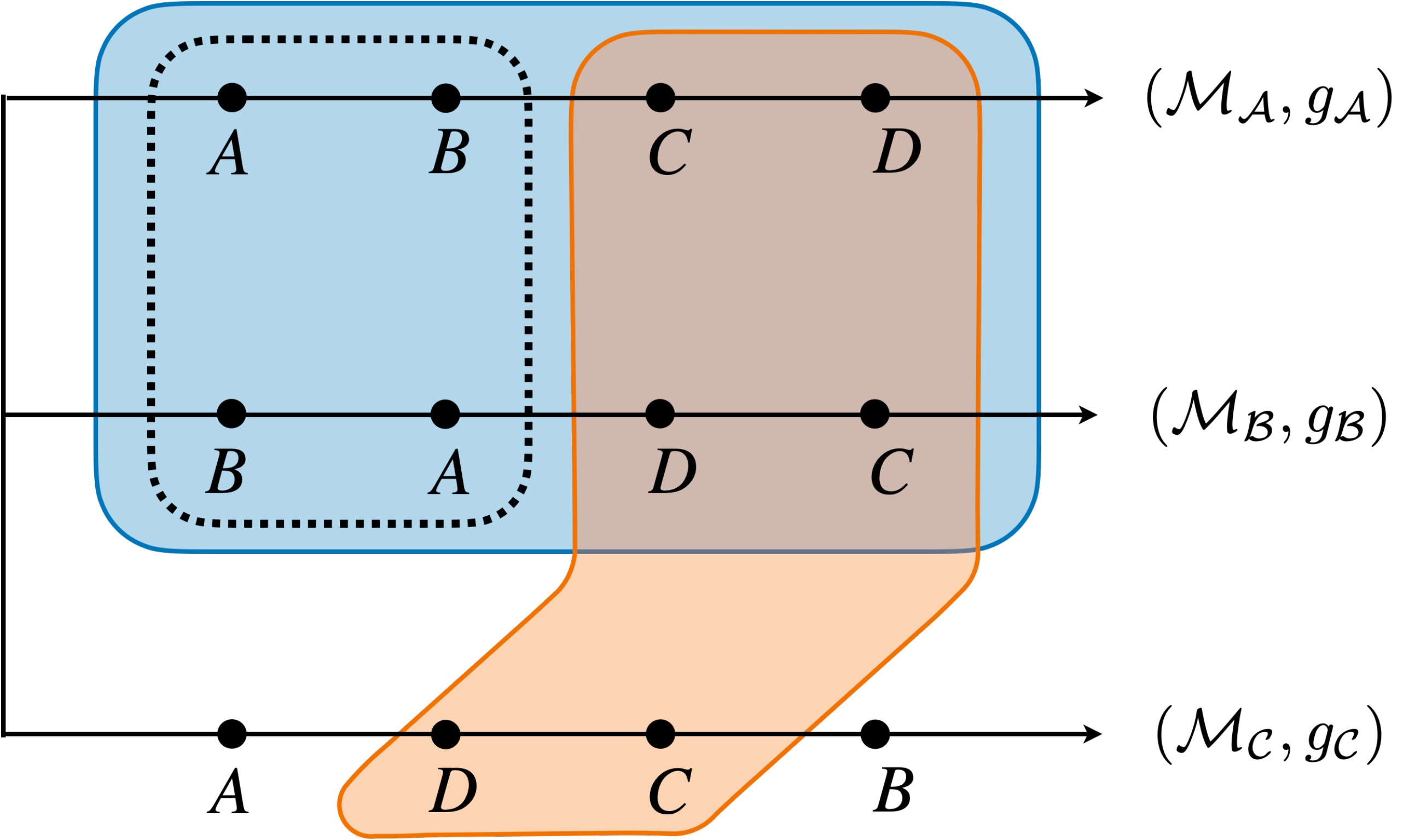}
    \caption{\justifying Measures of indefinite causal order for $N=4$ events across a superposition of $M=3$ spacetimes. The dashed black line describes \emph{pairwise causal order} $\sfrak_{12}^{\mathcal{AB}}=-1$; the blue rectangle captures the \emph{longitudinal causal order} $\lfrak^{\mathcal{A}\mathcal{B}}=2$, as well the causal indefiniteness $\delta(\mathcal{A},\mathcal{B})=2$; the orange shape represents the \emph{transverse causal order} $\tfrak_{34}=-1$. The \emph{total causal order} is $\stotfrak = 2$ and the total causal indefiniteness is $\Delta = 8$.}
    \label{fig:Causal orders}
\end{figure}

\begin{proposition}
    For two spacetimes $(\mathcal{M}_\mathcal{A},g_{\mathcal{A}})$ and $(\mathcal{M}_\mathcal{B},g_{\mathcal{B}})$ with $N$ events, 
    \begin{equation}
        -\begin{pmatrix} N \\ 2 \end{pmatrix}\leq\ayin^{\mathcal{A}\mathcal{B}}\leq \begin{pmatrix} N \\ 2 \end{pmatrix}
    \end{equation}
\end{proposition}

\begin{proof}
    For the case in which the pairwise causal order is definite for each pair of events, which is the case that maximises the longitudinal causal order, we have that $\forall 1\leq i<j\leq N, \tsadi^{\mathcal{A}\mathcal{B}}_{ij} = +1$, that is, 
    \begin{equation}
        \max_{\tsadi} \ayin^{\mathcal{A}\mathcal{B}} = \sum_{1 \leq i < j}^N 1 = \begin{pmatrix} N \\ 2 \end{pmatrix}.
    \end{equation}
     On the other hand, the longitudinal causal order is minimised when the pairwise causal order is indefinite for every event $ab...mn \to nm...ba$, i.e.~$\forall 1\leq i<j\leq N, \tsadi^{\mathcal{A}\mathcal{B}}_{ij} = -1$ and the result follows.
\end{proof}

\begin{definition}[Indefinite longitudinal causal order]
    The longitudinal causal order between $N$ events of a superposition of two spacetimes $(\mathcal{M}_\mathcal{A},g_{\mathcal{A}})$ and $(\mathcal{M}_\mathcal{B},g_{\mathcal{B}})$ is \textit{definite} if and only if $\ayin^{\mathcal{A}\mathcal{B}} = \begin{pmatrix} N \\ 2 \end{pmatrix}$. Otherwise we say that it is \textit{indefinite}.
\end{definition}

\begin{definition}[Maximally indefinite longitudinal causal order]
    The longitudinal causal order between $N$ events of a superposition of two spacetimes $(\mathcal{M}_\mathcal{A},g_{\mathcal{A}})$ and $(\mathcal{M}_\mathcal{B},g_{\mathcal{B}})$ is \textit{maximally indefinite} if and only if $\ayin^{\mathcal{A}\mathcal{B}} = -\begin{pmatrix} N \\ 2 \end{pmatrix}$. If it is indefinite but not maximally indefinite we say that it is \textit{braided}.
\end{definition}

This naming already reflects knot theoretic properties of the causal order that will be explored below. A natural quantifier of the indefiniteness of the causal ordering is defined as follows.

\begin{definition}[Causal indefiniteness]
    The \textit{causal indefiniteness} of a superposition of two spacetimes $(\mathcal{M}_\mathcal{A},g_{\mathcal{A}})$ and $(\mathcal{M}_\mathcal{B},g_{\mathcal{B}})$ with $N$ events is
        \begin{equation}
            \delta(\mathcal{A},\mathcal{B}) := \sum_{1 \leq i < j}^N {|s_{ij}^{\mathcal{A}} - s_{ij}^{\mathcal{B}}|/2.}
        \end{equation}
\end{definition}
Indeed, this indicates how many pairs of events have opposite causal order. As we shall see, the causal indefiniteness will also be related to a knot invariant.
\begin{proposition}
    \label{Causal indefiniteness and pairwise causal order}
    \begin{equation}
    \delta(\mathcal{A},\mathcal{B}) = \sum_{1 \leq i < j}^N \frac{1-\tsadi^{\mathcal{A}\mathcal{B}}_{ij}}{2} = \frac{1}{2} \Bigg[\begin{pmatrix}
        N \\ 2
    \end{pmatrix} - \ayin^{\mathcal{A}\mathcal{B}} \Bigg]
    \end{equation}
\end{proposition}
\begin{proof}
    We have that $1-\tsadi^{\mathcal{A}\mathcal{B}}_{ij}=0$ for definite pairwise causal order and $1-\tsadi^{\mathcal{A}\mathcal{B}}_{ij}=2$ for indefinite pairwise causal order. Thus, $\abs{s_{ij}^{\mathcal{A}} - s_{ij}^{\mathcal{B}}}^0=\frac{1}{2}(1-\tsadi^{\mathcal{A}\mathcal{B}}_{ij})$. Summing over all pairs of events, the result follows.
\end{proof}

\begin{lemma}
    \label{Causal indefiniteness maximally indefinite}
    The longitudinal causal order is definite for $\delta(\mathcal{A},\mathcal{B}) = 0$ and maximally indefinite for $\delta(\mathcal{A},\mathcal{B}) = N(N-1)/2$.
\end{lemma}

\begin{lemma}
    \label{Additivity of indefiniteness}
    The causal indefiniteness is additive under the concatenation of sequences.
\end{lemma}

\begin{proof}
    If we have two ordered collections of $N$ events, which can both be split into $n$ subsequences of respective lengths $l_1,...,l_{n-1},l_n = N-(l_1,...,l_{n-1})$ -- that is into $n$ ordered subsets containing $l_i, i =1,\dots, n$ elements each -- we have
    \begin{align}
        \delta_{(N)}(\mathcal{A},\mathcal{B}) &= \sum_{1 \leq i < j}^N \abs{s_{ij}^{\mathcal{A}} - s_{ij}^{\mathcal{B}}}^0 \\
        &= \sum_{1\leq i < j}^{l_1} \abs{s_{ij}^{\mathcal{A}} - s_{ij}^{\mathcal{B}}}^0  \underbrace{ + 0 + ...}_{s_{im}^\mathcal{A} = s_{im}^\mathcal{B} \forall m \in \{l_1+1,...,N\}} \nonumber \\ &+ \sum_{l_1 \leq i < j}^{l_2} \abs{s_{ij}^{\mathcal{A}} - s_{ij}^{\mathcal{B}}}^0 + 0 + ... \nonumber \\ &+ \sum_{l_{n-1} \leq i < j}^N \abs{s_{ij}^{\mathcal{A}} - s_{ij}^{\mathcal{B}}}^0 \\
        &= \delta_{(l_1)}(\mathcal{A},\mathcal{B}) + ... + \delta_{(l_n)}(\mathcal{A},\mathcal{B})
    \end{align}
    which concludes the proof.
\end{proof}

\subsection{$M<\infty$ spacetimes, $N<\infty$ events}

We now consider a superposition of finitely many spacetimes $\{(\mathcal{M}_\mathcal{X},g_{\mathcal{X}})\}_{\mathcal{X}=1}^{M}$, i.e.
\begin{equation}
    \ket{\Psi} = \sum_{\mathcal{X}=1}^M c_\mathcal{X} \ket{\Tilde{g}^\mathcal{X}}
\end{equation}
where $\forall \mathcal{X}, c_\mathcal{X} \in \mathbb{C}, \sum_{\mathcal{X}=1}^M\abs{c_\mathcal{X}}^2 = 1$. We continue with the case of finitely many events, the same number $N$ in all spacetimes (see Fig.~\ref{fig:indefinite causal order superposition of spacetimes} for an example).

\begin{figure}[b!]
    \centering
    \begin{tikzpicture}[x=0.75pt,y=0.75pt,yscale=-0.75,xscale=0.75]
\draw  [fill={rgb, 255:red, 155; green, 155; blue, 155 }  ,fill opacity=0.14 ] (288.02,10.33) -- (455.67,10.33) -- (416.98,94.33) -- (249.33,94.33) -- cycle ;
\draw  [fill={rgb, 255:red, 155; green, 155; blue, 155 }  ,fill opacity=0.19 ] (288.69,109.67) -- (456.33,109.67) -- (417.65,193.67) -- (250,193.67) -- cycle ;
\draw  [fill={rgb, 255:red, 155; green, 155; blue, 155 }  ,fill opacity=0.13 ] (288.69,209) -- (456.33,209) -- (417.65,293) -- (250,293) -- cycle ;
\draw [color={rgb, 255:red, 208; green, 2; blue, 27 }  ,draw opacity=1 ]   (285,81.67) .. controls (308.33,49.67) and (323,59.67) .. (344.33,18.33) ;
\draw    (485,293.67) -- (485,229.5) ;
\draw [shift={(485,227.5)}, rotate = 89.42] [color={rgb, 255:red, 0; green, 0; blue, 0 }  ][line width=0.75]    (10.93,-3.29) .. controls (6.95,-1.4) and (3.31,-0.3) .. (0,0) .. controls (3.31,0.3) and (6.95,1.4) .. (10.93,3.29)   ;
\draw [color={rgb, 255:red, 65; green, 117; blue, 5 }  ,draw opacity=1 ]   (323.67,83) .. controls (347,51) and (361.67,61) .. (383,19.67) ;
\draw [color={rgb, 255:red, 74; green, 144; blue, 226 }  ,draw opacity=1 ]   (362.33,85) .. controls (385.67,53) and (400.33,63) .. (421.67,21.67) ;
\draw [color={rgb, 255:red, 74; green, 144; blue, 226 }  ,draw opacity=1 ]   (366.33,178.33) .. controls (281.67,115) and (369.67,123.67) .. (382.33,121) .. controls (395,118.33) and (398.33,119.67) .. (411.67,113.67) ;
\draw [color={rgb, 255:red, 208; green, 2; blue, 27 }  ,draw opacity=1 ]   (303.67,179.67) .. controls (341.67,153) and (377,159.67) .. (343,114.33) ;
\draw [color={rgb, 255:red, 65; green, 117; blue, 5 }  ,draw opacity=1 ]   (339,178.33) .. controls (409.67,151) and (393.67,138.33) .. (375.67,115) ;
\draw [color={rgb, 255:red, 65; green, 117; blue, 5 }  ,draw opacity=1 ]   (338.33,281.67) .. controls (291,194.33) and (357,227.67) .. (375,218.33) ;
\draw [color={rgb, 255:red, 208; green, 2; blue, 27 }  ,draw opacity=1 ]   (305.67,278.33) .. controls (366.33,247.67) and (432.33,267) .. (345,213) ;
\draw [color={rgb, 255:red, 74; green, 144; blue, 226 }  ,draw opacity=1 ]   (383,283.67) .. controls (298.33,220.33) and (386.33,229) .. (399,226.33) .. controls (411.67,223.67) and (415,225) .. (428.33,219) ;
\draw    (318.2,20.6) .. controls (328.2,31.8) and (348.6,41.4) .. (351.4,47.8) .. controls (354.2,54.2) and (335,59.8) .. (358.2,67) .. controls (381.4,74.2) and (363.4,69) .. (389.8,76.6) ;
\draw [shift={(335.06,34.32)}, rotate = 0] [color={rgb, 255:red, 0; green, 0; blue, 0 }  ][fill={rgb, 255:red, 0; green, 0; blue, 0 }  ][line width=0.75]      (0, 0) circle [x radius= 1.34, y radius= 1.34]   ;
\draw [shift={(346.74,60.21)}, rotate = 0] [color={rgb, 255:red, 0; green, 0; blue, 0 }  ][fill={rgb, 255:red, 0; green, 0; blue, 0 }  ][line width=0.75]      (0, 0) circle [x radius= 1.34, y radius= 1.34]   ;
\draw [shift={(373.95,72.01)}, rotate = 0] [color={rgb, 255:red, 0; green, 0; blue, 0 }  ][fill={rgb, 255:red, 0; green, 0; blue, 0 }  ][line width=0.75]      (0, 0) circle [x radius= 1.34, y radius= 1.34]   ;
\draw    (320.75,117.13) .. controls (326.5,122.88) and (340.75,133.63) .. (347,140.63) .. controls (353.25,147.63) and (359.5,156.88) .. (361.5,158.88) .. controls (363.5,160.88) and (366.25,163.88) .. (377.5,171.88) ;
\draw [shift={(334.14,128.79)}, rotate = 0] [color={rgb, 255:red, 0; green, 0; blue, 0 }  ][fill={rgb, 255:red, 0; green, 0; blue, 0 }  ][line width=0.75]      (0, 0) circle [x radius= 1.34, y radius= 1.34]   ;
\draw [shift={(354.39,149.71)}, rotate = 0] [color={rgb, 255:red, 0; green, 0; blue, 0 }  ][fill={rgb, 255:red, 0; green, 0; blue, 0 }  ][line width=0.75]      (0, 0) circle [x radius= 1.34, y radius= 1.34]   ;
\draw [shift={(369.26,165.76)}, rotate = 0] [color={rgb, 255:red, 0; green, 0; blue, 0 }  ][fill={rgb, 255:red, 0; green, 0; blue, 0 }  ][line width=0.75]      (0, 0) circle [x radius= 1.34, y radius= 1.34]   ;
\draw    (327.25,214.63) .. controls (329.75,221.88) and (329.5,222.63) .. (333.75,228.13) .. controls (338,233.63) and (351.5,252.88) .. (363.25,253.88) .. controls (375,254.88) and (373.5,253.38) .. (385,266.88) ;
\draw [shift={(329.73,221.76)}, rotate = 0] [color={rgb, 255:red, 0; green, 0; blue, 0 }  ][fill={rgb, 255:red, 0; green, 0; blue, 0 }  ][line width=0.75]      (0, 0) circle [x radius= 1.34, y radius= 1.34]   ;
\draw [shift={(346.66,243.69)}, rotate = 0] [color={rgb, 255:red, 0; green, 0; blue, 0 }  ][fill={rgb, 255:red, 0; green, 0; blue, 0 }  ][line width=0.75]      (0, 0) circle [x radius= 1.34, y radius= 1.34]   ;
\draw [shift={(376.08,257.02)}, rotate = 0] [color={rgb, 255:red, 0; green, 0; blue, 0 }  ][fill={rgb, 255:red, 0; green, 0; blue, 0 }  ][line width=0.75]      (0, 0) circle [x radius= 1.34, y radius= 1.34]   ;
\draw    (250.6,68.6) .. controls (199.77,69.78) and (194.74,130.73) .. (250.41,129.12) ;
\draw [shift={(253,129)}, rotate = 176.47] [fill={rgb, 255:red, 0; green, 0; blue, 0 }  ][line width=0.08]  [draw opacity=0] (8.93,-4.29) -- (0,0) -- (8.93,4.29) -- cycle    ;
\draw    (251.8,179) .. controls (200.97,180.18) and (195.94,241.13) .. (251.61,239.52) ;
\draw [shift={(254.2,239.4)}, rotate = 176.47] [fill={rgb, 255:red, 0; green, 0; blue, 0 }  ][line width=0.08]  [draw opacity=0] (8.93,-4.29) -- (0,0) -- (8.93,4.29) -- cycle    ;
\draw    (251.4,33.4) .. controls (145.27,33.4) and (162.64,279.21) .. (234.8,281.8) ;
\draw [shift={(237,281.8)}, rotate = 178.14] [fill={rgb, 255:red, 0; green, 0; blue, 0 }  ][line width=0.08]  [draw opacity=0] (8.93,-4.29) -- (0,0) -- (8.93,4.29) -- cycle    ;
\draw    (485,196) -- (485,131.83) ;
\draw [shift={(485,129.83)}, rotate = 89.42] [color={rgb, 255:red, 0; green, 0; blue, 0 }  ][line width=0.75]    (10.93,-3.29) .. controls (6.95,-1.4) and (3.31,-0.3) .. (0,0) .. controls (3.31,0.3) and (6.95,1.4) .. (10.93,3.29)   ;
\draw    (485,94) -- (485,29.83) ;
\draw [shift={(485,27.83)}, rotate = 89.42] [color={rgb, 255:red, 0; green, 0; blue, 0 }  ][line width=0.75]    (10.93,-3.29) .. controls (6.95,-1.4) and (3.31,-0.3) .. (0,0) .. controls (3.31,0.3) and (6.95,1.4) .. (10.93,3.29)   ;

\draw (482,215) node [anchor=north west][inner sep=0.75pt]  [xscale=0.5,yscale=0.5]  {$\tau $};
\draw (421.33,80.07) node [anchor=north west][inner sep=0.75pt]  [font=\footnotesize,xscale=0.5,yscale=0.5]  {$\left(\mathcal{M}_{\mathcal{A}} ,g^{\mathcal{A}} ,s^{\mathcal{A}}\right)$};
\draw (423.33,176.73) node [anchor=north west][inner sep=0.75pt]  [font=\footnotesize,xscale=0.5,yscale=0.5]  {$\left(\mathcal{M}_{\mathcal{B}} ,g^{\mathcal{B}} ,s^{\mathcal{B}}\right)$};
\draw (428,267.4) node [anchor=north west][inner sep=0.75pt]  [font=\footnotesize,xscale=0.5,yscale=0.5]  {$\left(\mathcal{M}_{\mathcal{C}} ,g^{\mathcal{C}} ,s^{\mathcal{C}}\right)$};
\draw (370,62) node [anchor=north west][inner sep=0.75pt]  [xscale=0.5,yscale=0.5] [align=left] {{\tiny A}};
\draw (330,120) node [anchor=north west][inner sep=0.75pt]  [xscale=0.5,yscale=0.5] [align=left] {{\tiny A}};
\draw (350,240) node [anchor=north west][inner sep=0.75pt]  [xscale=0.5,yscale=0.5] [align=left] {{\tiny A}};
\draw (340,50) node [anchor=north west][inner sep=0.75pt]  [xscale=0.5,yscale=0.5] [align=left] {{\tiny B}};
\draw (330,23) node [anchor=north west][inner sep=0.75pt]  [xscale=0.5,yscale=0.5] [align=left] {{\tiny C}};
\draw (356.25,145) node [anchor=north west][inner sep=0.75pt]  [xscale=0.5,yscale=0.5] [align=left] {{\tiny C}};
\draw (373,247) node [anchor=north west][inner sep=0.75pt]  [xscale=0.5,yscale=0.5] [align=left] {{\tiny C}};
\draw (367,155) node [anchor=north west][inner sep=0.75pt]  [xscale=0.5,yscale=0.5] [align=left] {{\tiny B}};
\draw (319.75,215) node [anchor=north west][inner sep=0.75pt]  [xscale=0.5,yscale=0.5] [align=left] {{\tiny B}};
\draw (214.67,92.07) node [anchor=north west][inner sep=0.75pt]  [font=\small,xscale=0.5,yscale=0.5]  {$C^{\mathcal{AB}}$};
\draw (216.67,204.07) node [anchor=north west][inner sep=0.75pt]  [font=\small,xscale=0.5,yscale=0.5]  {$C^{\mathcal{BC}}$};
\draw (179.33,146.73) node [anchor=north west][inner sep=0.75pt]  [font=\small,xscale=0.5,yscale=0.5]  {$C^{\mathcal{AC}}$};
\draw (482,117) node [anchor=north west][inner sep=0.75pt]  [xscale=0.5,yscale=0.5]  {$\tau $};
\draw (482,15) node [anchor=north west][inner sep=0.75pt]  [xscale=0.5,yscale=0.5]  {$\tau $};
\end{tikzpicture}
    \caption{\justifying Indefinite causal order of three events in a superposition of $M=3$ spacetimes : $ABC-BCA-CAB$.}
    \label{fig:indefinite causal order superposition of spacetimes}
\end{figure}

\begin{definition}[Transverse causal order]
    The \textit{transverse causal order} between two events $\mathcal{E}_a$ and $\mathcal{E}_b$ in a collection of $M$ spacetimes $\{(\mathcal{M}_\mathcal{X},g_{\mathcal{X}})\}_{\mathcal{X}=1}^M$ is defined as
    \begin{equation}
        \lamed_{ab} := \sum_{1 \leq \mathcal{X} < \mathcal{Y}}^M \tsadi_{ab}^{\mathcal{X}\mathcal{Y}}
    \end{equation}
    with $\lamed_{ab} = \lamed_{ba} = \lamed_{(ab)}$.
\end{definition}

An intuitive depiction of transverse causal order as a measure of indefinite causal order is given in Fig.~\ref{fig:Causal orders}.

\begin{proposition}
    \label{Lamed bounds}
    For $M$ spacetimes,
    \begin{equation}
        \begin{pmatrix}
            M \\ 2
        \end{pmatrix} - 2\Big\lfloor \frac{M}{2} \Big\rfloor \Big\lceil \frac{M}{2} \Big\rceil \leq\lamed_{ab}\leq \begin{pmatrix}
            M \\ 2
        \end{pmatrix}
    \end{equation}
    where $\lfloor \cdot \rfloor$ and $\lceil \cdot \rceil$ denote the floor and ceiling functions, respectively.
\end{proposition}

\begin{proof}
    If the pairwise causal order between events $\mathcal{E}_a$ and $\mathcal{E}_b$ is \sam{the same in all spacetimes}, which is the case that maximises the transverse causal order, we have that $\forall 1\leq \mathcal{X} < \mathcal{Y} \leq M, \tsadi^{\mathcal{X}\mathcal{Y}}_{ij} = +1$, that is, 
    \begin{equation}
        \max_{\tsadi} \lamed_{ab} = \sum_{1 \leq \mathcal{X} < \mathcal{Y}}^M 1 = \binom{M}{2}.
    \end{equation}
    The case that minimises the transverse causal order is the one in which the pairwise causal order is \sam{different} between  \sam{any} pair of spacetimes. Indeed, any other combination would end up giving a total positive contribution due to the overlap. For such an alternating case, we proceed by induction.\\
  
    \noindent \emph{Induction hypothesis:}\\
    \begin{equation*}
    \min_\tsadi \lamed_{ab}^{(M)} = \binom{M}{2} - 2\Big\lfloor \frac{M}{2} \Big\rfloor \Big\lceil \frac{M}{2} \Big\rceil \text{ for } M \geq 2\end{equation*}
    
    \noindent \emph{Base case ($M=2$):}
        \begin{equation*}
            \min_\tsadi \lamed_{ab}^{(2)} = -1 = \binom{2}{2} - 2\cdot1\cdot1.
    \end{equation*}

    \noindent \emph{Inductive step ($M\to M+1)$: }\\
       
    Suppose the induction hypothesis holds, where we alternate the order of each ordered collection of events between spacetimes to minimise $\lamed_{ab}^{(M)}$. Then, adding another ordered collection of events, whose order alternates with the last, increases the transverse causal order by $\Big\lfloor \frac{M}{2} \Big\rfloor$ and decreases it by $\Big\lceil \frac{M}{2} \Big\rceil$. Thus, 
    \begin{equation*}
    \min_\tsadi \lamed_{ab}^{(M+1)} = \binom{M}{2} - 2\Big\lfloor \frac{M}{2} \Big\rfloor \Big\lceil \frac{M}{2} \Big\rceil + \Big\lfloor \frac{M}{2} \Big\rfloor - \Big\lceil \frac{M}{2} \Big\rceil.
    \end{equation*}
    Next, we want to show that this is equal to
    \begin{equation*}
        f^{(M+1)} \equiv \binom{M+1}{2} - 2\Big\lfloor \frac{M+1}{2} \Big\rfloor \Big\lceil \frac{M+1}{2} \Big\rceil.
    \end{equation*}
    We do so by case analysis and using $\binom{M}{2} = (M^2-M)/2$ as well as $\binom{M+1}{2} = (M^2+M)/2$. For $M$ \emph{even}, we have
    \begin{align*}
    \min_\tsadi \lamed_{ab}^{(M+1)} &= \binom{M}{2} - \frac{M^2}{2} = -\frac{M}{2},\\
    f^{(M+1)} &= \binom{M+1}{2} - \frac{M(M+2)}{2} = -\frac{M}{2}.
    \end{align*}
    Similarly, for $M$ \emph{odd}, we have
    \begin{align*}
    \min_\tsadi \lamed_{ab}^{(M+1)} &= \binom{M}{2} - \frac{(M-1)(M+1)}{2}-1 = -\frac{(M+1)}{2},\\
    f^{(M+1)} &= \binom{M+1}{2} - \frac{(M+1)^2}{2} = -\frac{(M+1)}{2}.
    \end{align*}
    This proves the inductive step for any $M > 2$.
    Thus, by induction, we have that $\forall M\geq2$
    \begin{equation}
        \min \lamed_{ab}^{(M)} = \begin{pmatrix}
            M \\ 2
        \end{pmatrix} - 2 \Big\lfloor \frac{M}{2} \Big\rfloor \Big\lceil \frac{M}{2} \Big\rceil,
    \end{equation}
    which provides the lower bound of $\lamed_{ab}$.
\end{proof}

\begin{definition}[Indefinite transverse causal order]
    Transverse causal order between $N$ events of a superposition of two spacetimes $(\mathcal{M}_\mathcal{A},g_{\mathcal{A}})$ and $(\mathcal{M}_\mathcal{B},g_{\mathcal{B}})$ is \textit{definite} if and only if $\lamed_{ab} = \begin{pmatrix}
            M \\ 2
        \end{pmatrix}$. Otherwise we say that it is \textit{indefinite}.
\end{definition}

\begin{definition}[Maximally indefinite transverse causal order]
    We say that the transverse causal order between $N$ events of a superposition of two spacetimes $(\mathcal{M}_\mathcal{A},g_{\mathcal{A}})$ and $(\mathcal{M}_\mathcal{B},g_{\mathcal{B}})$ is \textit{maximally indefinite} if and only if $\lamed_{ab} = \begin{pmatrix}
            M \\ 2
        \end{pmatrix} - 2\Big\lfloor \frac{M}{2} \Big\rfloor \Big\lceil \frac{M}{2} \Big\rceil$. If it is indefinite but not maximally indefinite we say that it is \textit{braided}.
\end{definition}

\begin{definition}[Total causal indefiniteness]
    The \textit{total causal indefiniteness} of a collection of $N$ events in a superposition of $M$ spacetimes $\{(\mathcal{M}_\mathcal{X},g_{\mathcal{X}})\}_{\mathcal{X}=1}^M$ is
    \begin{equation}
        \Delta := \sum_{1 \leq \mathcal{X} < \mathcal{Y}}^M \delta(\mathcal{X},\mathcal{Y})
    \end{equation}
\end{definition}

The total causal indefiniteness counts how many events in total have an indefinite pairwise causal order across the entire superposition of spacetimes. 

\begin{definition}[Total causal order]
    The \textit{total causal order} between a collection of $M$ spacetimes $\{(\mathcal{M}_\mathcal{X},g_{\mathcal{X}})\}_{\mathcal{X}=1}^M$ is 
\begin{equation}
    \shin := \sum_{1 \leq \mathcal{X} < \mathcal{Y}}^M \ayin^{\mathcal{X}\mathcal{Y}} \equiv \sum_{1 \leq i < j}^N \lamed_{ij}
\end{equation}
\end{definition}
Note that for $M=2$ spacetimes the longitudinal causal order is the total causal order of the collection of spacetimes, i.e.~$\shin = \ayin^{\mathcal{A}\mathcal{B}}$, whilst for $N=2$ events the transverse causal order is the total causal order of the collection of spacetimes, i.e.~$\shin = \lamed_{12}$. The difference between these measures of the indefiniteness of causal order is depicted in Fig.~\ref{fig:Causal orders}.

\begin{proposition}
    \label{max Delta}
    \begin{equation}
    0 \leq \Delta \leq \Big\lfloor \frac{M}{2} \Big\rfloor \Big\lceil \frac{M}{2} \Big\rceil \begin{pmatrix} N \\ 2 \end{pmatrix}
\end{equation}
\end{proposition}

\begin{proof}
    The lower bound for $\Delta$ is straightforward: just assume definite causal order between every pair of events across the whole superposition of spacetimes. The upper bound for $\Delta$ follows from a similar argument as in proposition \ref{Lamed bounds}: take $\big\lfloor \frac{M}{2} \big\rfloor$ of the orderings to be some fixed permutation of event orderings, and the rest $\big\lceil \frac{M}{2} \big\rceil$ to be the opposite ordering so as to have alternating maximally indefinite causal order, which maximises the indefiniteness.
\end{proof}

\begin{proposition}
    \label{Delta shin}
    \begin{equation}
        \Delta = \frac{1}{2} \Bigg[\begin{pmatrix}
            M \\ 2
        \end{pmatrix} \begin{pmatrix}
            N \\ 2
        \end{pmatrix} - \shin\Bigg]
    \end{equation}
\end{proposition}
\begin{proof}
    This follows directly from proposition \ref{Causal indefiniteness and pairwise causal order}.
\end{proof}

\begin{proposition}
    \label{shin bounds}
    For a collection of $M$ spacetimes $\{(\mathcal{M}_\mathcal{X},g_{\mathcal{X}})\}_{\mathcal{X}=1}^M$ with $N$ events,
    \begin{equation}
        \begin{pmatrix}
        N \\ 2
    \end{pmatrix}.\Bigg[\begin{pmatrix}
        M \\ 2
    \end{pmatrix} - 2\Big\lfloor \frac{M}{2} \Big\rfloor \Big\lceil \frac{M}{2} \Big\rceil \Bigg] \leq\shin\leq  \begin{pmatrix}
            N \\ 2
        \end{pmatrix}\begin{pmatrix}
            M \\ 2
        \end{pmatrix}. 
    \end{equation}
\end{proposition}

\begin{proof}
    If the pairwise causal order is definite for each pair of events, which is the case that maximises the longitudinal causal order, we have that $\forall \mathcal{X}\ \forall \mathcal{Y}, \shin^{\mathcal{X}\mathcal{Y}} = \begin{pmatrix} N \\ 2 \end{pmatrix}$ so that
    \begin{align}
        \max \shin &= \sum_{1 \leq \mathcal{X} < \mathcal{Y}}^M \begin{pmatrix} N \\ 2 \end{pmatrix} \\ &= \begin{pmatrix}
            M \\ 2
        \end{pmatrix} \begin{pmatrix}
            N \\ 2
        \end{pmatrix}.
    \end{align}
    The case which minimises the total causal order follows from propositions \ref{max Delta} and \ref{Delta shin}, or equivalently by maximising the longitudinal causal order and applying the same argument as in proposition \ref{Lamed bounds}.
\end{proof}

We can finally define a proper notion of indefinite causal order for arbitrary superpositions of metrics and events.
\begin{definition}[Indefinite total causal order]
    We say that the total causal order between $N$ events of a collection of $M$ spacetimes $\{(\mathcal{M}_\mathcal{X},g_{\mathcal{X}})\}_{\mathcal{X}=1}^M$ in superposition is \textit{indefinite} iff $\shin \neq \begin{pmatrix}
            M \\ 2
        \end{pmatrix} \begin{pmatrix}
            N \\ 2
        \end{pmatrix}$ or, equivalently, iff $\Delta \neq 0$. Otherwise, we say that it is \textit{definite}.
\end{definition}

\begin{definition}[Maximally indefinite total causal order]
    The total causal order between $N$ events of a collection of $M$ spacetimes $\{(\mathcal{M}_\mathcal{X},g_{\mathcal{X}})\}_{\mathcal{X}=1}^M$ in superposition is \textit{maximally indefinite} iff $\shin =\begin{pmatrix}
        N \\ 2
    \end{pmatrix}.\Bigg[\begin{pmatrix}
        M \\ 2
    \end{pmatrix} - 2\Big\lfloor \frac{M}{2} \Big\rfloor \Big\lceil \frac{M}{2} \Big\rceil \Bigg]$ or, equivalently, iff $\Delta = \Big\lfloor \frac{M}{2} \Big\rfloor \Big\lceil \frac{M}{2} \Big\rceil \begin{pmatrix} N \\ 2 \end{pmatrix}$. If it is indefinite but not maximally indefinite we say that it is \textit{braided}.
\end{definition}

\begin{proposition}
    \label{Diffeomorphism invariance}
    The pairwise, longitudinal, transverse, and total causal orders as well as the causal indefiniteness and total causal indefiniteness are invariant under quantum diffeomorphisms.
\end{proposition}

\begin{proof}
    It was shown \cite{Delahamette2022} that quantum diffeomorphisms cannot change the causal order of two events, i.e.~that the pairwise causal order is invariant under quantum diffeomorphisms. Thus, any linear combination of pairwise causal orders will be invariant under quantum diffeomorphisms. By proposition \ref{Causal indefiniteness and pairwise causal order} this extends to causal indefiniteness and so to total causal indefiniteness.
\end{proof}

A similar theorem was proven in the context of the process matrix formalism, where it was shown that causal order is invariant under continuous and reversible transformation \cite{CastroRuiz2018}.

\section{Diagrammatic representation and topology}
\label{sec:knots}

We want to visualise the above introduced concepts of (in-)definite causal order beyond combinatorial considerations. We may adopt some Everettian-like intuition of the form of Fig.~\ref{fig:Everettian representation}. Indeed, this reduces the complexity to a two-dimensional problem: the ordering of events and the superposition of spacetimes.
\begin{figure}[t!]
    \centering
    \begin{tikzpicture}[x=0.75pt,y=0.75pt,yscale=-0.8,xscale=0.8]
\draw    (170.5,165.5) -- (222.5,165) ;
\draw    (222.56,167.86) .. controls (223.68,213.34) and (228.8,197.31) .. (292.5,200.25) ;
\draw [shift={(222.5,165)}, rotate = 88.83] [color={rgb, 255:red, 0; green, 0; blue, 0 }  ][line width=0.75]      (0, 0) circle [x radius= 3.35, y radius= 3.35]   ;
\draw    (292.5,200.25) -- (359,199) ;
\draw [shift={(292.5,200.25)}, rotate = 358.92] [color={rgb, 255:red, 0; green, 0; blue, 0 }  ][fill={rgb, 255:red, 0; green, 0; blue, 0 }  ][line width=0.75]      (0, 0) circle [x radius= 3.35, y radius= 3.35]   ;
\draw    (359,199) -- (427,198.76) ;
\draw [shift={(429,198.75)}, rotate = 179.8] [color={rgb, 255:red, 0; green, 0; blue, 0 }  ][line width=0.75]    (10.93,-3.29) .. controls (6.95,-1.4) and (3.31,-0.3) .. (0,0) .. controls (3.31,0.3) and (6.95,1.4) .. (10.93,3.29)   ;
\draw [shift={(359,199)}, rotate = 359.8] [color={rgb, 255:red, 0; green, 0; blue, 0 }  ][fill={rgb, 255:red, 0; green, 0; blue, 0 }  ][line width=0.75]      (0, 0) circle [x radius= 3.35, y radius= 3.35]   ;
\draw    (222.39,162.21) .. controls (220.68,117.84) and (225.8,132.21) .. (289.5,130.25) ;
\draw [shift={(222.5,165)}, rotate = 267.6] [color={rgb, 255:red, 0; green, 0; blue, 0 }  ][line width=0.75]      (0, 0) circle [x radius= 3.35, y radius= 3.35]   ;
\draw    (289.5,130.25) -- (355.5,130.25) ;
\draw [shift={(289.5,130.25)}, rotate = 0] [color={rgb, 255:red, 0; green, 0; blue, 0 }  ][fill={rgb, 255:red, 0; green, 0; blue, 0 }  ][line width=0.75]      (0, 0) circle [x radius= 3.35, y radius= 3.35]   ;
\draw    (355.5,130.25) -- (423,130.25) ;
\draw [shift={(425,130.25)}, rotate = 180] [color={rgb, 255:red, 0; green, 0; blue, 0 }  ][line width=0.75]    (10.93,-3.29) .. controls (6.95,-1.4) and (3.31,-0.3) .. (0,0) .. controls (3.31,0.3) and (6.95,1.4) .. (10.93,3.29)   ;
\draw [shift={(355.5,130.25)}, rotate = 0] [color={rgb, 255:red, 0; green, 0; blue, 0 }  ][fill={rgb, 255:red, 0; green, 0; blue, 0 }  ][line width=0.75]      (0, 0) circle [x radius= 3.35, y radius= 3.35]   ;
\draw (177.5,104.9) node [anchor=north west][inner sep=0.75pt]    {$(\mathcal{M}_{\mathcal{A}} ,g_{\mathcal{A}})$};
\draw (184.5,203.9) node [anchor=north west][inner sep=0.75pt]    {$(\mathcal{M}_{\mathcal{B}} ,g_{\mathcal{B}})$};
\draw (425.5,118.9) node [anchor=north west][inner sep=0.75pt]    {$\tau _{\mathcal{A}}$};
\draw (430,186.9) node [anchor=north west][inner sep=0.75pt]    {$\tau _{\mathcal{B}}$};
\draw (283.5,110) node [anchor=north west][inner sep=0.75pt]   [align=left] {A};
\draw (353.5,204.5) node [anchor=north west][inner sep=0.75pt]   [align=left] {A};
\draw (350.5,110) node [anchor=north west][inner sep=0.75pt]   [align=left] {B};
\draw (287,204.5) node [anchor=north west][inner sep=0.75pt]   [align=left] {B};
\end{tikzpicture}
    \caption{\justifying A representation of indefinite causal order for two events $\mathcal{E}_1 = A$ and $\mathcal{E}_2 = B$ in a superposition of two spacetimes $(\mathcal{M}_\mathcal{X}, g_\mathcal{X})$, where $\mathcal{X} = \mathcal{A}, \mathcal{B}$. In each spacetime $(\mathcal{M}_{\mathcal{X}}, g_{\mathcal{X}})$, $\tau_\mathcal{X}$ denotes the proper time of a specified timelike worldline. The events are ordered according to the pairwise causal order $s_{12}^{\mathcal{X}}$ between them such that $A \prec B$ iff $s_{12}^\mathcal{X}=\text{sign}(\Delta \tau_\mathcal{X}) = 1$ (that is, iff $A$ occurs at an earlier proper time $\tau_\mathcal{X}$ than $B$.)}
    \label{fig:Everettian representation}
\end{figure}
We would like to stress, however, that, while it may be most naturally interpreted in terms of an Everettian viewpoint, using this representation does not commit us in any way to the many-worlds interpretation of quantum theory. Rather, we can view it as an interpretation-neutral representation of the different branches of the wavefunction in a particular basis, while leaving aside any questions about their metaphysical status. Moreover, this representation alone does not allow us to understand the \emph{relative} ordering of events. A natural way to analyse the relative ordering of events is to ``attach" one string of events onto the other spacetime, i.e.~compare and contrast the different strings on a similar footing.\\

One may imagine that some quantifiers of indefinite causal order can be reflected in braiding operations, and that some topological properties correspond in both the algebraic and diagrammatic pictures. This visualisation also has the advantage of being easily generalisable to arbitrary superpositions of $M$ spacetimes. However, although natural and potentially powerful, this construction runs into several issues. The first one is that there is a clear ambiguity in the braiding choice -- should one braid cross over or under another? More worrying is that such a braiding construction would depend on the choice of reference spacetime -- this is most clear in the case of $M\geq3$, as can be seen in Fig.~\ref{fig:Braiding problems}, but is already present for $M=2$ spacetimes, where setting a braiding convention (over/under) immediately differentiates between two choices of reference spacetimes.

\begin{figure}[h!]
    \centering
    \begin{tikzpicture}[x=0.75pt,y=0.75pt,yscale=-0.55,xscale=0.55]
\draw    (33.5,141.83) -- (85.5,141.33) ;
\draw    (85.56,144.2) .. controls (86.68,189.67) and (91.8,173.64) .. (155.5,176.58) ;
\draw [shift={(85.5,141.33)}, rotate = 88.83] [color={rgb, 255:red, 0; green, 0; blue, 0 }  ][line width=0.75]      (0, 0) circle [x radius= 3.35, y radius= 3.35]   ;
\draw    (155.5,176.58) -- (222,175.33) ;
\draw [shift={(155.5,176.58)}, rotate = 358.92] [color={rgb, 255:red, 0; green, 0; blue, 0 }  ][fill={rgb, 255:red, 0; green, 0; blue, 0 }  ][line width=0.75]      (0, 0) circle [x radius= 3.35, y radius= 3.35]   ;
\draw    (222,175.33) -- (290,175.09) ;
\draw [shift={(292,175.08)}, rotate = 179.8] [color={rgb, 255:red, 0; green, 0; blue, 0 }  ][line width=0.75]    (10.93,-3.29) .. controls (6.95,-1.4) and (3.31,-0.3) .. (0,0) .. controls (3.31,0.3) and (6.95,1.4) .. (10.93,3.29)   ;
\draw [shift={(222,175.33)}, rotate = 359.8] [color={rgb, 255:red, 0; green, 0; blue, 0 }  ][fill={rgb, 255:red, 0; green, 0; blue, 0 }  ][line width=0.75]      (0, 0) circle [x radius= 3.35, y radius= 3.35]   ;
\draw    (85.39,138.54) .. controls (83.68,94.18) and (88.8,108.54) .. (152.5,106.58) ;
\draw [shift={(85.5,141.33)}, rotate = 267.6] [color={rgb, 255:red, 0; green, 0; blue, 0 }  ][line width=0.75]      (0, 0) circle [x radius= 3.35, y radius= 3.35]   ;
\draw    (152.5,106.58) -- (218.5,106.58) ;
\draw [shift={(152.5,106.58)}, rotate = 0] [color={rgb, 255:red, 0; green, 0; blue, 0 }  ][fill={rgb, 255:red, 0; green, 0; blue, 0 }  ][line width=0.75]      (0, 0) circle [x radius= 3.35, y radius= 3.35]   ;
\draw    (218.5,106.58) -- (286,106.58) ;
\draw [shift={(288,106.58)}, rotate = 180] [color={rgb, 255:red, 0; green, 0; blue, 0 }  ][line width=0.75]    (10.93,-3.29) .. controls (6.95,-1.4) and (3.31,-0.3) .. (0,0) .. controls (3.31,0.3) and (6.95,1.4) .. (10.93,3.29)   ;
\draw [shift={(218.5,106.58)}, rotate = 0] [color={rgb, 255:red, 0; green, 0; blue, 0 }  ][fill={rgb, 255:red, 0; green, 0; blue, 0 }  ][line width=0.75]      (0, 0) circle [x radius= 3.35, y radius= 3.35]   ;
\draw    (219.67,140.33) -- (287.17,140.33) ;
\draw [shift={(289.17,140.33)}, rotate = 180] [color={rgb, 255:red, 0; green, 0; blue, 0 }  ][line width=0.75]    (10.93,-3.29) .. controls (6.95,-1.4) and (3.31,-0.3) .. (0,0) .. controls (3.31,0.3) and (6.95,1.4) .. (10.93,3.29)   ;
\draw [shift={(219.67,140.33)}, rotate = 0] [color={rgb, 255:red, 0; green, 0; blue, 0 }  ][fill={rgb, 255:red, 0; green, 0; blue, 0 }  ][line width=0.75]      (0, 0) circle [x radius= 3.35, y radius= 3.35]   ;
\draw    (87.85,140.91) .. controls (92.01,140.55) and (103,140.87) .. (153.67,140.33) ;
\draw [shift={(85.5,141.33)}, rotate = 344.05] [color={rgb, 255:red, 0; green, 0; blue, 0 }  ][line width=0.75]      (0, 0) circle [x radius= 3.35, y radius= 3.35]   ;
\draw    (153.67,140.33) -- (219.67,140.33) ;
\draw [shift={(153.67,140.33)}, rotate = 0] [color={rgb, 255:red, 0; green, 0; blue, 0 }  ][fill={rgb, 255:red, 0; green, 0; blue, 0 }  ][line width=0.75]      (0, 0) circle [x radius= 3.35, y radius= 3.35]   ;
\draw    (366.67,140.5) -- (418.67,140) ;
\draw    (418.73,142.86) .. controls (419.85,188.34) and (424.97,172.31) .. (488.67,175.25) ;
\draw [shift={(418.67,140)}, rotate = 88.83] [color={rgb, 255:red, 0; green, 0; blue, 0 }  ][line width=0.75]      (0, 0) circle [x radius= 3.35, y radius= 3.35]   ;
\draw    (488.67,175.25) -- (555.17,174) ;
\draw [shift={(488.67,175.25)}, rotate = 358.92] [color={rgb, 255:red, 0; green, 0; blue, 0 }  ][fill={rgb, 255:red, 0; green, 0; blue, 0 }  ][line width=0.75]      (0, 0) circle [x radius= 3.35, y radius= 3.35]   ;
\draw    (555.17,174) -- (623.17,173.76) ;
\draw [shift={(625.17,173.75)}, rotate = 179.8] [color={rgb, 255:red, 0; green, 0; blue, 0 }  ][line width=0.75]    (10.93,-3.29) .. controls (6.95,-1.4) and (3.31,-0.3) .. (0,0) .. controls (3.31,0.3) and (6.95,1.4) .. (10.93,3.29)   ;
\draw [shift={(555.17,174)}, rotate = 359.8] [color={rgb, 255:red, 0; green, 0; blue, 0 }  ][fill={rgb, 255:red, 0; green, 0; blue, 0 }  ][line width=0.75]      (0, 0) circle [x radius= 3.35, y radius= 3.35]   ;
\draw    (418.55,137.21) .. controls (416.85,92.84) and (421.97,107.21) .. (485.67,105.25) ;
\draw [shift={(418.67,140)}, rotate = 267.6] [color={rgb, 255:red, 0; green, 0; blue, 0 }  ][line width=0.75]      (0, 0) circle [x radius= 3.35, y radius= 3.35]   ;
\draw    (485.67,105.25) -- (551.67,105.25) ;
\draw [shift={(485.67,105.25)}, rotate = 0] [color={rgb, 255:red, 0; green, 0; blue, 0 }  ][fill={rgb, 255:red, 0; green, 0; blue, 0 }  ][line width=0.75]      (0, 0) circle [x radius= 3.35, y radius= 3.35]   ;
\draw    (551.67,105.25) -- (619.17,105.25) ;
\draw [shift={(621.17,105.25)}, rotate = 180] [color={rgb, 255:red, 0; green, 0; blue, 0 }  ][line width=0.75]    (10.93,-3.29) .. controls (6.95,-1.4) and (3.31,-0.3) .. (0,0) .. controls (3.31,0.3) and (6.95,1.4) .. (10.93,3.29)   ;
\draw [shift={(551.67,105.25)}, rotate = 0] [color={rgb, 255:red, 0; green, 0; blue, 0 }  ][fill={rgb, 255:red, 0; green, 0; blue, 0 }  ][line width=0.75]      (0, 0) circle [x radius= 3.35, y radius= 3.35]   ;
\draw    (552.83,139) -- (620.33,139) ;
\draw [shift={(622.33,139)}, rotate = 180] [color={rgb, 255:red, 0; green, 0; blue, 0 }  ][line width=0.75]    (10.93,-3.29) .. controls (6.95,-1.4) and (3.31,-0.3) .. (0,0) .. controls (3.31,0.3) and (6.95,1.4) .. (10.93,3.29)   ;
\draw [shift={(552.83,139)}, rotate = 0] [color={rgb, 255:red, 0; green, 0; blue, 0 }  ][fill={rgb, 255:red, 0; green, 0; blue, 0 }  ][line width=0.75]      (0, 0) circle [x radius= 3.35, y radius= 3.35]   ;
\draw    (421.02,139.58) .. controls (425.18,139.21) and (436.17,139.53) .. (486.83,139) ;
\draw [shift={(418.67,140)}, rotate = 344.05] [color={rgb, 255:red, 0; green, 0; blue, 0 }  ][line width=0.75]      (0, 0) circle [x radius= 3.35, y radius= 3.35]   ;
\draw    (486.83,139) -- (552.83,139) ;
\draw [shift={(486.83,139)}, rotate = 0] [color={rgb, 255:red, 0; green, 0; blue, 0 }  ][fill={rgb, 255:red, 0; green, 0; blue, 0 }  ][line width=0.75]      (0, 0) circle [x radius= 3.35, y radius= 3.35]   ;
\draw    (152.5,106.58) .. controls (159,121.25) and (220,126.75) .. (219.67,140.33) ;
\draw    (219.67,140.33) .. controls (215,149.67) and (207,151.92) .. (193,157.25) ;
\draw    (218.5,106.58) .. controls (213.83,115.92) and (205.4,116.87) .. (191.4,122.2) ;
\draw    (182.6,125) .. controls (170.3,124.48) and (150.5,127.25) .. (153.67,140.33) ;
\draw    (186.5,160.75) .. controls (178.2,164.1) and (154.33,165) .. (155.5,176.58) ;
\draw    (486.83,139) -- (488.67,175.25) ;
\draw    (552.83,139) -- (555.17,174) ;
\draw    (485.67,105.25) .. controls (492.17,119.92) and (553.17,125.42) .. (552.83,139) ;
\draw    (515.77,123.67) .. controls (503.47,123.15) and (483.67,125.92) .. (486.83,139) ;
\draw    (551.67,105.25) .. controls (547,114.58) and (538.57,115.53) .. (524.57,120.87) ;
\draw    (153.67,140.33) .. controls (160.17,155) and (222.33,161.75) .. (222,175.33) ;

\draw (30,81.23) node [anchor=north west][inner sep=0.75pt]    {\tiny $(\mathcal{M}_{\mathcal{A}},g_\mathcal{A})$};
\draw (30,177.57) node [anchor=north west][inner sep=0.75pt]    {\tiny $(\mathcal{M}_{\mathcal{C}},g_\mathcal{C})$};
\draw (289.83,93.9) node [anchor=north west][inner sep=0.75pt]    {$\tau _{\mathcal{A}}$};
\draw (291,128.57) node [anchor=north west][inner sep=0.75pt]    {$\tau _{\mathcal{B}}$};
\draw (146.5,82.83) node [anchor=north west][inner sep=0.75pt]   [align=left] {A};
\draw (216.5,180.83) node [anchor=north west][inner sep=0.75pt]   [align=left] {B};
\draw (213.5,82.33) node [anchor=north west][inner sep=0.75pt]   [align=left] {B};
\draw (150,179.83) node [anchor=north west][inner sep=0.75pt]   [align=left] {A};
\draw (5,118.23) node [anchor=north west][inner sep=0.75pt]    {\tiny $(\mathcal{M}_{\mathcal{B}},g_\mathcal{B})$};
\draw (292.33,163.9) node [anchor=north west][inner sep=0.75pt]    {$\tau _{\mathcal{C}}$};
\draw (224.5,120.83) node [anchor=north west][inner sep=0.75pt]   [align=left] {A};
\draw (138.17,121.67) node [anchor=north west][inner sep=0.75pt]   [align=left] {B};
\draw (370,82.57) node [anchor=north west][inner sep=0.75pt]    {\tiny $(\mathcal{M}_{\mathcal{B}},g_\mathcal{B})$};
\draw (370,174.9) node [anchor=north west][inner sep=0.75pt]    {\tiny $(\mathcal{M}_{\mathcal{C}},g_\mathcal{C})$};
\draw (623,92.57) node [anchor=north west][inner sep=0.75pt]    {$\tau _{\mathcal{B}}$};
\draw (624.17,127.23) node [anchor=north west][inner sep=0.75pt]    {$\tau _{\mathcal{A}}$};
\draw (479.67,81.5) node [anchor=north west][inner sep=0.75pt]   [align=left] {B};
\draw (549.67,179.5) node [anchor=north west][inner sep=0.75pt]   [align=left] {B};
\draw (546.67,81) node [anchor=north west][inner sep=0.75pt]   [align=left] {A};
\draw (483.17,178.5) node [anchor=north west][inner sep=0.75pt]   [align=left] {A};
\draw (340.33,118.23) node [anchor=north west][inner sep=0.75pt]    {\tiny $(\mathcal{M}_{\mathcal{A}},g_\mathcal{A})$};
\draw (625.5,162.57) node [anchor=north west][inner sep=0.75pt]    {$\tau _{\mathcal{C}}$};
\draw (557.67,119.5) node [anchor=north west][inner sep=0.75pt]   [align=left] {B};
\draw (471.33,120.33) node [anchor=north west][inner sep=0.75pt]   [align=left] {A};
\draw (316,123.07) node [anchor=north west][inner sep=0.75pt]  [font=\Large]  {$\neq $};
\end{tikzpicture}
    \caption{\justifying
    Braid diagram for three spacetimes in superposition and two events. A disadvantage of the braid diagram is that this construction depends explicitly on the choice of reference spacetime -- $(\mathcal{M}_\mathcal{B}, g_\mathcal{B})$ in the left and $(\mathcal{M}_\mathcal{A}, g_\mathcal{A})$ in the right diagram.}
    \label{fig:Braiding problems}
\end{figure}

A more fruitful relation can be established when considering \sam{graph and }knot theory. We now proceed to the diagrammatic construction, which we will use to represent the causal ordering of events as knot diagrams. Let the diagram representing the causal ordering of $\mathcal{A}$ and $\mathcal{B}$ be $\mathcal{D}_{\mathcal{A}\mathcal{B}}$. Consider two strings of events $\mathcal{S_A}$ and $\mathcal{S_B}$. Let 
\begin{align}
    \label{Bijection}
    \varphi: \mathcal{S}_\mathcal{A} \times \mathcal{S}_\mathcal{B} &\to \mathcal{S}_{\mathcal{A}\mathcal{B}} \nonumber \\
    \varphi(\mathcal{E}_i^{\mathcal{A}},
    \mathcal{E}_j^{\mathcal{B}}) &= \{(\mathcal{E}_i^{\mathcal{A}},\mathcal{E}_j^{\mathcal{B}})\}_{\delta_{ij}}
 \end{align}
 where $\{X\}_{\delta_{ij}} = \begin{cases}
     \{X\} \text{ if } i=j \\
     \varnothing \text{ if } i \neq j
\end{cases}$. Note that, in general, $\mathcal{S_{AB}} \neq \mathcal{S_{BA}}$. Let the drawing map be \acd{$\gamma: \mathcal{S}_{\mathcal{AB}} \times \mathcal{S}_{\mathcal{AB}} \to \mathcal{D}_{\mathcal{A}\mathcal{B}}$} such that
\begin{widetext}
    \centering
    \begin{tikzpicture}[x=0.75pt,y=0.75pt,yscale=-1,xscale=1]

\draw    (275,85) -- (325,85) ;
\draw [shift={(305,85)}, rotate = 179.61] [fill={rgb, 255:red, 0; green, 0; blue, 0 }  ][line width=0.08]  [draw opacity=0] (8.93,-4.29) -- (0,0) -- (8.93,4.29) -- cycle    ;
\draw    (275,105) -- (325,105) ;
\draw [shift={(305,105)}, rotate = 179.61] [fill={rgb, 255:red, 0; green, 0; blue, 0 }  ][line width=0.08]  [draw opacity=0] (8.93,-4.29) -- (0,0) -- (8.93,4.29) -- cycle    ;
\draw    (275,125) -- (325,125) ;
\draw [shift={(305,125)}, rotate = 179.61] [fill={rgb, 255:red, 0; green, 0; blue, 0 }  ][line width=0.08]  [draw opacity=0] (8.93,-4.29) -- (0,0) -- (8.93,4.29) -- cycle    ;
\draw  [dash pattern={on 0.84pt off 2.51pt}]  (275,145) -- (325,145) ;
\draw [shift={(305,145)}, rotate = 179.61] [fill={rgb, 255:red, 0; green, 0; blue, 0 }  ][line width=0.08]  [draw opacity=0] (8.93,-4.29) -- (0,0) -- (8.93,4.29) -- cycle    ;
\draw  [dash pattern={on 0.84pt off 2.51pt}]  (275,165) -- (325,165) ;
\draw [shift={(305,165)}, rotate = 180] [fill={rgb, 255:red, 0; green, 0; blue, 0 }  ][line width=0.08]  [draw opacity=0] (8.93,-4.29) -- (0,0) -- (8.93,4.29) -- cycle    ;
\draw  [dash pattern={on 4.5pt off 4.5pt}] (275,182) -- (325,182) ;
\draw [shift={(305,182)}, rotate = 180] [fill={rgb, 255:red, 0; green, 0; blue, 0 }  ][line width=0.08]  [draw opacity=0] (8.93,-4.29) -- (0,0) -- (8.93,4.29) -- cycle    ;

\draw (89,73.9) node [anchor=north west][inner sep=0.75pt] {$\gamma \left(\left(\mathcal{E}_{i}^{\mathcal{A}} ,\mathcal{E}_{i}^{\mathcal{B}}\right) ,\left(\mathcal{E}_{j}^{\mathcal{A}} ,\mathcal{E}_{j}^{\mathcal{B}}\right)\right) =\begin{cases}
\mathcal{E}_{i}^{\mathcal{B}} \ \ \ \ \ \ \ \ \ \ \ \ \ \mathcal{E}_{j}^{\mathcal{A}} & \text{if }\ \mathcal{E}_{i}^{\mathcal{B}} =\mathcal{S}_{\mathcal{B}}^{( 1)} \ \& \ \mathcal{E}_j^\mathcal{A} = \mathcal{S}_\mathcal{A}^{(1)}  \\
\mathcal{E}_{i}^{\mathcal{A}} \ \ \ \ \ \ \ \ \ \ \ \ \ \mathcal{E}_{j}^{\mathcal{A}} & \text{if }\ \mathcal{S}_{\mathcal{A}}^{( i)} =\mathcal{S}_{\mathcal{A}}^{( j-1)}\\
\mathcal{E}_{i}^{\mathcal{A}} \ \ \ \ \ \ \ \ \ \ \ \ \ \mathcal{E}_{j}^{\mathcal{B}} & \text{if }\ \mathcal{E}_{i}^\mathcal{A} = \mathcal{S}_{\mathcal{A}}^{(N)} \ \&\ \mathcal{E}_{j}^{\mathcal{B}} = \mathcal{S}_{\mathcal{B}}^{( N)}\\
\mathcal{E}_{i}^{\mathcal{B}} \ \ \ \ \ \ \ \ \ \ \ \ \ \mathcal{E}_{i}^{\mathcal{A}} & \text{if }\ \mathcal{E}_{i}^{\mathcal{B}} =\mathcal{S}_{\mathcal{B}}^{( N)} \\
\mathcal{E}_{i}^{\mathcal{A}} \ \ \ \ \ \ \ \ \ \ \ \ \ \mathcal{E}_{j}^{\mathcal{A}} & \text{if }\ \mathcal{E}_{j}^{\mathcal{B}} =\mathcal{S}_{\mathcal{B}}^{(\alpha)} \ \&\ \mathcal{E}_{i}^{\mathcal{B}} =\mathcal{S}_{\mathcal{B}}^{(\alpha+1)} \text{ for some } \alpha \in \{1,...,N-1\} \ \\
\mathcal{E}_{i}^{\mathcal{A}} \ \ \ \ \ \ \ \ \ \ \ \ \ \mathcal{E}_{i}^{\mathcal{B}} & \text{if }\ \mathcal{E}_{i}^{\mathcal{B}} = \mathcal{S}_{\mathcal{B}}^{( 1)}
\end{cases}$};
\end{tikzpicture}
\end{widetext}

Here, $\mathcal{S}_{\mathcal{X}}^{(i)}$, $\mathcal{X} = \mathcal{A},\mathcal{B}$ denotes the $i$-th element of the ordered sequence $\mathcal{S}_{\mathcal{X}}$. \sam{We project} all events of a given spacetime on one line -- spacetime $\mathcal{M}_\mathcal{A}$ will be the top line, while spacetime $\mathcal{M}_\mathcal{B}$ will be the bottom line. The intuition behind the above prescription is as follows (see also Figs.~\ref{fig:(2,2) definite causal order} and \ref{fig:(2,2) indefinite causal order} for an illustration in a concrete example):  
Start from the first entry of $\mathcal{S}_\mathcal{B}$ and draw a solid arrow to the first entry of $\mathcal{S}_\mathcal{A}$. Connect all events in $\mathcal{S}_\mathcal{A}$ with solid arrows according to their causal order until you reach $\mathcal{S}_\mathcal{A}^{(N)}$. Go back down to $\mathcal{S}_\mathcal{B}^{(N)}$ with a final solid arrow, connecting the two branches. Next, draw a \sam{short}-dashed arrow from $\mathcal{S}_\mathcal{B}^{(N)}$ to its counterpart in $\mathcal{S}_\mathcal{A}$ (recalling that $\mathcal{E}_i^\mathcal{A} = \mathcal{E}_i^\mathcal{B}$ for all $i$). Then, go ``backwards'' in the causal order prescribed by $\mathcal{S}_\mathcal{B}$ until you reach the counterpart of $\mathcal{S}_{\mathcal{B}}^{(1)}$. At this point, draw a final \sam{long}-dashed arrow back to the first entry of $\mathcal{S}_{\mathcal{B}}$ to close the diagram. 

\sam{
Formally, this construction yields a tricolored directed graph $G = (\mathcal{S_A} \cup \mathcal{S_B} \cup \tilde{S}, E_{\mathcal{AB}}, \gamma)$, where $E_{\mathcal{AB}}$ are the edges of the graph and $\tilde{S}$ is the set of all intersections of the short-dashed arrows with both short-dashed and long-dashed arrows. It is tricolored as there are three distinct types of edges (solid, small-dashed and long-dashed), and directed as each edge is endowed with a direction. Moreover, the graph can be ordered (as a directed graph), with the starting point being the vertex corresponding to $\mathcal{S}^{(1)}_{\mathcal{B}}$.}

By construction, each diagram is uniquely defined for each $\mathcal{S_{AB}}$\sam{, though the diagrammatic representation of $\mathcal{S_{AB}}$ will generally be different to that of $\mathcal{S_{BA}}$. We will come back to that point in the following section.} For example, if $\mathcal{S}_\mathcal{A} = \{A \prec B\}$ while $\mathcal{S}_\mathcal{B} = \{A \prec B\}$ (i.e.~definite causal order), the associated diagram is given by Fig.~\ref{fig:(2,2) definite causal order}, whilst if $\mathcal{S}_\mathcal{A} = \{A \prec B\}$ while $\mathcal{S}_\mathcal{B} = \{B \prec A\}$ (i.e.~indefinite causal order), the resulting diagram is given by Fig.~\ref{fig:(2,2) indefinite causal order}.
\begin{figure}[b!]
    \centering
    \subfloat[Definite causal order for two spacetimes and two events. \label{fig:(2,2) definite causal order}]{\begin{tikzpicture}[x=0.75pt,y=0.75pt,yscale=-0.6,xscale=0.6]
\draw    (250.97,78.27) -- (481.62,76.8) ;
\draw [shift={(371.29,77.51)}, rotate = 179.64] [fill={rgb, 255:red, 0; green, 0; blue, 0 }  ][line width=0.08]  [draw opacity=0] (8.93,-4.29) -- (0,0) -- (8.93,4.29) -- cycle    ;
\draw    (247.91,183.77) -- (250.97,78.27) ;
\draw [shift={(249.58,126.02)}, rotate = 91.66] [fill={rgb, 255:red, 0; green, 0; blue, 0 }  ][line width=0.08]  [draw opacity=0] (8.93,-4.29) -- (0,0) -- (8.93,4.29) -- cycle    ;
\draw    (481.62,76.8) -- (478.56,183.77) ;
\draw [shift={(479.95,135.28)}, rotate = 271.64] [fill={rgb, 255:red, 0; green, 0; blue, 0 }  ][line width=0.08]  [draw opacity=0] (8.93,-4.29) -- (0,0) -- (8.93,4.29) -- cycle    ;
\draw  [dash pattern={on 0.84pt off 2.51pt}]  (478.56,183.77) .. controls (449.54,144.94) and (452.59,114.17) .. (481.62,76.8) ;
\draw [shift={(458.77,124.57)}, rotate = 94.46] [fill={rgb, 255:red, 0; green, 0; blue, 0 }  ][line width=0.08]  [draw opacity=0] (8.93,-4.29) -- (0,0) -- (8.93,4.29) -- cycle    ;
\draw  [dash pattern={on 4.5pt off 4.5pt}]  (250.97,78.27) .. controls (286.1,100.98) and (267.77,190.36) .. (247.91,183.77) ;
\draw [shift={(270.41,138.33)}, rotate = 273.18] [fill={rgb, 255:red, 0; green, 0; blue, 0 }  ][line width=0.08]  [draw opacity=0] (8.93,-4.29) -- (0,0) -- (8.93,4.29) -- cycle    ;
\draw  [dash pattern={on 0.84pt off 2.51pt}]  (481.62,76.8) .. controls (421,125.5) and (304,125.5) .. (250.97,78.27) ;
\draw [shift={(360.84,113.51)}, rotate = 359.76] [fill={rgb, 255:red, 0; green, 0; blue, 0 }  ][line width=0.08]  [draw opacity=0] (8.93,-4.29) -- (0,0) -- (8.93,4.29) -- cycle    ;

\draw (231.22,54.38) node [anchor=north west][inner sep=0.75pt]   [align=left] {A};
\draw (229.69,184.79) node [anchor=north west][inner sep=0.75pt]   [align=left] {A};
\draw (486.31,54.38) node [anchor=north west][inner sep=0.75pt]   [align=left] {B};
\draw (484.78,192.12) node [anchor=north west][inner sep=0.75pt]   [align=left] {B};

\end{tikzpicture}}
    \bigskip
    \subfloat[Indefinite causal order for two spacetimes and two events. \label{fig:(2,2) indefinite causal order}]{\begin{tikzpicture}[x=0.75pt,y=0.75pt,yscale=-0.6,xscale=0.6]

\draw    (242.45,93.56) -- (469.3,92.14) ;
\draw [shift={(360.87,92.82)}, rotate = 179.64] [fill={rgb, 255:red, 0; green, 0; blue, 0 }  ][line width=0.08]  [draw opacity=0] (8.93,-4.29) -- (0,0) -- (8.93,4.29) -- cycle    ;
\draw    (239.45,195.84) -- (242.45,93.56) ;
\draw [shift={(241.1,139.7)}, rotate = 91.68] [fill={rgb, 255:red, 0; green, 0; blue, 0 }  ][line width=0.08]  [draw opacity=0] (8.93,-4.29) -- (0,0) -- (8.93,4.29) -- cycle    ;
\draw    (469.3,92.14) -- (466.29,195.84) ;
\draw [shift={(467.65,148.99)}, rotate = 271.66] [fill={rgb, 255:red, 0; green, 0; blue, 0 }  ][line width=0.08]  [draw opacity=0] (8.93,-4.29) -- (0,0) -- (8.93,4.29) -- cycle    ;
\draw  [dash pattern={on 0.84pt off 2.51pt}]  (466.29,195.84) -- (242.45,93.56) ;
\draw [shift={(349.82,142.62)}, rotate = 24.56] [fill={rgb, 255:red, 0; green, 0; blue, 0 }  ][line width=0.08]  [draw opacity=0] (8.93,-4.29) -- (0,0) -- (8.93,4.29) -- cycle    ;
\draw  [dash pattern={on 0.84pt off 2.51pt}]  (242.45,93.56) .. controls (313.38,126.61) and (403.31,112.15) .. (469.3,92.14) ;
\draw [shift={(361.29,112.55)}, rotate = 177.04] [fill={rgb, 255:red, 0; green, 0; blue, 0 }  ][line width=0.08]  [draw opacity=0] (8.93,-4.29) -- (0,0) -- (8.93,4.29) -- cycle    ;
\draw  [dash pattern={on 4.5pt off 4.5pt}]  (469.3,92.14) -- (239.45,195.84) ;
\draw [shift={(354.37,143.99)}, rotate = 335.72] [fill={rgb, 255:red, 0; green, 0; blue, 0 }  ][line width=0.08]  [draw opacity=0] (8.93,-4.29) -- (0,0) -- (8.93,4.29) -- cycle    ;

\draw (225.19,66.1) node [anchor=north west][inner sep=0.75pt]   [align=left] {$\displaystyle A$};
\draw (473.81,68.73) node [anchor=north west][inner sep=0.75pt]   [align=left] {B};
\draw (468.56,195.36) node [anchor=north west][inner sep=0.75pt]   [align=left] {$\displaystyle A$};
\draw (227.44,199.41) node [anchor=north west][inner sep=0.75pt]   [align=left] {B};

\end{tikzpicture}}
    
    \caption{\justifying Diagrams for the causal order of two events for two spacetimes in superposition.}
    \label{fig:diag (2,2)}
\end{figure}

\sam{
In fact, the resulting graph is a knot in the graph-theoretic sense: it is a collection of vertices and edges with the property that every vertex in the knot has outgoing edges, and all outgoing edges from vertices in the knot terminate at other vertices in the knot. This is interesting for us as it will allow us to link combinatorial considerations (related to quantifiers of causal indefiniteness) with topological considerations (related to invariant knot polynomials), and will thus provide a way to ensure that some notion of ``topological protection" holds. This motivates the following knot construction, which connects the diagrammatic representation of the causal ordering of events to knots in a natural \sam{and, importantly, unambiguous fashion.}
}

\subsubsection{Fine-grained knot construction}

\vik{Informally, the construction of a knot from the diagrammatic representation of causal order can be understood in terms of the following drawing rules: Start by following the arrows of the directed graph. If a small-dashed arrow meets a solid-dashed arrow (a) \emph{wind under and then over} the vertex if it continues as a small-dashed arrow and (b) \emph{stay under} the vertex if it turns into a long-dashed arrow. As for the crossings outside of the events, i.e.~at the vertices in $\tilde{\mathcal{S}}$, we can distinguish two cases. If two small-dashed arrows meet, \emph{stay over} the vertex, and if a long-dashed arrow meets a short-dashed arrow, \emph{stay under} the vertex. This is illustrated for the examples of definite and indefinite causal order between two events in Fig.~\ref{Unknot definite causal order} and Fig.~\ref{Trefoil indefinite causal order}, respectively. We now provide a formal construction thereof. The quick reader may wish to skip the precise steps of the algorithm laid down below. The ensuing results may be understood without this additional layer of rigour.}

\vik{Formally, one can encode the winding of the lines in the so-called Dowker-Thistlethwaite (DT) word \cite{dowker_classification_1983}. This is a sequence of numbers which relates knots and graphs in the following manner: starting from a knot, take an arbitrary projection thereof, thus generating a graph. Pick an arbitrary starting point, start a counter $c=1$, and label each of the crossings with $c = 1,\cdots,2n$ following a direction (here, associated with the arrows). If the label is an even number and the strand followed crosses \emph{over} at the crossing, then change the sign on the label to be negative. At the end of this procedure, each crossing will be labelled \acd{with} a pair of integers, one even and one odd. The DT word is the ordered sequence of the even numbers generated like this. Examples of DT words are given in Appendix \ref{sec:App conj}.\\
Using DT words, we can formalise the connection between the diagrammatic representation of the causal order with a particular knot as follows.}

\begin{algorithm}
\label{alg:knots}
\vik{Start from the directed graph $G$, choose $V_1 = \mathcal{S}_\mathcal{B}^{(1)}$ as a starting point to obtain an ordered graph, and label the remaining vertices by $V_i$ in ascending order.
\begin{enumerate}
\item The crossings occur at those vertices with four emerging edges. Associate a pair of integers $(a_i,b_i)$ to each of the vertices $V_i \in G$ that involve only dashed arrows.
\item The vertices, involving both dashed and solid arrows, correspond to those crossings where the lines wind around or underneath each other, thus doubling the crossing points in the projection of the corresponding knot. Therefore, associate to the vertices $V_i$ satisfying this criterion \emph{two} tuples of integers, denoted by $((a_i,c_i),(b_i,d_i))$.
\item Start a counter $c=1$ at $V_1$ and follow the arrows of the graph $G$. At every passing of a vertex, fill up the first entry of each pair with the counter, then set $c \mapsto c+1$. For those vertices involving two pairs, replace first $a_i$ and then $b_i$ such that the tuple reads $((c,c_i),(c+1,d_i))$. Repeat this procedure after the first round, without resetting the counter, replacing $c_i$ (and then $d_i$) until all values have been fixed.
\item Change the sign of the even numbers for those vertices at which the corresponding strand crosses \emph{over} the other at the crossing. In terms of graph labels this corresponds to changing the sign of the \emph{even} numbers as follows:
\begin{enumerate}
\item If the vertex has two small-dashed and two solid edges, change the sign of $a_i$ and $d_i$ (if they are even). This corresponds to \emph{winding under and then over} the vertex during the second round.
\item If the vertex has a small-dashed, a long-dashed, and two solid edges, change the sign of $a_i$ or $b_i$, (whichever one is even). This corresponds to \emph{staying under} the vertex during the second round.
\item If the vertex has four small-dashed edges, change the sign of $b_i$ (if even). This corresponds to \emph{staying over} the vertex in the second round.
\item If the vertex has two small-dashed and two long-dashed edges, change the sign of $a_i$ (if even). This corresponds to \emph{staying under} the vertex in the second round.
\end{enumerate}
\item This provides a sequence of numbers $\{1,n_1\},\{3,n_2\},\{5,n_3\}, \dots$ for $n_i \in 2\mathbb{Z}$. The DT code is then $n_1 n_2 n_3 \dots$.
\end{enumerate}
}
\end{algorithm}

\begin{figure*}
    \includegraphics[scale=0.45]{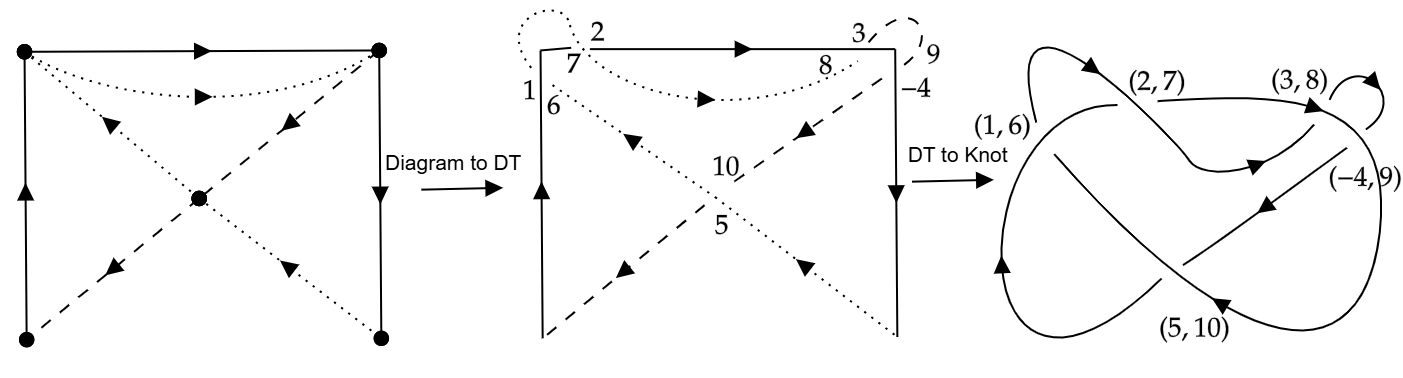}
    \caption{\justifying {Passing from a diagrammatic construction to a knot in the case of AB-BA. We first assign the DT numbers to each vertex by following the arrows and remembering the winding rules for the sign of each even number. From the DT code, the knot can be reconstructed: in this case, the DT word (following 1,3,5,7,9) is 6,8,10,2,-4 which yields a trefoil knot.}}
    \label{fig:informal}
\end{figure*}

\sam{A simple example of this procedure is given in Fig. \ref{fig:informal}. This algorithm provides rules to associate DT letters to each crossing of the graph seen as a specific projection of a knot. This projection is a choice, but fixing the rules makes the resulting knot construction unambiguous, and highlights that one can deform the (now DT-numbered) diagram in any way desired without changing the resulting knot. In a way, this is seeing the diagrammatic representation as a projection of the knot representation, i.e.~one could very well map a causal sequence directly to a (unique) knot, and only then choose a specific projection as a diagrammatic representation. One question remains, however: is the resulting knot independent of the choice of a base spacetime, that is, of the spacetime chosen as $\mathcal{S}_\mathcal{A}$? Numerically, it seems so -- and though we are unable to provide a definite proof for this, we postulate this for the following.}

\sam{

\begin{conjecture}
    \label{conj:renaming}
    The knot built from the fine-grained diagram of the causal sequences using Algorithm \ref{alg:knots} is independent of the choice of reference spacetime $\mathcal{S}_\mathcal{A}$.
\end{conjecture}

There is evidence for the validity of this conjecture, which is shown to hold for all pairs of sequences of $N \leq 4$ events in Appendix \ref{sec:App conj} (that is, for the first $2! + 3! + 4! = 32$ knots). If this conjecture holds, it is equivalent to saying that the rules associating DT letters to each crossing are consistent, and the DT words of the \enquote{renamed} knots (which arise from choosing the other base spacetime) give the same knot as the DT words of the original knots. Note that if this conjecture holds true, then about half of the causal sequences for any number $N$ of events will be associated to another causal sequence, i.e.~we should really see such pairs of causal sequences as \enquote{equivalent} in some respect (which, of course, is required by consistency: there is no physical meaning to choosing a preferred base spacetime). For example, ABC-BCA and ABC-CAB represent the same physical situation: let $\text{A}':=\text{B}$, $\text{B}':=\text{C}$ and $\text{C}':= \text{A}$, then ABC-BCA is just C'A'B'-A'B'C', i.e.~ABC-CAB.
}

\sam{Note that pairs of diagrams which are identical upon a reflection around a vertical axis (up to long-dashed arrows being identified as small-dashed arrows) trivially lead to the same knot: this corresponds to \enquote{flipping} the knot \acd{by} 180 degrees. The identification of a long-dashed arrow to a small-dashed arrow \acd{that} terminates at a vertex with an incoming or outgoing solid arrow is equivalent to choosing a different starting point in the DT knot construction, which does not affect the resulting knot. Thus pairs of causal sequences which present such a symmetry, such as ABC-BCA and ABC-CAB, as well as ABCD-CBDA and ABCD-DBAC respectively, yield the same knot and verify Conjecture \ref{conj:renaming}.}

We may then start categorising the type of knot associated with each causal order.

\begin{theorem}
    For a collection of two spacetimes $\{(\mathcal{M}_\mathcal{A},g_{\mathcal{A}}),(\mathcal{M}_\mathcal{B},g_{\mathcal{B}})\}$ with two events, the knot associated with definite causal order is an unknot.
\end{theorem}

\begin{proof}
    This follows by construction and by the fact that two type I Reidermeister moves \cite{Simon2023} (type I') relate regular isotopic knots, as is shown in Fig.~\ref{Unknot definite causal order}.
    \begin{figure}[h!]
        \centering
        \begin{tikzpicture}[x=0.75pt,y=0.75pt,yscale=-1,xscale=1]

\draw    (281.79,19.84) -- (399.66,19.16) ;
\draw [shift={(345.73,19.47)}, rotate = 179.67] [fill={rgb, 255:red, 0; green, 0; blue, 0 }  ][line width=0.08]  [draw opacity=0] (8.93,-4.29) -- (0,0) -- (8.93,4.29) -- cycle    ;
\draw    (280.23,68.42) -- (281.79,19.84) ;
\draw [shift={(281.17,39.13)}, rotate = 91.84] [fill={rgb, 255:red, 0; green, 0; blue, 0 }  ][line width=0.08]  [draw opacity=0] (8.93,-4.29) -- (0,0) -- (8.93,4.29) -- cycle    ;
\draw    (399.66,19.16) -- (398.1,68.42) ;
\draw [shift={(398.72,48.79)}, rotate = 271.82] [fill={rgb, 255:red, 0; green, 0; blue, 0 }  ][line width=0.08]  [draw opacity=0] (8.93,-4.29) -- (0,0) -- (8.93,4.29) -- cycle    ;
\draw  [dash pattern={on 0.84pt off 2.51pt}]  (398.1,68.42) .. controls (383.27,50.54) and (384.83,36.37) .. (399.66,19.16) ;
\draw [shift={(388.49,38.55)}, rotate = 97.3] [fill={rgb, 255:red, 0; green, 0; blue, 0 }  ][line width=0.08]  [draw opacity=0] (8.93,-4.29) -- (0,0) -- (8.93,4.29) -- cycle    ;
\draw  [dash pattern={on 4.5pt off 4.5pt}]  (281.79,19.84) .. controls (299.75,30.29) and (290.38,71.46) .. (280.23,68.42) ;
\draw [shift={(291.42,50.19)}, rotate = 274.65] [fill={rgb, 255:red, 0; green, 0; blue, 0 }  ][line width=0.08]  [draw opacity=0] (8.93,-4.29) -- (0,0) -- (8.93,4.29) -- cycle    ;
\draw  [dash pattern={on 0.84pt off 2.51pt}]  (399.66,19.16) .. controls (367,45.75) and (306,31.75) .. (281.79,19.84) ;
\draw [shift={(336.22,33.83)}, rotate = 3.62] [fill={rgb, 255:red, 0; green, 0; blue, 0 }  ][line width=0.08]  [draw opacity=0] (8.93,-4.29) -- (0,0) -- (8.93,4.29) -- cycle    ;
\draw    (344.02,83.26) -- (344.02,103.58)(341.02,83.26) -- (341.02,103.58) ;
\draw [shift={(342.52,111.58)}, rotate = 270] [color={rgb, 255:red, 0; green, 0; blue, 0 }  ][line width=0.75]    (10.93,-3.29) .. controls (6.95,-1.4) and (3.31,-0.3) .. (0,0) .. controls (3.31,0.3) and (6.95,1.4) .. (10.93,3.29)   ;
\draw    (417.5,140.25) .. controls (433.1,108.65) and (398.5,138.75) .. (393.5,149.25) .. controls (388.5,159.75) and (326,164.25) .. (296,134.75) ;
\draw [shift={(406.92,133.87)}, rotate = 321.14] [fill={rgb, 255:red, 0; green, 0; blue, 0 }  ][line width=0.08]  [draw opacity=0] (8.93,-4.29) -- (0,0) -- (8.93,4.29) -- cycle    ;
\draw [shift={(338.09,154.95)}, rotate = 9.72] [fill={rgb, 255:red, 0; green, 0; blue, 0 }  ][line width=0.08]  [draw opacity=0] (8.93,-4.29) -- (0,0) -- (8.93,4.29) -- cycle    ;
\draw    (262.33,145.67) .. controls (260.33,121) and (370.6,138.85) .. (395,139.25) ;
\draw [shift={(331.65,133.82)}, rotate = 183.72] [fill={rgb, 255:red, 0; green, 0; blue, 0 }  ][line width=0.08]  [draw opacity=0] (8.93,-4.29) -- (0,0) -- (8.93,4.29) -- cycle    ;
\draw    (349.75,192.33) -- (349.75,212.65)(346.75,192.33) -- (346.75,212.65) ;
\draw [shift={(348.25,220.65)}, rotate = 270] [color={rgb, 255:red, 0; green, 0; blue, 0 }  ][line width=0.75]    (10.93,-3.29) .. controls (6.95,-1.4) and (3.31,-0.3) .. (0,0) .. controls (3.31,0.3) and (6.95,1.4) .. (10.93,3.29)   ;
\draw    (340.2,277.27) .. controls (286.2,280.47) and (293,226.87) .. (336.2,226.87) .. controls (379.4,226.87) and (385.8,237.27) .. (385.8,254.47) .. controls (385.8,272.06) and (357.34,276.26) .. (338.95,277.34) ;
\draw [shift={(342.16,277.13)}, rotate = 176.27] [fill={rgb, 255:red, 0; green, 0; blue, 0 }  ][line width=0.08]  [draw opacity=0] (8.93,-4.29) -- (0,0) -- (8.93,4.29) -- cycle    ;
\draw    (405,141.25) .. controls (467.67,157.92) and (376.5,239.75) .. (412.5,148.25) ;
\draw [shift={(414.9,188.67)}, rotate = 317.87] [fill={rgb, 255:red, 0; green, 0; blue, 0 }  ][line width=0.08]  [draw opacity=0] (8.93,-4.29) -- (0,0) -- (8.93,4.29) -- cycle    ;
\draw    (295,129.25) .. controls (297.5,107.75) and (269.5,111.75) .. (278.5,132.25) ;
\draw [shift={(279.62,117.29)}, rotate = 336.11] [fill={rgb, 255:red, 0; green, 0; blue, 0 }  ][line width=0.08]  [draw opacity=0] (8.93,-4.29) -- (0,0) -- (8.93,4.29) -- cycle    ;
\draw    (279,138.25) .. controls (297.5,244.25) and (257,180.25) .. (262.33,145.67) ;
\draw [shift={(276.45,195.44)}, rotate = 355.64] [fill={rgb, 255:red, 0; green, 0; blue, 0 }  ][line width=0.08]  [draw opacity=0] (8.93,-4.29) -- (0,0) -- (8.93,4.29) -- cycle    ;

\draw (268.77,4.25) node [anchor=north west][inner sep=0.75pt]   [align=left] {A};
\draw (267.99,64.3) node [anchor=north west][inner sep=0.75pt]   [align=left] {A};
\draw (399.13,4.25) node [anchor=north west][inner sep=0.75pt]   [align=left] {B};
\draw (398.35,67.68) node [anchor=north west][inner sep=0.75pt]   [align=left] {B};
\draw (351.32,84.66) node [anchor=north west][inner sep=0.75pt]   [align=left] {{\footnotesize Knot}};

\end{tikzpicture}
    \caption{Definite causal order between two events corresponds to the unknot.}
    \label{Unknot definite causal order}
    \end{figure}
\end{proof}

\begin{theorem}
    For a collection of two spacetimes $\{(\mathcal{M}_\mathcal{A},g_{\mathcal{A}}),(\mathcal{M}_\mathcal{B},g_{\mathcal{B}})\}$ with two events, the diagram is a trefoil knot if and only if the causal order is indefinite.
\end{theorem}

\begin{proof}
    This follows by construction, as can be seen in Fig.~\ref{Trefoil indefinite causal order}.
    \begin{figure}[h!]
        \centering
        \begin{tikzpicture}[x=0.75pt,y=0.75pt,yscale=-1,xscale=1]

\draw    (306.15,22.68) -- (400.99,22.09) ;
\draw [shift={(358.57,22.36)}, rotate = 179.64] [fill={rgb, 255:red, 0; green, 0; blue, 0 }  ][line width=0.08]  [draw opacity=0] (8.93,-4.29) -- (0,0) -- (8.93,4.29) -- cycle    ;
\draw    (304.89,65.14) -- (306.15,22.68) ;
\draw [shift={(305.67,38.91)}, rotate = 91.69] [fill={rgb, 255:red, 0; green, 0; blue, 0 }  ][line width=0.08]  [draw opacity=0] (8.93,-4.29) -- (0,0) -- (8.93,4.29) -- cycle    ;
\draw    (400.99,22.09) -- (399.73,65.14) ;
\draw [shift={(400.22,48.61)}, rotate = 271.67] [fill={rgb, 255:red, 0; green, 0; blue, 0 }  ][line width=0.08]  [draw opacity=0] (8.93,-4.29) -- (0,0) -- (8.93,4.29) -- cycle    ;
\draw  [dash pattern={on 0.84pt off 2.51pt}]  (399.73,65.14) -- (306.15,22.68) ;
\draw [shift={(348.39,41.84)}, rotate = 24.4] [fill={rgb, 255:red, 0; green, 0; blue, 0 }  ][line width=0.08]  [draw opacity=0] (8.93,-4.29) -- (0,0) -- (8.93,4.29) -- cycle    ;
\draw  [dash pattern={on 0.84pt off 2.51pt}]  (306.15,22.68) .. controls (335.8,36.4) and (373.4,30.4) .. (400.99,22.09) ;
\draw [shift={(358.8,30.36)}, rotate = 176.59] [fill={rgb, 255:red, 0; green, 0; blue, 0 }  ][line width=0.08]  [draw opacity=0] (8.93,-4.29) -- (0,0) -- (8.93,4.29) -- cycle    ;
\draw  [dash pattern={on 4.5pt off 4.5pt}]  (400.99,22.09) -- (304.89,65.14) ;
\draw [shift={(352.94,43.61)}, rotate = 335.87] [fill={rgb, 255:red, 0; green, 0; blue, 0 }  ][line width=0.08]  [draw opacity=0] (8.93,-4.29) -- (0,0) -- (8.93,4.29) -- cycle    ;
\draw    (350.61,75.64) -- (350.53,93.34)(347.61,75.63) -- (347.53,93.33) ;
\draw [shift={(349,101.33)}, rotate = 270.24] [color={rgb, 255:red, 0; green, 0; blue, 0 }  ][line width=0.75]    (10.93,-3.29) .. controls (6.95,-1.4) and (3.31,-0.3) .. (0,0) .. controls (3.31,0.3) and (6.95,1.4) .. (10.93,3.29)   ;
\draw    (333.53,284.57) .. controls (310.65,294.29) and (299.68,236.62) .. (328.51,232.86) .. controls (357.35,229.1) and (362.68,229.1) .. (384.31,231.92) ;
\draw [shift={(311.2,256.31)}, rotate = 87.53] [fill={rgb, 255:red, 0; green, 0; blue, 0 }  ][line width=0.08]  [draw opacity=0] (8.93,-4.29) -- (0,0) -- (8.93,4.29) -- cycle    ;
\draw [shift={(361.39,229.91)}, rotate = 178.82] [fill={rgb, 255:red, 0; green, 0; blue, 0 }  ][line width=0.08]  [draw opacity=0] (8.93,-4.29) -- (0,0) -- (8.93,4.29) -- cycle    ;
\draw    (388.7,233.49) .. controls (412.52,233.8) and (412.21,282.07) .. (387.76,282.69) .. controls (363.31,283.32) and (326.63,253.55) .. (324.75,236) ;
\draw [shift={(405.85,263.42)}, rotate = 274.73] [fill={rgb, 255:red, 0; green, 0; blue, 0 }  ][line width=0.08]  [draw opacity=0] (8.93,-4.29) -- (0,0) -- (8.93,4.29) -- cycle    ;
\draw [shift={(345.58,264.42)}, rotate = 37.36] [fill={rgb, 255:red, 0; green, 0; blue, 0 }  ][line width=0.08]  [draw opacity=0] (8.93,-4.29) -- (0,0) -- (8.93,4.29) -- cycle    ;
\draw    (323.19,230.35) .. controls (322.56,208.26) and (392.77,206.54) .. (386.82,231.92) .. controls (380.86,257.31) and (369.26,262.01) .. (357.35,270.16) ;
\draw [shift={(361.49,213.47)}, rotate = 181.16] [fill={rgb, 255:red, 0; green, 0; blue, 0 }  ][line width=0.08]  [draw opacity=0] (8.93,-4.29) -- (0,0) -- (8.93,4.29) -- cycle    ;
\draw [shift={(373.39,258.45)}, rotate = 311.38] [fill={rgb, 255:red, 0; green, 0; blue, 0 }  ][line width=0.08]  [draw opacity=0] (8.93,-4.29) -- (0,0) -- (8.93,4.29) -- cycle    ;
\draw    (352.34,273.29) -- (333.53,284.57) ;
\draw [shift={(338.65,281.5)}, rotate = 329.04] [fill={rgb, 255:red, 0; green, 0; blue, 0 }  ][line width=0.08]  [draw opacity=0] (8.93,-4.29) -- (0,0) -- (8.93,4.29) -- cycle    ;
\draw    (354.5,175.99) -- (354.58,196.41)(351.5,176.01) -- (351.58,196.42) ;
\draw [shift={(353.11,204.42)}, rotate = 269.78] [color={rgb, 255:red, 0; green, 0; blue, 0 }  ][line width=0.75]    (10.93,-3.29) .. controls (6.95,-1.4) and (3.31,-0.3) .. (0,0) .. controls (3.31,0.3) and (6.95,1.4) .. (10.93,3.29)   ;
\draw    (352.33,154.67) .. controls (280.33,174.67) and (283.67,112.67) .. (334.33,113.33) ;
\draw [shift={(297.39,141.64)}, rotate = 70.08] [fill={rgb, 255:red, 0; green, 0; blue, 0 }  ][line width=0.08]  [draw opacity=0] (8.93,-4.29) -- (0,0) -- (8.93,4.29) -- cycle    ;
\draw    (351,111.33) .. controls (387.67,110) and (422.33,109.33) .. (417,139.33) .. controls (411.67,169.33) and (373,166.67) .. (314.33,123.33) ;
\draw [shift={(397.82,113.1)}, rotate = 190.44] [fill={rgb, 255:red, 0; green, 0; blue, 0 }  ][line width=0.08]  [draw opacity=0] (8.93,-4.29) -- (0,0) -- (8.93,4.29) -- cycle    ;
\draw [shift={(362.05,151.86)}, rotate = 21.85] [fill={rgb, 255:red, 0; green, 0; blue, 0 }  ][line width=0.08]  [draw opacity=0] (8.93,-4.29) -- (0,0) -- (8.93,4.29) -- cycle    ;
\draw    (307.67,114) .. controls (316.33,78) and (353,120.67) .. (357,124.67) .. controls (361,128.67) and (382.33,128) .. (396.33,117.33) ;
\draw [shift={(336.33,105.8)}, rotate = 211.52] [fill={rgb, 255:red, 0; green, 0; blue, 0 }  ][line width=0.08]  [draw opacity=0] (8.93,-4.29) -- (0,0) -- (8.93,4.29) -- cycle    ;
\draw [shift={(382.42,124.38)}, rotate = 165.58] [fill={rgb, 255:red, 0; green, 0; blue, 0 }  ][line width=0.08]  [draw opacity=0] (8.93,-4.29) -- (0,0) -- (8.93,4.29) -- cycle    ;
\draw    (403.67,108) .. controls (403,90.67) and (427.67,96) .. (414.33,115.33) ;
\draw [shift={(418.03,102.74)}, rotate = 231.43] [fill={rgb, 255:red, 0; green, 0; blue, 0 }  ][line width=0.08]  [draw opacity=0] (8.93,-4.29) -- (0,0) -- (8.93,4.29) -- cycle    ;
\draw    (407,124) .. controls (395.67,138) and (376.33,143.33) .. (366.33,147.33) ;
\draw [shift={(383.7,140.56)}, rotate = 334.26] [fill={rgb, 255:red, 0; green, 0; blue, 0 }  ][line width=0.08]  [draw opacity=0] (8.93,-4.29) -- (0,0) -- (8.93,4.29) -- cycle    ;

\draw (294.56,6.02) node [anchor=north west][inner sep=0.75pt]   [align=left] {$\displaystyle A$};
\draw (399.39,7.4) node [anchor=north west][inner sep=0.75pt]   [align=left] {B};
\draw (396.32,59.67) node [anchor=north west][inner sep=0.75pt]   [align=left] {$\displaystyle A$};
\draw (296.38,61.65) node [anchor=north west][inner sep=0.75pt]   [align=left] {B};
\draw (356.19,78.26) node [anchor=north west][inner sep=0.75pt]   [align=left] {{\footnotesize Knot}};

\end{tikzpicture}
    \caption{Indefinite causal order between two events corresponds to the trefoil knot.}
    \label{Trefoil indefinite causal order}
    \end{figure}
\end{proof}

Thus, for the case of two events in two spacetimes in superposition, there are two possible knots: the unknot associated with definite causal order, and the trefoil knot associated with indefinite causal order.

We will use the Alexander-Conway polynomial as a knot invariant to discuss properties of causal orderings. The Alexander-Conway polynomial is a polynomial function associated with knots and links that can be can be constructed recursively from the following \textit{skein relations}:
\begin{itemize}
    \item $\nabla(\text{O}) = 1$ where $\text{O}$ denotes the unknot,
    \item $\nabla(L_{+}) = \nabla(L_{-}) + z \nabla(L_{0})$
\end{itemize}
where $L_+$, $L_-$ and $L_0$ are the crossings given in Fig.~\ref{fig:skein relations}. More concretely, starting from a given knot, one ``zooms in" on each of the crossings and considers the knots obtained by swapping said crossings for the remaining two options depicted in Fig.~\ref{fig:skein relations}. The second skein relation then tells us how to relate the Conway polynomial of the given knot to that of the modified knots. By iterating this procedure, one will eventually reach the unknot, whose Conway polynomial is simply $\nabla(\text{O}) =1$. This way, we obtain a recursive definition of the Conway polynomial of all other knots. For example, the Hopf link (two circles linked together exactly once) has Alexander-Conway polynomial $\nabla(\text{Hopf link}) = z$, while the trefoil knot has $\nabla(\text{Trefoil}) = z^2 + 1$. Let us also stress that the usefulness of the Alexander-Conway polynomial lies in the fact that it is a knot-invariant -- that is, it remains the same under deformations of the knot (Reidemeister moves) and is thus the same, independently of which two-dimensional projection we choose to represent it with. We refer the reader to \cite{Kauffman2001} for a more detailed introduction to knot invariants.

\begin{figure}[t!]
    \centering
    \begin{tikzpicture}[x=0.75pt,y=0.75pt,yscale=-0.8,xscale=0.8]
\draw    (166,198.5) -- (261.65,93.98) ;
\draw [shift={(263,92.5)}, rotate = 132.46] [color={rgb, 255:red, 0; green, 0; blue, 0 }  ][line width=0.75]    (10.93,-3.29) .. controls (6.95,-1.4) and (3.31,-0.3) .. (0,0) .. controls (3.31,0.3) and (6.95,1.4) .. (10.93,3.29)   ;
\draw    (219,148.5) -- (261,194.5) ;
\draw    (210,140) -- (168.37,95.95) ;
\draw [shift={(167,94.5)}, rotate = 46.62] [color={rgb, 255:red, 0; green, 0; blue, 0 }  ][line width=0.75]    (10.93,-3.29) .. controls (6.95,-1.4) and (3.31,-0.3) .. (0,0) .. controls (3.31,0.3) and (6.95,1.4) .. (10.93,3.29)   ;
\draw    (395,191.67) -- (305.38,97.78) ;
\draw [shift={(304,96.33)}, rotate = 46.33] [color={rgb, 255:red, 0; green, 0; blue, 0 }  ][line width=0.75]    (10.93,-3.29) .. controls (6.95,-1.4) and (3.31,-0.3) .. (0,0) .. controls (3.31,0.3) and (6.95,1.4) .. (10.93,3.29)   ;
\draw    (349,151) -- (312.33,192.33) ;
\draw    (355.33,144) -- (396.35,97.17) ;
\draw [shift={(397.67,95.67)}, rotate = 131.21] [color={rgb, 255:red, 0; green, 0; blue, 0 }  ][line width=0.75]    (10.93,-3.29) .. controls (6.95,-1.4) and (3.31,-0.3) .. (0,0) .. controls (3.31,0.3) and (6.95,1.4) .. (10.93,3.29)   ;
\draw    (440.67,192.33) .. controls (479.94,173.19) and (476.41,127.92) .. (440.11,99.2) ;
\draw [shift={(439,98.33)}, rotate = 37.52] [color={rgb, 255:red, 0; green, 0; blue, 0 }  ][line width=0.75]    (10.93,-3.29) .. controls (6.95,-1.4) and (3.31,-0.3) .. (0,0) .. controls (3.31,0.3) and (6.95,1.4) .. (10.93,3.29)   ;
\draw    (509.67,192.33) .. controls (481.43,167.38) and (477.13,117.85) .. (516.5,100.44) ;
\draw [shift={(518.33,99.67)}, rotate = 158.04] [color={rgb, 255:red, 0; green, 0; blue, 0 }  ][line width=0.75]    (10.93,-3.29) .. controls (6.95,-1.4) and (3.31,-0.3) .. (0,0) .. controls (3.31,0.3) and (6.95,1.4) .. (10.93,3.29)   ;
\draw (204,192.73) node [anchor=north west][inner sep=0.75pt]    {$L_{+}$};
\draw (348.67,193.4) node [anchor=north west][inner sep=0.75pt]    {$L_{-}$};
\draw (470,192.07) node [anchor=north west][inner sep=0.75pt]    {$L_{0}$};
\end{tikzpicture}
    \caption{\justifying The three possible crossings $L_+, L_-$, and $L_0$ of two lines in a knot, used in the recursive construction of the Alexander-Conway polynomial.}
    \label{fig:skein relations}
\end{figure}

Let us now consider all the possibilities for the case $N=3,M=2$. The resulting knots are shown in Fig.~\ref{fig:N=3,M=2}. Note that the knot \ref{ABC-ABC} for definite causal order in $N=3,M=2$ is still an unknot. Further note that the knots \ref{ABC-ACB} of ABC-ACB and \ref{ABC-BAC} of ABC-BAC are both equivalent to that for indefinite causal order for \ref{fig:(2,2) indefinite causal order} -- this just highlights the fact that appending an additional event before or after the series of events of interest does not change anything for the analysis of relative causal order. Moreover, note that knots \ref{ABC-BCA} and $\ref{ABC-CAB}$ have the same causal indefiniteness $\delta = 2$ -- it turns out that these knots are equivalent: they are both the three-twist knot ($5_2$ knot in Alexander-Briggs notation) as they have Alexander-Conway polynomial $\nabla(z) = 1+2z^2$, while they are inequivalent to knot \ref{ABC-CBA} with $\delta = 3$, which is a cinquefoil knot ($5_1$ knot in Alexander-Briggs notation, a (2,5)-torus knot) as its Alexander-Conway polynomial is $\nabla(z) = 1+3z^2 + z^4$. There is thus a clear link to be made between measures of causal order and topological invariants.

\begin{figure}[b!]
    \centering
    \subfloat[\centering ABC-ABC, \newline  $\delta = 0$ \label{ABC-ABC}]{\begin{tikzpicture}[x=0.75pt,y=0.75pt,yscale=-1,xscale=1]
\draw    (364.2,210) .. controls (310.2,213.2) and (317,159.6) .. (360.2,159.6) .. controls (403.4,159.6) and (409.8,170) .. (409.8,187.2) .. controls (409.8,203.54) and (385.25,208.33) .. (367.02,209.79) ;
\draw [shift={(364.2,210)}, rotate = 356.27] [fill={rgb, 255:red, 0; green, 0; blue, 0 }  ][line width=0.08]  [draw opacity=0] (8.93,-4.29) -- (0,0) -- (8.93,4.29) -- cycle    ;
\end{tikzpicture}}
    \subfloat[\centering ABC-ACB, \newline  $\delta = 1$ \label{ABC-ACB}]{\begin{tikzpicture}[x=0.75pt,y=0.75pt,yscale=-0.5,xscale=0.5]
\draw    (289.5,212.25) .. controls (253,227.75) and (235.5,135.75) .. (281.5,129.75) .. controls (327.5,123.75) and (336,123.75) .. (370.5,128.25) ;
\draw [shift={(253.85,170.47)}, rotate = 86.4] [fill={rgb, 255:red, 0; green, 0; blue, 0 }  ][line width=0.08]  [draw opacity=0] (8.93,-4.29) -- (0,0) -- (8.93,4.29) -- cycle    ;
\draw [shift={(331.04,125.05)}, rotate = 178.4] [fill={rgb, 255:red, 0; green, 0; blue, 0 }  ][line width=0.08]  [draw opacity=0] (8.93,-4.29) -- (0,0) -- (8.93,4.29) -- cycle    ;
\draw    (377.5,130.75) .. controls (415.5,131.25) and (415,208.25) .. (376,209.25) .. controls (337,210.25) and (278.5,162.75) .. (275.5,134.75) ;
\draw [shift={(405.26,175.63)}, rotate = 273.58] [fill={rgb, 255:red, 0; green, 0; blue, 0 }  ][line width=0.08]  [draw opacity=0] (8.93,-4.29) -- (0,0) -- (8.93,4.29) -- cycle    ;
\draw [shift={(311.06,181.96)}, rotate = 36.98] [fill={rgb, 255:red, 0; green, 0; blue, 0 }  ][line width=0.08]  [draw opacity=0] (8.93,-4.29) -- (0,0) -- (8.93,4.29) -- cycle    ;
\draw    (273,125.75) .. controls (272,90.5) and (384,87.75) .. (374.5,128.25) .. controls (365,168.75) and (346.5,176.25) .. (327.5,189.25) ;
\draw [shift={(331.72,98.69)}, rotate = 180.75] [fill={rgb, 255:red, 0; green, 0; blue, 0 }  ][line width=0.08]  [draw opacity=0] (8.93,-4.29) -- (0,0) -- (8.93,4.29) -- cycle    ;
\draw [shift={(355.05,168.54)}, rotate = 310.28] [fill={rgb, 255:red, 0; green, 0; blue, 0 }  ][line width=0.08]  [draw opacity=0] (8.93,-4.29) -- (0,0) -- (8.93,4.29) -- cycle    ;
\draw    (319.5,194.25) -- (289.5,212.25) ;
\draw [shift={(300.21,205.82)}, rotate = 329.04] [fill={rgb, 255:red, 0; green, 0; blue, 0 }  ][line width=0.08]  [draw opacity=0] (8.93,-4.29) -- (0,0) -- (8.93,4.29) -- cycle    ;\end{tikzpicture}}
    \subfloat[\centering ABC-BAC, $\delta = 1$ \label{ABC-BAC}]{\begin{tikzpicture}[x=0.75pt,y=0.75pt,yscale=-0.5,xscale=0.5]
\draw    (289.5,212.25) .. controls (253,227.75) and (235.5,135.75) .. (281.5,129.75) .. controls (327.5,123.75) and (336,123.75) .. (370.5,128.25) ;
\draw [shift={(253.85,170.47)}, rotate = 86.4] [fill={rgb, 255:red, 0; green, 0; blue, 0 }  ][line width=0.08]  [draw opacity=0] (8.93,-4.29) -- (0,0) -- (8.93,4.29) -- cycle    ;
\draw [shift={(331.04,125.05)}, rotate = 178.4] [fill={rgb, 255:red, 0; green, 0; blue, 0 }  ][line width=0.08]  [draw opacity=0] (8.93,-4.29) -- (0,0) -- (8.93,4.29) -- cycle    ;
\draw    (377.5,130.75) .. controls (415.5,131.25) and (415,208.25) .. (376,209.25) .. controls (337,210.25) and (278.5,162.75) .. (275.5,134.75) ;
\draw [shift={(405.26,175.63)}, rotate = 273.58] [fill={rgb, 255:red, 0; green, 0; blue, 0 }  ][line width=0.08]  [draw opacity=0] (8.93,-4.29) -- (0,0) -- (8.93,4.29) -- cycle    ;
\draw [shift={(311.06,181.96)}, rotate = 36.98] [fill={rgb, 255:red, 0; green, 0; blue, 0 }  ][line width=0.08]  [draw opacity=0] (8.93,-4.29) -- (0,0) -- (8.93,4.29) -- cycle    ;
\draw    (273,125.75) .. controls (272,90.5) and (384,87.75) .. (374.5,128.25) .. controls (365,168.75) and (346.5,176.25) .. (327.5,189.25) ;
\draw [shift={(331.72,98.69)}, rotate = 180.75] [fill={rgb, 255:red, 0; green, 0; blue, 0 }  ][line width=0.08]  [draw opacity=0] (8.93,-4.29) -- (0,0) -- (8.93,4.29) -- cycle    ;
\draw [shift={(355.05,168.54)}, rotate = 310.28] [fill={rgb, 255:red, 0; green, 0; blue, 0 }  ][line width=0.08]  [draw opacity=0] (8.93,-4.29) -- (0,0) -- (8.93,4.29) -- cycle    ;
\draw    (319.5,194.25) -- (289.5,212.25) ;
\draw [shift={(300.21,205.82)}, rotate = 329.04] [fill={rgb, 255:red, 0; green, 0; blue, 0 }  ][line width=0.08]  [draw opacity=0] (8.93,-4.29) -- (0,0) -- (8.93,4.29) -- cycle    ;\end{tikzpicture}} \\
    \subfloat[\centering ABC-BCA, \newline $\delta = 2$ \label{ABC-BCA}]{\includegraphics[scale=0.11]{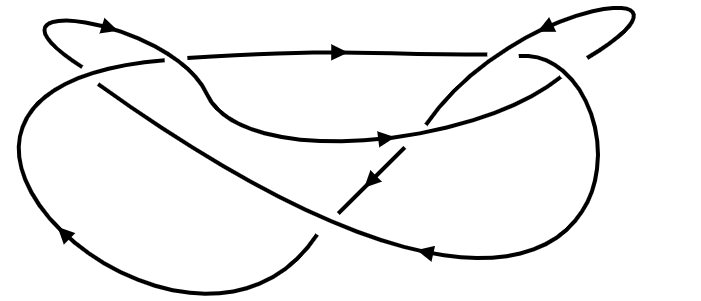}
}
    \subfloat[\centering ABC-CAB, \newline $\delta = 2$ \label{ABC-CAB}]{\includegraphics[scale=0.16]{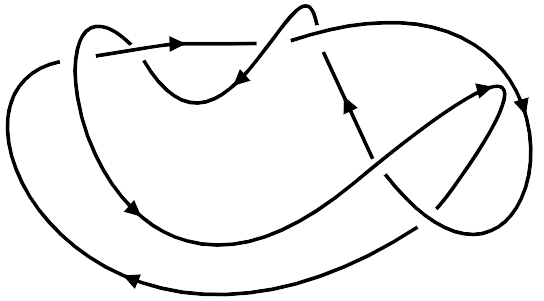}
}
    \subfloat[\centering ABC-CBA,\newline $\delta = 3$ \label{ABC-CBA}]{\begin{tikzpicture}[x=0.75pt,y=0.75pt,yscale=-0.5,xscale=0.5]

\draw    (288.33,196.33) .. controls (248.5,193.25) and (251,115.75) .. (271,111.75) .. controls (291,107.75) and (293,109.75) .. (308,110.25) ;
\draw [shift={(257.18,153.78)}, rotate = 83.06] [fill={rgb, 255:red, 0; green, 0; blue, 0 }  ][line width=0.08]  [draw opacity=0] (8.93,-4.29) -- (0,0) -- (8.93,4.29) -- cycle    ;
\draw [shift={(294.54,109.43)}, rotate = 180.77] [fill={rgb, 255:red, 0; green, 0; blue, 0 }  ][line width=0.08]  [draw opacity=0] (8.93,-4.29) -- (0,0) -- (8.93,4.29) -- cycle    ;
\draw    (320,108.75) -- (390.33,109) ;
\draw [shift={(360.17,108.89)}, rotate = 180.2] [fill={rgb, 255:red, 0; green, 0; blue, 0 }  ][line width=0.08]  [draw opacity=0] (8.93,-4.29) -- (0,0) -- (8.93,4.29) -- cycle    ;
\draw    (397.67,109) .. controls (443,109.67) and (443,194.33) .. (399.67,195.67) .. controls (356.33,197) and (275.67,141) .. (281.67,115) ;
\draw [shift={(431.76,157.21)}, rotate = 271.11] [fill={rgb, 255:red, 0; green, 0; blue, 0 }  ][line width=0.08]  [draw opacity=0] (8.93,-4.29) -- (0,0) -- (8.93,4.29) -- cycle    ;
\draw [shift={(324.42,167.11)}, rotate = 32.98] [fill={rgb, 255:red, 0; green, 0; blue, 0 }  ][line width=0.08]  [draw opacity=0] (8.93,-4.29) -- (0,0) -- (8.93,4.29) -- cycle    ;
\draw    (283,107) .. controls (301.67,81.67) and (310.17,99.08) .. (315.5,111.75) .. controls (320.83,124.42) and (360.33,127.67) .. (377,113) ;
\draw [shift={(305.88,96.2)}, rotate = 201.69] [fill={rgb, 255:red, 0; green, 0; blue, 0 }  ][line width=0.08]  [draw opacity=0] (8.93,-4.29) -- (0,0) -- (8.93,4.29) -- cycle    ;
\draw [shift={(350.85,122.46)}, rotate = 177.83] [fill={rgb, 255:red, 0; green, 0; blue, 0 }  ][line width=0.08]  [draw opacity=0] (8.93,-4.29) -- (0,0) -- (8.93,4.29) -- cycle    ;
\draw    (380,103.67) .. controls (385.67,93) and (398.33,86.33) .. (393.67,111.67) .. controls (389,137) and (360.33,158.33) .. (343,173.67) ;
\draw [shift={(394.54,99.68)}, rotate = 243.33] [fill={rgb, 255:red, 0; green, 0; blue, 0 }  ][line width=0.08]  [draw opacity=0] (8.93,-4.29) -- (0,0) -- (8.93,4.29) -- cycle    ;
\draw [shift={(369.45,150.11)}, rotate = 314.34] [fill={rgb, 255:red, 0; green, 0; blue, 0 }  ][line width=0.08]  [draw opacity=0] (8.93,-4.29) -- (0,0) -- (8.93,4.29) -- cycle    ;
\draw    (337,178.33) .. controls (317.67,191.67) and (306.33,196.33) .. (288.33,196.33) ;
\draw [shift={(309.25,193.35)}, rotate = 339.42] [fill={rgb, 255:red, 0; green, 0; blue, 0 }  ][line width=0.08]  [draw opacity=0] (8.93,-4.29) -- (0,0) -- (8.93,4.29) -- cycle    ;\end{tikzpicture}}
    \caption{\justifying Fine-grained knot classification of the causal order for $N=3,M=2$.}
    \label{fig:N=3,M=2}
\end{figure}

\subsubsection{Coarse-grained knot construction}

One may argue that repeated subsequences of events do not characterise the indefiniteness of the situation, e.g.~ABCDEFG-GBCDEFA is completely analogous to AHG-GHA from a topological point of view, where H=BCDEF. In practice, this means we need to coarse-grain the sequences of events: repeated subsequences can just be thought of as one big event as far as indefinite causal order goes. For example, if the event ``Alice looks at her microscope in the lab" occurs in both spacetimes, it makes sense to talk about this macroscopic event rather than (``photon 1 reaches Alice's eye" $\prec$ ``Alice's eye molecules react and send a signal to her brain" etc.). In practice, this is reached with the following coarse-grained ordered set:
\begin{widetext}
    \begin{equation}
    \slashed{\mathcal{S}}_\mathcal{AB} := \{(\mathcal{E}_i^{\mathcal{A}},\mathcal{E}_i^{\mathcal{B}}) \pprec (\mathcal{E}_j^{\mathcal{A}},\mathcal{E}_j^{\mathcal{B}}) \pprec ...: (\mathcal{E}_i^{\mathcal{A}},\mathcal{E}_i^{\mathcal{B}}) \in \mathcal{S}_\mathcal{AB} \text{ } \& \text{ } (\forall i>1, \forall j\neq i, \{\mathcal{E}_{i-1}^{\mathcal{B}}\prec \mathcal{E}_{i}^{\mathcal{B}}\} \neq \{\mathcal{E}_{j-1}^{\mathcal{A}}\prec \mathcal{E}_{j}^{\mathcal{A}}\})\}.
\end{equation}
\end{widetext}

Elements of $\slashed{\mathcal{S}}_\mathcal{AB}$ do not share two consecutive elements of $\mathcal{S}_\mathcal{AB}$, i.e.~we coarse-grained over the redundant subsequences of events. For example, $\text{ABCD-CABD} \equiv \text{ECD-CED} \equiv \text{EC-CE} \equiv \text{AB-BA}$, so the resulting knot ought to be a trefoil (and indeed now is). In this example, we call E (and C) a \textit{coarse-grained event}.

\begin{definition}[Effective causal order]
    The \textit{effective causal order} of \textit{coarse-grained events} $\mathcal{E}_{a}^\mathcal{X}$ and $\mathcal{E}_{b}^\mathcal{X}$ of spacetime $(\mathcal{M}_\mathcal{X},g_\mathcal{X})$ are defined as
    \begin{equation}
        s^{\mathcal{M_X}}_{\text{eff},ab} := \begin{cases}
        +1 \text{ if } \mathcal{E}_{a}^{\mathcal{X}}\prec \mathcal{E}_{b}^{\mathcal{X}} \\
        -1 \text{ otherwise}
        \end{cases}
    \end{equation}
    respectively, with convention that for $\mathcal{X} = \mathcal{A},\mathcal{B}$, $[s^{\mathcal{M_X}}_{\text{eff}}] \equiv \begin{pmatrix}
        0 && 1 \\ -1 && 0
    \end{pmatrix}$ if $\slashed{\mathcal{S}}_\mathcal{AB} = \varnothing$, which corresponds to the case of definite causal order for two coarse-grained events.
\end{definition}
All quantifiers of causal order introduced in Sec.~\ref{sec:quantifiers} straightforwardly follow for the effective causal order, and the diagrammatic and knot constructions can be performed straightforwardly in this case. For example, the coarse-grained knot classification of the causal order for $N=3,M=2$ is presented in Fig.~$\ref{fig:N=3,M=2 CG}$. We now see that cases \ref{ABC-ACB CG}, \ref{ABC-BAC CG}, \ref{ABC-BCA CG}, and \ref{ABC-CAB CG} are all equivalent after coarse-graining: they are all trefoil knots. 

\begin{figure}[b!]
    \centering
    \subfloat[\centering ABC-ABC, $\delta_{\text{eff}} = 0$ \label{ABC-ABC CG}]{\begin{tikzpicture}[x=0.75pt,y=0.75pt,yscale=-1,xscale=1]
\draw    (364.2,210) .. controls (310.2,213.2) and (317,159.6) .. (360.2,159.6) .. controls (403.4,159.6) and (409.8,170) .. (409.8,187.2) .. controls (409.8,203.54) and (385.25,208.33) .. (367.02,209.79) ;
\draw [shift={(364.2,210)}, rotate = 356.27] [fill={rgb, 255:red, 0; green, 0; blue, 0 }  ][line width=0.08]  [draw opacity=0] (8.93,-4.29) -- (0,0) -- (8.93,4.29) -- cycle    ;
\end{tikzpicture}}
    \subfloat[\centering ABC-ACB, $\delta_{\text{eff}} = 1$ \label{ABC-ACB CG}]{\begin{tikzpicture}[x=0.75pt,y=0.75pt,yscale=-0.5,xscale=0.5]
\draw    (289.5,212.25) .. controls (253,227.75) and (235.5,135.75) .. (281.5,129.75) .. controls (327.5,123.75) and (336,123.75) .. (370.5,128.25) ;
\draw [shift={(253.85,170.47)}, rotate = 86.4] [fill={rgb, 255:red, 0; green, 0; blue, 0 }  ][line width=0.08]  [draw opacity=0] (8.93,-4.29) -- (0,0) -- (8.93,4.29) -- cycle    ;
\draw [shift={(331.04,125.05)}, rotate = 178.4] [fill={rgb, 255:red, 0; green, 0; blue, 0 }  ][line width=0.08]  [draw opacity=0] (8.93,-4.29) -- (0,0) -- (8.93,4.29) -- cycle    ;
\draw    (377.5,130.75) .. controls (415.5,131.25) and (415,208.25) .. (376,209.25) .. controls (337,210.25) and (278.5,162.75) .. (275.5,134.75) ;
\draw [shift={(405.26,175.63)}, rotate = 273.58] [fill={rgb, 255:red, 0; green, 0; blue, 0 }  ][line width=0.08]  [draw opacity=0] (8.93,-4.29) -- (0,0) -- (8.93,4.29) -- cycle    ;
\draw [shift={(311.06,181.96)}, rotate = 36.98] [fill={rgb, 255:red, 0; green, 0; blue, 0 }  ][line width=0.08]  [draw opacity=0] (8.93,-4.29) -- (0,0) -- (8.93,4.29) -- cycle    ;
\draw    (273,125.75) .. controls (272,90.5) and (384,87.75) .. (374.5,128.25) .. controls (365,168.75) and (346.5,176.25) .. (327.5,189.25) ;
\draw [shift={(331.72,98.69)}, rotate = 180.75] [fill={rgb, 255:red, 0; green, 0; blue, 0 }  ][line width=0.08]  [draw opacity=0] (8.93,-4.29) -- (0,0) -- (8.93,4.29) -- cycle    ;
\draw [shift={(355.05,168.54)}, rotate = 310.28] [fill={rgb, 255:red, 0; green, 0; blue, 0 }  ][line width=0.08]  [draw opacity=0] (8.93,-4.29) -- (0,0) -- (8.93,4.29) -- cycle    ;
\draw    (319.5,194.25) -- (289.5,212.25) ;
\draw [shift={(300.21,205.82)}, rotate = 329.04] [fill={rgb, 255:red, 0; green, 0; blue, 0 }  ][line width=0.08]  [draw opacity=0] (8.93,-4.29) -- (0,0) -- (8.93,4.29) -- cycle    ;\end{tikzpicture}}
    \subfloat[\centering ABC-BAC, $\delta_{\text{eff}} = 1$ \label{ABC-BAC CG}]{\begin{tikzpicture}[x=0.75pt,y=0.75pt,yscale=-0.5,xscale=0.5]
\draw    (289.5,212.25) .. controls (253,227.75) and (235.5,135.75) .. (281.5,129.75) .. controls (327.5,123.75) and (336,123.75) .. (370.5,128.25) ;
\draw [shift={(253.85,170.47)}, rotate = 86.4] [fill={rgb, 255:red, 0; green, 0; blue, 0 }  ][line width=0.08]  [draw opacity=0] (8.93,-4.29) -- (0,0) -- (8.93,4.29) -- cycle    ;
\draw [shift={(331.04,125.05)}, rotate = 178.4] [fill={rgb, 255:red, 0; green, 0; blue, 0 }  ][line width=0.08]  [draw opacity=0] (8.93,-4.29) -- (0,0) -- (8.93,4.29) -- cycle    ;
\draw    (377.5,130.75) .. controls (415.5,131.25) and (415,208.25) .. (376,209.25) .. controls (337,210.25) and (278.5,162.75) .. (275.5,134.75) ;
\draw [shift={(405.26,175.63)}, rotate = 273.58] [fill={rgb, 255:red, 0; green, 0; blue, 0 }  ][line width=0.08]  [draw opacity=0] (8.93,-4.29) -- (0,0) -- (8.93,4.29) -- cycle    ;
\draw [shift={(311.06,181.96)}, rotate = 36.98] [fill={rgb, 255:red, 0; green, 0; blue, 0 }  ][line width=0.08]  [draw opacity=0] (8.93,-4.29) -- (0,0) -- (8.93,4.29) -- cycle    ;
\draw    (273,125.75) .. controls (272,90.5) and (384,87.75) .. (374.5,128.25) .. controls (365,168.75) and (346.5,176.25) .. (327.5,189.25) ;
\draw [shift={(331.72,98.69)}, rotate = 180.75] [fill={rgb, 255:red, 0; green, 0; blue, 0 }  ][line width=0.08]  [draw opacity=0] (8.93,-4.29) -- (0,0) -- (8.93,4.29) -- cycle    ;
\draw [shift={(355.05,168.54)}, rotate = 310.28] [fill={rgb, 255:red, 0; green, 0; blue, 0 }  ][line width=0.08]  [draw opacity=0] (8.93,-4.29) -- (0,0) -- (8.93,4.29) -- cycle    ;
\draw    (319.5,194.25) -- (289.5,212.25) ;
\draw [shift={(300.21,205.82)}, rotate = 329.04] [fill={rgb, 255:red, 0; green, 0; blue, 0 }  ][line width=0.08]  [draw opacity=0] (8.93,-4.29) -- (0,0) -- (8.93,4.29) -- cycle    ;\end{tikzpicture}} \\
    \subfloat[\centering ABC-BCA, $\delta_{\text{eff}} = 1$ \label{ABC-BCA CG}]{\begin{tikzpicture}[x=0.75pt,y=0.75pt,yscale=-0.5,xscale=0.5]
\draw    (289.5,212.25) .. controls (253,227.75) and (235.5,135.75) .. (281.5,129.75) .. controls (327.5,123.75) and (336,123.75) .. (370.5,128.25) ;
\draw [shift={(253.85,170.47)}, rotate = 86.4] [fill={rgb, 255:red, 0; green, 0; blue, 0 }  ][line width=0.08]  [draw opacity=0] (8.93,-4.29) -- (0,0) -- (8.93,4.29) -- cycle    ;
\draw [shift={(331.04,125.05)}, rotate = 178.4] [fill={rgb, 255:red, 0; green, 0; blue, 0 }  ][line width=0.08]  [draw opacity=0] (8.93,-4.29) -- (0,0) -- (8.93,4.29) -- cycle    ;
\draw    (377.5,130.75) .. controls (415.5,131.25) and (415,208.25) .. (376,209.25) .. controls (337,210.25) and (278.5,162.75) .. (275.5,134.75) ;
\draw [shift={(405.26,175.63)}, rotate = 273.58] [fill={rgb, 255:red, 0; green, 0; blue, 0 }  ][line width=0.08]  [draw opacity=0] (8.93,-4.29) -- (0,0) -- (8.93,4.29) -- cycle    ;
\draw [shift={(311.06,181.96)}, rotate = 36.98] [fill={rgb, 255:red, 0; green, 0; blue, 0 }  ][line width=0.08]  [draw opacity=0] (8.93,-4.29) -- (0,0) -- (8.93,4.29) -- cycle    ;
\draw    (273,125.75) .. controls (272,90.5) and (384,87.75) .. (374.5,128.25) .. controls (365,168.75) and (346.5,176.25) .. (327.5,189.25) ;
\draw [shift={(331.72,98.69)}, rotate = 180.75] [fill={rgb, 255:red, 0; green, 0; blue, 0 }  ][line width=0.08]  [draw opacity=0] (8.93,-4.29) -- (0,0) -- (8.93,4.29) -- cycle    ;
\draw [shift={(355.05,168.54)}, rotate = 310.28] [fill={rgb, 255:red, 0; green, 0; blue, 0 }  ][line width=0.08]  [draw opacity=0] (8.93,-4.29) -- (0,0) -- (8.93,4.29) -- cycle    ;
\draw    (319.5,194.25) -- (289.5,212.25) ;
\draw [shift={(300.21,205.82)}, rotate = 329.04] [fill={rgb, 255:red, 0; green, 0; blue, 0 }  ][line width=0.08]  [draw opacity=0] (8.93,-4.29) -- (0,0) -- (8.93,4.29) -- cycle    ;\end{tikzpicture}}
    \subfloat[\centering ABC-CAB, $\delta_{\text{eff}} = 1$ \label{ABC-CAB CG}]{\begin{tikzpicture}[x=0.75pt,y=0.75pt,yscale=-0.5,xscale=0.5]
\draw    (289.5,212.25) .. controls (253,227.75) and (235.5,135.75) .. (281.5,129.75) .. controls (327.5,123.75) and (336,123.75) .. (370.5,128.25) ;
\draw [shift={(253.85,170.47)}, rotate = 86.4] [fill={rgb, 255:red, 0; green, 0; blue, 0 }  ][line width=0.08]  [draw opacity=0] (8.93,-4.29) -- (0,0) -- (8.93,4.29) -- cycle    ;
\draw [shift={(331.04,125.05)}, rotate = 178.4] [fill={rgb, 255:red, 0; green, 0; blue, 0 }  ][line width=0.08]  [draw opacity=0] (8.93,-4.29) -- (0,0) -- (8.93,4.29) -- cycle    ;
\draw    (377.5,130.75) .. controls (415.5,131.25) and (415,208.25) .. (376,209.25) .. controls (337,210.25) and (278.5,162.75) .. (275.5,134.75) ;
\draw [shift={(405.26,175.63)}, rotate = 273.58] [fill={rgb, 255:red, 0; green, 0; blue, 0 }  ][line width=0.08]  [draw opacity=0] (8.93,-4.29) -- (0,0) -- (8.93,4.29) -- cycle    ;
\draw [shift={(311.06,181.96)}, rotate = 36.98] [fill={rgb, 255:red, 0; green, 0; blue, 0 }  ][line width=0.08]  [draw opacity=0] (8.93,-4.29) -- (0,0) -- (8.93,4.29) -- cycle    ;
\draw    (273,125.75) .. controls (272,90.5) and (384,87.75) .. (374.5,128.25) .. controls (365,168.75) and (346.5,176.25) .. (327.5,189.25) ;
\draw [shift={(331.72,98.69)}, rotate = 180.75] [fill={rgb, 255:red, 0; green, 0; blue, 0 }  ][line width=0.08]  [draw opacity=0] (8.93,-4.29) -- (0,0) -- (8.93,4.29) -- cycle    ;
\draw [shift={(355.05,168.54)}, rotate = 310.28] [fill={rgb, 255:red, 0; green, 0; blue, 0 }  ][line width=0.08]  [draw opacity=0] (8.93,-4.29) -- (0,0) -- (8.93,4.29) -- cycle    ;
\draw    (319.5,194.25) -- (289.5,212.25) ;
\draw [shift={(300.21,205.82)}, rotate = 329.04] [fill={rgb, 255:red, 0; green, 0; blue, 0 }  ][line width=0.08]  [draw opacity=0] (8.93,-4.29) -- (0,0) -- (8.93,4.29) -- cycle    ;\end{tikzpicture}}
    \subfloat[\centering ABC-CBA,\newline $\delta_{\text{eff}} = 3$ \label{ABC-CBA CG}]{\begin{tikzpicture}[x=0.75pt,y=0.75pt,yscale=-0.5,xscale=0.5]
\draw    (288.33,196.33) .. controls (235.67,195.67) and (232.33,109.67) .. (294.33,109) .. controls (356.33,108.33) and (313.67,109) .. (343.67,109) ;
\draw [shift={(249.08,144.9)}, rotate = 96.61] [fill={rgb, 255:red, 0; green, 0; blue, 0 }  ][line width=0.08]  [draw opacity=0] (8.93,-4.29) -- (0,0) -- (8.93,4.29) -- cycle    ;
\draw [shift={(324.29,108.72)}, rotate = 179.58] [fill={rgb, 255:red, 0; green, 0; blue, 0 }  ][line width=0.08]  [draw opacity=0] (8.93,-4.29) -- (0,0) -- (8.93,4.29) -- cycle    ;
\draw    (350.67,109.33) -- (390.33,109) ;
\draw [shift={(375.5,109.12)}, rotate = 179.52] [fill={rgb, 255:red, 0; green, 0; blue, 0 }  ][line width=0.08]  [draw opacity=0] (8.93,-4.29) -- (0,0) -- (8.93,4.29) -- cycle    ;
\draw    (397.67,109) .. controls (443,109.67) and (443,194.33) .. (399.67,195.67) .. controls (356.33,197) and (275.67,141) .. (281.67,115) ;
\draw [shift={(431.76,157.21)}, rotate = 271.11] [fill={rgb, 255:red, 0; green, 0; blue, 0 }  ][line width=0.08]  [draw opacity=0] (8.93,-4.29) -- (0,0) -- (8.93,4.29) -- cycle    ;
\draw [shift={(324.42,167.11)}, rotate = 32.98] [fill={rgb, 255:red, 0; green, 0; blue, 0 }  ][line width=0.08]  [draw opacity=0] (8.93,-4.29) -- (0,0) -- (8.93,4.29) -- cycle    ;
\draw    (283,107) .. controls (301.67,81.67) and (341,93) .. (346.33,105.67) .. controls (351.67,118.33) and (360.33,127.67) .. (377,113) ;
\draw [shift={(319.67,92.08)}, rotate = 182.22] [fill={rgb, 255:red, 0; green, 0; blue, 0 }  ][line width=0.08]  [draw opacity=0] (8.93,-4.29) -- (0,0) -- (8.93,4.29) -- cycle    ;
\draw [shift={(364.22,120.24)}, rotate = 183.54] [fill={rgb, 255:red, 0; green, 0; blue, 0 }  ][line width=0.08]  [draw opacity=0] (8.93,-4.29) -- (0,0) -- (8.93,4.29) -- cycle    ;
\draw    (380,103.67) .. controls (385.67,93) and (398.33,86.33) .. (393.67,111.67) .. controls (389,137) and (360.33,158.33) .. (343,173.67) ;
\draw [shift={(394.54,99.68)}, rotate = 243.33] [fill={rgb, 255:red, 0; green, 0; blue, 0 }  ][line width=0.08]  [draw opacity=0] (8.93,-4.29) -- (0,0) -- (8.93,4.29) -- cycle    ;
\draw [shift={(369.45,150.11)}, rotate = 314.34] [fill={rgb, 255:red, 0; green, 0; blue, 0 }  ][line width=0.08]  [draw opacity=0] (8.93,-4.29) -- (0,0) -- (8.93,4.29) -- cycle    ;
\draw    (337,178.33) .. controls (317.67,191.67) and (306.33,196.33) .. (288.33,196.33) ;
\draw [shift={(309.25,193.35)}, rotate = 339.42] [fill={rgb, 255:red, 0; green, 0; blue, 0 }  ][line width=0.08]  [draw opacity=0] (8.93,-4.29) -- (0,0) -- (8.93,4.29) -- cycle    ; \end{tikzpicture}}
    \caption{\justifying Coarse-grained knot classification of the causal order for $N=3,M=2$.}
    \label{fig:N=3,M=2 CG}
\end{figure}

\begin{definition}[Coarse-grainability]
    A pair of causal sequences of events is called \textit{coarse-grainable} if $S_{\mathcal{AB}} \neq \slashed{S}_{\mathcal{AB}}$.
\end{definition}

A related but inequivalent notion is that of reducibility. Consider a pair of causal sequences containing groups (that we shall call \textit{causal subsequences}) of events -- this pair of causal sequences are thus concatenations of causal subsequences. Within each subsequence, the causal order between events may be definite or indefinite. An example of this is ABCD-BADC, which is the concatenation of AB-BA with CD-DC.

\begin{definition}[Irreducible causal orderings]
    A pair of causal sequences of events is called \textit{reducible} if they contain non-trivial (i.e.~not themselves nor the empty set) causal subsequences, and \textit{irreducible} otherwise.
\end{definition}

Coarse-graining is, in a sense, equivalent to renaming events and thus does not bear any effect on the subsequent theorems -- it reflects some sort of redundancy in the sequences, and simply encapsulates different ways of looking at the same problem. On the other hand, reducibility is less trivial: if one concatenates subsequences, one expects some sort of additivity on a topological level. This is expressed in the following lemma and subsequent theorems.

\begin{lemma}
    \label{Prime knots irreducible subsequences}
    Knots associated to irreducible causal orderings are either prime knots or the unknot.
\end{lemma}

\begin{proof}
    Suppose that the knot is not prime nor the unknot, i.e.~it is the knot sum of several knots. By construction, this implies that there exist subsequences of causal orders that are delimited by such a composition, in the sense that there will be no crossing between different parts of the knot during the construction. This is shown in Fig.~\ref{fig:reducible sequences} for the example ABCD-BADC.
    \begin{figure*}
        \centering
        \begin{tikzpicture}[x=0.75pt,y=0.75pt,yscale=-1,xscale=1]

\draw    (48.15,31) -- (142.99,30.41) ;
\draw [shift={(100.57,30.67)}, rotate = 179.64] [fill={rgb, 255:red, 0; green, 0; blue, 0 }  ][line width=0.08]  [draw opacity=0] (8.93,-4.29) -- (0,0) -- (8.93,4.29) -- cycle    ;
\draw    (46.89,73.45) -- (48.15,31) ;
\draw [shift={(47.67,47.23)}, rotate = 91.69] [fill={rgb, 255:red, 0; green, 0; blue, 0 }  ][line width=0.08]  [draw opacity=0] (8.93,-4.29) -- (0,0) -- (8.93,4.29) -- cycle    ;
\draw    (142.99,30.41) -- (141.73,73.45) ;
\draw [shift={(142.22,56.93)}, rotate = 271.67] [fill={rgb, 255:red, 0; green, 0; blue, 0 }  ][line width=0.08]  [draw opacity=0] (8.93,-4.29) -- (0,0) -- (8.93,4.29) -- cycle    ;
\draw  [dash pattern={on 0.84pt off 2.51pt}]  (141.73,73.45) -- (48.15,31) ;
\draw [shift={(90.39,50.16)}, rotate = 24.4] [fill={rgb, 255:red, 0; green, 0; blue, 0 }  ][line width=0.08]  [draw opacity=0] (8.93,-4.29) -- (0,0) -- (8.93,4.29) -- cycle    ;
\draw  [dash pattern={on 0.84pt off 2.51pt}]  (48.15,31) .. controls (78.33,40.33) and (112.33,43) .. (142.99,30.41) ;
\draw [shift={(100.71,38.91)}, rotate = 179.93] [fill={rgb, 255:red, 0; green, 0; blue, 0 }  ][line width=0.08]  [draw opacity=0] (8.93,-4.29) -- (0,0) -- (8.93,4.29) -- cycle    ;
\draw  [dash pattern={on 4.5pt off 4.5pt}]  (142.99,30.41) -- (46.89,73.45) ;
\draw [shift={(94.94,51.93)}, rotate = 335.87] [fill={rgb, 255:red, 0; green, 0; blue, 0 }  ][line width=0.08]  [draw opacity=0] (8.93,-4.29) -- (0,0) -- (8.93,4.29) -- cycle    ;
\draw    (197.48,30.33) -- (292.32,29.74) ;
\draw [shift={(249.9,30.01)}, rotate = 179.64] [fill={rgb, 255:red, 0; green, 0; blue, 0 }  ][line width=0.08]  [draw opacity=0] (8.93,-4.29) -- (0,0) -- (8.93,4.29) -- cycle    ;
\draw    (196.22,72.78) -- (197.48,30.33) ;
\draw [shift={(197,46.56)}, rotate = 91.69] [fill={rgb, 255:red, 0; green, 0; blue, 0 }  ][line width=0.08]  [draw opacity=0] (8.93,-4.29) -- (0,0) -- (8.93,4.29) -- cycle    ;
\draw    (292.32,29.74) -- (291.07,72.78) ;
\draw [shift={(291.55,56.26)}, rotate = 271.67] [fill={rgb, 255:red, 0; green, 0; blue, 0 }  ][line width=0.08]  [draw opacity=0] (8.93,-4.29) -- (0,0) -- (8.93,4.29) -- cycle    ;
\draw  [dash pattern={on 0.84pt off 2.51pt}]  (291.07,72.78) -- (197.48,30.33) ;
\draw [shift={(239.72,49.49)}, rotate = 24.4] [fill={rgb, 255:red, 0; green, 0; blue, 0 }  ][line width=0.08]  [draw opacity=0] (8.93,-4.29) -- (0,0) -- (8.93,4.29) -- cycle    ;
\draw  [dash pattern={on 0.84pt off 2.51pt}]  (197.48,30.33) .. controls (223,42.33) and (264.33,43) .. (292.32,29.74) ;
\draw [shift={(250.11,39.38)}, rotate = 178.75] [fill={rgb, 255:red, 0; green, 0; blue, 0 }  ][line width=0.08]  [draw opacity=0] (8.93,-4.29) -- (0,0) -- (8.93,4.29) -- cycle    ;
\draw  [dash pattern={on 4.5pt off 4.5pt}]  (292.32,29.74) -- (196.22,72.78) ;
\draw [shift={(244.27,51.26)}, rotate = 335.87] [fill={rgb, 255:red, 0; green, 0; blue, 0 }  ][line width=0.08]  [draw opacity=0] (8.93,-4.29) -- (0,0) -- (8.93,4.29) -- cycle    ;
\draw    (169.47,91.6) -- (168.18,167.67) ;
\draw [shift={(168.13,170.67)}, rotate = 270.97] [fill={rgb, 255:red, 0; green, 0; blue, 0 }  ][line width=0.08]  [draw opacity=0] (8.93,-4.29) -- (0,0) -- (8.93,4.29) -- cycle    ;
\draw    (53.48,203) -- (148.32,202.41) ;
\draw [shift={(105.9,202.67)}, rotate = 179.64] [fill={rgb, 255:red, 0; green, 0; blue, 0 }  ][line width=0.08]  [draw opacity=0] (8.93,-4.29) -- (0,0) -- (8.93,4.29) -- cycle    ;
\draw    (52.22,245.45) -- (53.48,203) ;
\draw [shift={(53,219.23)}, rotate = 91.69] [fill={rgb, 255:red, 0; green, 0; blue, 0 }  ][line width=0.08]  [draw opacity=0] (8.93,-4.29) -- (0,0) -- (8.93,4.29) -- cycle    ;
\draw  [dash pattern={on 0.84pt off 2.51pt}]  (53.48,203) .. controls (79,215.67) and (117.67,215.67) .. (148.32,202.41) ;
\draw [shift={(106.02,212.2)}, rotate = 177.98] [fill={rgb, 255:red, 0; green, 0; blue, 0 }  ][line width=0.08]  [draw opacity=0] (8.93,-4.29) -- (0,0) -- (8.93,4.29) -- cycle    ;
\draw  [dash pattern={on 4.5pt off 4.5pt}]  (148.32,202.41) -- (52.22,245.45) ;
\draw [shift={(100.27,223.93)}, rotate = 335.87] [fill={rgb, 255:red, 0; green, 0; blue, 0 }  ][line width=0.08]  [draw opacity=0] (8.93,-4.29) -- (0,0) -- (8.93,4.29) -- cycle    ;
\draw    (192.99,202.41) -- (286.99,201.08) ;
\draw [shift={(244.99,201.67)}, rotate = 179.19] [fill={rgb, 255:red, 0; green, 0; blue, 0 }  ][line width=0.08]  [draw opacity=0] (8.93,-4.29) -- (0,0) -- (8.93,4.29) -- cycle    ;
\draw    (286.99,201.08) -- (285.73,244.12) ;
\draw [shift={(286.22,227.59)}, rotate = 271.67] [fill={rgb, 255:red, 0; green, 0; blue, 0 }  ][line width=0.08]  [draw opacity=0] (8.93,-4.29) -- (0,0) -- (8.93,4.29) -- cycle    ;
\draw  [dash pattern={on 0.84pt off 2.51pt}]  (285.73,244.12) -- (192.99,202.41) ;
\draw [shift={(234.8,221.21)}, rotate = 24.21] [fill={rgb, 255:red, 0; green, 0; blue, 0 }  ][line width=0.08]  [draw opacity=0] (8.93,-4.29) -- (0,0) -- (8.93,4.29) -- cycle    ;
\draw  [dash pattern={on 0.84pt off 2.51pt}]  (192.99,202.41) .. controls (222.33,217) and (263,214.33) .. (286.99,201.08) ;
\draw [shift={(245.52,212.03)}, rotate = 178.06] [fill={rgb, 255:red, 0; green, 0; blue, 0 }  ][line width=0.08]  [draw opacity=0] (8.93,-4.29) -- (0,0) -- (8.93,4.29) -- cycle    ;
\draw    (148.32,202.41) -- (192.99,202.41) ;
\draw [shift={(175.66,202.41)}, rotate = 180] [fill={rgb, 255:red, 0; green, 0; blue, 0 }  ][line width=0.08]  [draw opacity=0] (8.93,-4.29) -- (0,0) -- (8.93,4.29) -- cycle    ;
\draw  [dash pattern={on 0.84pt off 2.51pt}]  (286.99,201.08) .. controls (246.8,247.6) and (101.47,247.6) .. (53.48,203) ;
\draw [shift={(164.63,236.13)}, rotate = 0.67] [fill={rgb, 255:red, 0; green, 0; blue, 0 }  ][line width=0.08]  [draw opacity=0] (8.93,-4.29) -- (0,0) -- (8.93,4.29) -- cycle    ;
\draw    (313.33,46.67) -- (400.2,47.7) ;
\draw [shift={(403.2,47.73)}, rotate = 180.68] [fill={rgb, 255:red, 0; green, 0; blue, 0 }  ][line width=0.08]  [draw opacity=0] (8.93,-4.29) -- (0,0) -- (8.93,4.29) -- cycle    ;
\draw    (311.33,222) -- (398.2,223.03) ;
\draw [shift={(401.2,223.07)}, rotate = 180.68] [fill={rgb, 255:red, 0; green, 0; blue, 0 }  ][line width=0.08]  [draw opacity=0] (8.93,-4.29) -- (0,0) -- (8.93,4.29) -- cycle    ;
\draw    (557.3,88.27) -- (621.5,174.15) ;
\draw [shift={(623.3,176.55)}, rotate = 233.22] [fill={rgb, 255:red, 0; green, 0; blue, 0 }  ][line width=0.08]  [draw opacity=0] (8.93,-4.29) -- (0,0) -- (8.93,4.29) -- cycle    ;
\draw    (456.03,82.7) .. controls (433.15,92.42) and (422.18,34.75) .. (451.01,30.99) .. controls (479.85,27.23) and (485.18,27.23) .. (506.81,30.05) ;
\draw [shift={(433.7,54.44)}, rotate = 87.53] [fill={rgb, 255:red, 0; green, 0; blue, 0 }  ][line width=0.08]  [draw opacity=0] (8.93,-4.29) -- (0,0) -- (8.93,4.29) -- cycle    ;
\draw [shift={(483.89,28.04)}, rotate = 178.82] [fill={rgb, 255:red, 0; green, 0; blue, 0 }  ][line width=0.08]  [draw opacity=0] (8.93,-4.29) -- (0,0) -- (8.93,4.29) -- cycle    ;
\draw    (511.2,31.62) .. controls (535.02,31.93) and (534.71,80.19) .. (510.26,80.82) .. controls (485.81,81.45) and (449.13,51.67) .. (447.25,34.12) ;
\draw [shift={(528.35,61.55)}, rotate = 274.73] [fill={rgb, 255:red, 0; green, 0; blue, 0 }  ][line width=0.08]  [draw opacity=0] (8.93,-4.29) -- (0,0) -- (8.93,4.29) -- cycle    ;
\draw [shift={(468.08,62.54)}, rotate = 37.36] [fill={rgb, 255:red, 0; green, 0; blue, 0 }  ][line width=0.08]  [draw opacity=0] (8.93,-4.29) -- (0,0) -- (8.93,4.29) -- cycle    ;
\draw    (445.69,28.48) .. controls (445.06,6.39) and (515.27,4.66) .. (509.32,30.05) .. controls (503.36,55.43) and (491.76,60.14) .. (479.85,68.28) ;
\draw [shift={(483.99,11.6)}, rotate = 181.16] [fill={rgb, 255:red, 0; green, 0; blue, 0 }  ][line width=0.08]  [draw opacity=0] (8.93,-4.29) -- (0,0) -- (8.93,4.29) -- cycle    ;
\draw [shift={(495.89,56.57)}, rotate = 311.38] [fill={rgb, 255:red, 0; green, 0; blue, 0 }  ][line width=0.08]  [draw opacity=0] (8.93,-4.29) -- (0,0) -- (8.93,4.29) -- cycle    ;
\draw    (474.84,71.42) -- (456.03,82.7) ;
\draw [shift={(461.15,79.63)}, rotate = 329.04] [fill={rgb, 255:red, 0; green, 0; blue, 0 }  ][line width=0.08]  [draw opacity=0] (8.93,-4.29) -- (0,0) -- (8.93,4.29) -- cycle    ;
\draw    (601.53,82.7) .. controls (578.65,92.42) and (567.68,34.75) .. (596.51,30.99) .. controls (625.35,27.23) and (630.68,27.23) .. (652.31,30.05) ;
\draw [shift={(579.2,54.44)}, rotate = 87.53] [fill={rgb, 255:red, 0; green, 0; blue, 0 }  ][line width=0.08]  [draw opacity=0] (8.93,-4.29) -- (0,0) -- (8.93,4.29) -- cycle    ;
\draw [shift={(629.39,28.04)}, rotate = 178.82] [fill={rgb, 255:red, 0; green, 0; blue, 0 }  ][line width=0.08]  [draw opacity=0] (8.93,-4.29) -- (0,0) -- (8.93,4.29) -- cycle    ;
\draw    (656.7,31.62) .. controls (680.52,31.93) and (680.21,80.19) .. (655.76,80.82) .. controls (631.31,81.45) and (594.63,51.67) .. (592.75,34.12) ;
\draw [shift={(673.85,61.55)}, rotate = 274.73] [fill={rgb, 255:red, 0; green, 0; blue, 0 }  ][line width=0.08]  [draw opacity=0] (8.93,-4.29) -- (0,0) -- (8.93,4.29) -- cycle    ;
\draw [shift={(613.58,62.54)}, rotate = 37.36] [fill={rgb, 255:red, 0; green, 0; blue, 0 }  ][line width=0.08]  [draw opacity=0] (8.93,-4.29) -- (0,0) -- (8.93,4.29) -- cycle    ;
\draw    (591.19,28.48) .. controls (590.56,6.39) and (660.77,4.66) .. (654.82,30.05) .. controls (648.86,55.43) and (637.26,60.14) .. (625.35,68.28) ;
\draw [shift={(629.49,11.6)}, rotate = 181.16] [fill={rgb, 255:red, 0; green, 0; blue, 0 }  ][line width=0.08]  [draw opacity=0] (8.93,-4.29) -- (0,0) -- (8.93,4.29) -- cycle    ;
\draw [shift={(641.39,56.57)}, rotate = 311.38] [fill={rgb, 255:red, 0; green, 0; blue, 0 }  ][line width=0.08]  [draw opacity=0] (8.93,-4.29) -- (0,0) -- (8.93,4.29) -- cycle    ;
\draw    (620.34,71.42) -- (601.53,82.7) ;
\draw [shift={(606.65,79.63)}, rotate = 329.04] [fill={rgb, 255:red, 0; green, 0; blue, 0 }  ][line width=0.08]  [draw opacity=0] (8.93,-4.29) -- (0,0) -- (8.93,4.29) -- cycle    ;
\draw    (584.82,259.2) .. controls (570.98,268.92) and (552.3,209.35) .. (581.79,207.49) .. controls (611.28,205.63) and (604.3,207.05) .. (614.3,203.55) ;
\draw [shift={(566.13,229.29)}, rotate = 85.17] [fill={rgb, 255:red, 0; green, 0; blue, 0 }  ][line width=0.08]  [draw opacity=0] (8.93,-4.29) -- (0,0) -- (8.93,4.29) -- cycle    ;
\draw [shift={(603.35,206.29)}, rotate = 176.96] [fill={rgb, 255:red, 0; green, 0; blue, 0 }  ][line width=0.08]  [draw opacity=0] (8.93,-4.29) -- (0,0) -- (8.93,4.29) -- cycle    ;
\draw    (656.01,247.65) .. controls (647.99,257.15) and (622.34,253.82) .. (610.8,254.45) .. controls (599.26,255.08) and (580.4,228.17) .. (579.51,210.62) ;
\draw [shift={(629.12,254.52)}, rotate = 358.44] [fill={rgb, 255:red, 0; green, 0; blue, 0 }  ][line width=0.08]  [draw opacity=0] (8.93,-4.29) -- (0,0) -- (8.93,4.29) -- cycle    ;
\draw [shift={(586.92,232.66)}, rotate = 58.39] [fill={rgb, 255:red, 0; green, 0; blue, 0 }  ][line width=0.08]  [draw opacity=0] (8.93,-4.29) -- (0,0) -- (8.93,4.29) -- cycle    ;
\draw    (578.57,204.98) .. controls (578.19,182.89) and (620.66,181.16) .. (617.06,206.55) .. controls (613.46,231.93) and (606.44,236.64) .. (599.24,244.78) ;
\draw [shift={(604.38,188.53)}, rotate = 185.39] [fill={rgb, 255:red, 0; green, 0; blue, 0 }  ][line width=0.08]  [draw opacity=0] (8.93,-4.29) -- (0,0) -- (8.93,4.29) -- cycle    ;
\draw [shift={(609.59,231.97)}, rotate = 295.73] [fill={rgb, 255:red, 0; green, 0; blue, 0 }  ][line width=0.08]  [draw opacity=0] (8.93,-4.29) -- (0,0) -- (8.93,4.29) -- cycle    ;
\draw    (596.2,247.92) -- (584.82,259.2) ;
\draw [shift={(586.96,257.08)}, rotate = 315.24] [fill={rgb, 255:red, 0; green, 0; blue, 0 }  ][line width=0.08]  [draw opacity=0] (8.93,-4.29) -- (0,0) -- (8.93,4.29) -- cycle    ;
\draw    (679.54,203.62) .. controls (698.13,203.93) and (697.88,252.19) .. (678.81,252.82) .. controls (659.73,253.45) and (631.11,223.67) .. (629.64,206.12) ;
\draw [shift={(692.92,233.55)}, rotate = 273.47] [fill={rgb, 255:red, 0; green, 0; blue, 0 }  ][line width=0.08]  [draw opacity=0] (8.93,-4.29) -- (0,0) -- (8.93,4.29) -- cycle    ;
\draw [shift={(644.22,232.76)}, rotate = 45.83] [fill={rgb, 255:red, 0; green, 0; blue, 0 }  ][line width=0.08]  [draw opacity=0] (8.93,-4.29) -- (0,0) -- (8.93,4.29) -- cycle    ;
\draw    (628.42,200.48) .. controls (627.93,178.39) and (679.95,176.46) .. (675.3,201.85) .. controls (670.65,227.24) and (669.59,236.2) .. (660.3,244.35) ;
\draw [shift={(658.54,183.68)}, rotate = 182.75] [fill={rgb, 255:red, 0; green, 0; blue, 0 }  ][line width=0.08]  [draw opacity=0] (8.93,-4.29) -- (0,0) -- (8.93,4.29) -- cycle    ;
\draw [shift={(669.6,229.37)}, rotate = 284.8] [fill={rgb, 255:red, 0; green, 0; blue, 0 }  ][line width=0.08]  [draw opacity=0] (8.93,-4.29) -- (0,0) -- (8.93,4.29) -- cycle    ;
\draw    (620.06,204.05) .. controls (636.58,202.05) and (655.13,202.45) .. (672.3,203.45) ;
\draw [shift={(651.25,202.65)}, rotate = 180.26] [fill={rgb, 255:red, 0; green, 0; blue, 0 }  ][line width=0.08]  [draw opacity=0] (8.93,-4.29) -- (0,0) -- (8.93,4.29) -- cycle    ;
\draw    (423.32,261.2) .. controls (409.48,270.92) and (390.8,211.35) .. (420.29,209.49) .. controls (449.78,207.63) and (440.96,205.73) .. (454.04,208.55) ;
\draw [shift={(404.63,231.29)}, rotate = 85.17] [fill={rgb, 255:red, 0; green, 0; blue, 0 }  ][line width=0.08]  [draw opacity=0] (8.93,-4.29) -- (0,0) -- (8.93,4.29) -- cycle    ;
\draw [shift={(442.21,207.4)}, rotate = 172.81] [fill={rgb, 255:red, 0; green, 0; blue, 0 }  ][line width=0.08]  [draw opacity=0] (8.93,-4.29) -- (0,0) -- (8.93,4.29) -- cycle    ;
\draw    (525.3,223.05) .. controls (517.27,232.55) and (460.84,255.82) .. (449.3,256.45) .. controls (437.76,257.08) and (418.9,230.17) .. (418.01,212.62) ;
\draw [shift={(484.15,244.96)}, rotate = 337.1] [fill={rgb, 255:red, 0; green, 0; blue, 0 }  ][line width=0.08]  [draw opacity=0] (8.93,-4.29) -- (0,0) -- (8.93,4.29) -- cycle    ;
\draw [shift={(425.42,234.66)}, rotate = 58.39] [fill={rgb, 255:red, 0; green, 0; blue, 0 }  ][line width=0.08]  [draw opacity=0] (8.93,-4.29) -- (0,0) -- (8.93,4.29) -- cycle    ;
\draw    (417.07,206.98) .. controls (416.69,184.89) and (459.16,183.16) .. (455.56,208.55) .. controls (451.96,233.93) and (444.94,238.64) .. (437.74,246.78) ;
\draw [shift={(442.88,190.53)}, rotate = 185.39] [fill={rgb, 255:red, 0; green, 0; blue, 0 }  ][line width=0.08]  [draw opacity=0] (8.93,-4.29) -- (0,0) -- (8.93,4.29) -- cycle    ;
\draw [shift={(448.09,233.97)}, rotate = 295.73] [fill={rgb, 255:red, 0; green, 0; blue, 0 }  ][line width=0.08]  [draw opacity=0] (8.93,-4.29) -- (0,0) -- (8.93,4.29) -- cycle    ;
\draw    (434.7,249.92) -- (423.32,261.2) ;
\draw [shift={(425.46,259.08)}, rotate = 315.24] [fill={rgb, 255:red, 0; green, 0; blue, 0 }  ][line width=0.08]  [draw opacity=0] (8.93,-4.29) -- (0,0) -- (8.93,4.29) -- cycle    ;
\draw    (458.56,206.05) .. controls (475.08,204.05) and (471.63,204.55) .. (488.8,205.55) ;
\draw [shift={(478.74,204.9)}, rotate = 182.07] [fill={rgb, 255:red, 0; green, 0; blue, 0 }  ][line width=0.08]  [draw opacity=0] (8.93,-4.29) -- (0,0) -- (8.93,4.29) -- cycle    ;
\draw    (530.8,223.05) .. controls (536.37,233.63) and (504.8,255.55) .. (498.8,242.55) ;
\draw [shift={(515,244)}, rotate = 329.54] [fill={rgb, 255:red, 0; green, 0; blue, 0 }  ][line width=0.08]  [draw opacity=0] (8.93,-4.29) -- (0,0) -- (8.93,4.29) -- cycle    ;
\draw    (493.3,234.05) .. controls (487.3,227.05) and (484.8,224.05) .. (476.3,209.05) ;
\draw [shift={(481.33,217.69)}, rotate = 57.78] [fill={rgb, 255:red, 0; green, 0; blue, 0 }  ][line width=0.08]  [draw opacity=0] (8.93,-4.29) -- (0,0) -- (8.93,4.29) -- cycle    ;
\draw    (472.3,200.55) .. controls (459.8,173.05) and (489.8,205.05) .. (496.3,211.05) .. controls (502.8,217.05) and (513.8,221.05) .. (516.3,211.55) ;
\draw [shift={(483.45,198.66)}, rotate = 221.43] [fill={rgb, 255:red, 0; green, 0; blue, 0 }  ][line width=0.08]  [draw opacity=0] (8.93,-4.29) -- (0,0) -- (8.93,4.29) -- cycle    ;
\draw [shift={(511.91,216.84)}, rotate = 180.99] [fill={rgb, 255:red, 0; green, 0; blue, 0 }  ][line width=0.08]  [draw opacity=0] (8.93,-4.29) -- (0,0) -- (8.93,4.29) -- cycle    ;
\draw    (518.3,205.05) .. controls (526.8,195.05) and (551.3,199.55) .. (525.3,223.05) ;
\draw [shift={(536.03,209.8)}, rotate = 275.74] [fill={rgb, 255:red, 0; green, 0; blue, 0 }  ][line width=0.08]  [draw opacity=0] (8.93,-4.29) -- (0,0) -- (8.93,4.29) -- cycle    ;
\draw    (494.8,206.55) .. controls (511.32,204.55) and (522.3,206.05) .. (526.3,215.55) ;
\draw [shift={(517.38,207.37)}, rotate = 192.5] [fill={rgb, 255:red, 0; green, 0; blue, 0 }  ][line width=0.08]  [draw opacity=0] (8.93,-4.29) -- (0,0) -- (8.93,4.29) -- cycle    ;

\draw (36.56,14.33) node [anchor=north west][inner sep=0.75pt]   [align=left] {$\displaystyle A$};
\draw (141.39,15.72) node [anchor=north west][inner sep=0.75pt]   [align=left] {B};
\draw (138.32,67.99) node [anchor=north west][inner sep=0.75pt]   [align=left] {$\displaystyle A$};
\draw (38.38,69.96) node [anchor=north west][inner sep=0.75pt]   [align=left] {B};
\draw (185.9,13.67) node [anchor=north west][inner sep=0.75pt]   [align=left] {$\displaystyle C$};
\draw (290.72,15.05) node [anchor=north west][inner sep=0.75pt]   [align=left] {D};
\draw (287.65,67.32) node [anchor=north west][inner sep=0.75pt]   [align=left] {$\displaystyle C$};
\draw (187.71,69.3) node [anchor=north west][inner sep=0.75pt]   [align=left] {D};
\draw (41.9,183.67) node [anchor=north west][inner sep=0.75pt]   [align=left] {$\displaystyle A$};
\draw (144.05,181.67) node [anchor=north west][inner sep=0.75pt]   [align=left] {B};
\draw (185.9,179) node [anchor=north west][inner sep=0.75pt]   [align=left] {$\displaystyle C$};
\draw (283.39,180.33) node [anchor=north west][inner sep=0.75pt]   [align=left] {D};
\draw (192.72,246.33) node [anchor=north west][inner sep=0.75pt]   [align=left] {D};
\draw (279.23,249) node [anchor=north west][inner sep=0.75pt]   [align=left] {$\displaystyle C$};
\draw (138.98,245.32) node [anchor=north west][inner sep=0.75pt]   [align=left] {$\displaystyle A$};
\draw (45.04,245.3) node [anchor=north west][inner sep=0.75pt]   [align=left] {B};
\draw (173.3,123) node [anchor=north west][inner sep=0.75pt]   [align=left] {Concatenation};
\draw (339.8,199.5) node [anchor=north west][inner sep=0.75pt]   [align=left] {Knot};
\draw (339.8,56.1) node [anchor=north west][inner sep=0.75pt]   [align=left] {Knot};
\draw (594.8,115.45) node [anchor=north west][inner sep=0.75pt]   [align=left] {\#};
\draw (536.8,221.3) node [anchor=north west][inner sep=0.75pt]    {$\cong $};

\end{tikzpicture}

        \caption{\justifying Concatenation of two sequences AB-BA and CD-DC into ABCD-BADC. The knot representation of the disjoint sequences gives two trefoil knots, which under a knot sum (denoted $\#$) gives the same knot as the knot representation of ABCD-BADC.}
        \label{fig:reducible sequences}
    \end{figure*}
    By contraposition, this means that if the knot is associated to an irreducible sequence of causal orders, then the knot is prime or the unknot.
\end{proof}

\begin{lemma}
    \label{Maximally indefinite knots}
    Knots associated to maximally indefinite causal order for superpositions of $M=2$ spacetimes are $(2,2N-1)$-torus knots. 
\end{lemma}

\begin{proof}
    A $(p,q)$ torus knot is a knot that lies on the surface of a torus and is created by looping a string $q$ times through the hole and $p$ times around the axis of rotational symmetry of the torus. By construction, knots associated to maximally indefinite causal order are just intertwining at every node with $2$ windings around the axis of rotational symmetry, thus resulting to a $(2,q)$-torus knot. For $N$ events, we get $2N-1$ windings, where we have a factor of $2(N-1)$ coming from winding around every event except the last, and $+1$ from the final crossing between the last event of the first sequence and the first event of the second sequence.
\end{proof}

\begin{notation}
    We write $\nabla_{T_{p,q}}(z) \equiv \nabla_{(p,q)}(z)$, where $T_{p,q}$ is the $(p,q)$-torus knot.
\end{notation}

\begin{lemma}
    For $(2,q)$ torus knots where $q \equiv 1 \bmod 2$, \begin{equation}
        \nabla_{(2,q)}(z)_2 = \frac{(q-1)(q+1)}{8}.
    \end{equation} 
\end{lemma}

\begin{proof}
    For $(2,q)$-torus knots, $\nabla_{(2,q)}(z) = F_q(z)$ \cite{Agle2012} where $F_q(z)$ are Fibonacci polynomials. But \begin{equation}
        F_q(z) = \sum_{k=0}^q F(q,k) z^k
    \end{equation}
    where 
    \begin{equation}
        F(q,k) = \begin{cases}
            \begin{pmatrix}
                \frac{1}{2}(q+k-1) \\ k
            \end{pmatrix} \text{ if } q \neq k \text{ mod } 2\\
            0 \text{ otherwise.}
        \end{cases}
    \end{equation}
    We are interested in the quadratic term, for which $k=2$, $\begin{pmatrix}
        $N$ \\ 2
    \end{pmatrix} = \begin{pmatrix} N \\ 2 \end{pmatrix}$, and $q \equiv 1 \bmod 2$ so that
    \begin{equation}
        \nabla_{(2,q)}(z)_2 = F(q,2) = \frac{(q-1)(q+1)}{8},
    \end{equation}
    which concludes the proof.
\end{proof}

\begin{lemma}
    \label{Causal indefiniteness maximally indefinite knot}
    The causal indefiniteness of a superposition of two spacetimes $(\mathcal{M}_\mathcal{A},g_{\mathcal{A}})$ and $(\mathcal{M}_\mathcal{B},g_{\mathcal{B}})$ with $N$ events in either definite causal order or in maximally indefinite causal order is related to its knot representation as  
    \begin{equation}
        \nabla(z)_2 = \delta(\mathcal{A},\mathcal{B}).
    \end{equation}
\end{lemma}

\begin{proof}
    This is trivially true for definite causal orders. By lemma \ref{Causal indefiniteness maximally indefinite}, $\delta(\mathcal{A},\mathcal{B}) = \begin{pmatrix} N \\ 2 \end{pmatrix}$ for maximally indefinite causal order. But by lemma \ref{Maximally indefinite knots}, the knot associated to this scenario is a $(2,2N-1)$ torus knot, with $\nabla_{(2,2N-1)}(z)_2 = \begin{pmatrix} N \\ 2 \end{pmatrix}$.
\end{proof}

\begin{lemma}
    \label{Monicity Conway polynomial}
    For any prime knot or the unknot, the Alexander-Conway polynomial has $\nabla(z)_0 = 1$ and $\nabla(z)_1 = 0$.
\end{lemma}

\begin{proof}
    $\nabla(z)_0 = 1$ comes from the fact that applying the skein relations, we always end up with an unknot, for which $\nabla(\text{O}) = 1$, with only one skein branch leading to this. $\nabla(z)_1 = 0$ comes from the fact that we do not work with links but with knots, so applying skein relations will never yield a term with a single factor of z (as would be the case e.g.~for the Hopf link). 
\end{proof}

\begin{lemma}
    \label{Additivity of quadratic term single prime knot sum}
    The quadratic term of the Alexander-Conway polynomial is additive under $K_1\#K_2$ knot addition where $K_1$ and $K_2$ are prime or the unknot, i.e.
    \begin{equation}
        \nabla_{K_1\#K_2}(z)_2 = \nabla_{K_1}(z)_2 + \nabla_{K_2}(z)_2.
    \end{equation}
\end{lemma}

\begin{proof}
    Since the Alexander-Conway polynomial is multiplicative under knot addition \cite{Freyd1985}, we have
    \begin{align}
        \nabla_{K_1\#K_2}(z)_2 &= \nabla_{K_1}(z)_2 \cdot \nabla_{K_2}(z)_0 \nonumber \\ &+ \nabla_{K_1}(z)_1 \cdot \nabla_{K_2}(z)_1 \nonumber \\ &+ \nabla_{K_1}(z)_0 \cdot \nabla_{K_2}(z)_2 \\
        &= \nabla_{K_1}(z)_2\cdot1 + 0\cdot0 + 1\cdot\nabla_{K_1}(z)_2
    \end{align}
    by lemma \ref{Monicity Conway polynomial}, which concludes the proof.
\end{proof}

\begin{lemma}
    \label{Additivity of quadratic term arbitrary prime knot sum}
    The quadratic term of the Alexander-Conway polynomial of arbitrary knot sums $K_1\#K_2\#...\#K_n$ of prime knots (and unknots) $K_i$ is additive.
\end{lemma}

\begin{proof}
    We prove this by induction. The base case is proven in lemma \ref{Additivity of quadratic term single prime knot sum}. Suppose that this holds for the sum of $k$ prime knots of the form above; we write this knot $K_{\#k}$ -- this is a composite knot. We want to show that this relation holds for the composite knot $K_{\#k}\#K_{k+1}$.

    Since the Alexander-Conway polynomial is multiplicative under knot addition \cite{Freyd1985}, and $\nabla_{K_i}(z)_0 = 1$ for any prime knot $K_i$ (by lemma \ref{Monicity Conway polynomial}), then $\nabla_{K_{\#k}}(z)_0 = \prod_i^k \nabla_{K_i}(z)_0 = 1$. Then
    \begin{align}
        \nabla_{K_{\#k}\#K_{k+1}}(z)_2 &= \nabla_{K_{\#k}}(z)_2 \cdot \nabla_{K_{k+1}}(z)_0 \nonumber \\ &+ \nabla_{K_{\#k}}(z)_1 \cdot \nabla_{K_{k+1}}(z)_1 \nonumber \\ &+ \nabla_{K_{\#k}}(z)_0 \cdot \nabla_{K_{k+1}}(z)_2 \\
        &= \nabla_{K_{\#k}}(z)_2\cdot1 \nonumber \\ &+ \nabla_{K_{\#k}}(z)_1\cdot0 \nonumber \\ &+ 1\cdot\nabla_{K_{k+1}}(z)_2
    \end{align}
    by lemma \ref{Monicity Conway polynomial}, so $\nabla_{T_{\#k}\#K_{k+1}} = \nabla_{T_{\#k}}(z)_2 + \nabla_{K_{k+1}}(z)_2$, which proves the inductive step. Thus, for arbitrary sums of such prime knots, the second term of the Alexander-Conway polynomial is additive. 
\end{proof}

\begin{theorem}
    \label{Causal indefiniteness maximally indefinite subsequences}
    The causal indefiniteness of a superposition of two spacetimes $(\mathcal{M}_\mathcal{A},g_{\mathcal{A}})$ and $(\mathcal{M}_\mathcal{B},g_{\mathcal{B}})$ with $N$ events whose subsequences are either in definite causal order or in maximally indefinite causal order is related to its knot representation as  
    \begin{equation}
        \nabla(z)_2 = \delta(\mathcal{A},\mathcal{B}).
    \end{equation}
\end{theorem}

\begin{proof}
    This follows from lemmas \ref{Causal indefiniteness maximally indefinite knot}, \ref{Additivity of quadratic term arbitrary prime knot sum}, and \ref{Additivity of indefiniteness} where we have that $\delta(K_1\#K_2\#...\#K_n) = \delta(K_1) + \delta(K_2) + ... + \delta(K_n)$ and since non-trivial torus knots (not $(2,1)$ torus knots, i.e.~the unknot) are prime \cite{Norwood1982} -- however, the arguments above also apply for unknots.
\end{proof}

Note that theorem \ref{Causal indefiniteness maximally indefinite subsequences} is a generalisation of lemma \ref{Causal indefiniteness maximally indefinite knot}, where we allow for concatenations of maximally indefinite causal orders.

\begin{conjecture}
    \label{Causal indefiniteness prime knot}
    The causal indefiniteness of a superposition of two spacetimes $(\mathcal{M}_\mathcal{A},g_{\mathcal{A}})$ and $(\mathcal{M}_\mathcal{B},g_{\mathcal{B}})$ with $N$ events for a pair of irreducible causal sequences is related to its prime (by lemma \ref{Prime knots irreducible subsequences}) knot representation as  
    \begin{equation}
        \nabla(z)_2 = \delta(\mathcal{A},\mathcal{B}).
    \end{equation}
\end{conjecture}

Note that we have shown that this is true for knots whose subsequences are either in definite causal orders or in maximally indefinite causal orders from theorem \ref{Causal indefiniteness maximally indefinite subsequences}. The question is then whether this holds for prime knots representing irreducible braided causal orders. \sam{To support this conjecture, we provide in Fig.~\ref{fig:conjexamples} three non-trivial examples for which it holds: ABCD-BDAC (cinquefoil knot $5_1$ in Alexander-Briggs notation with $\nabla(z) = 1 + 3z^2 + z^4$ and $\delta = 3$), ABCD-DACB ($7_3$ prime knot in Alexander-Briggs notation with $\nabla(z) = 1+4z^2+2z^4$ and $\delta = 4$), and ABCD-CBDA (also a $7_3$ prime knot in Alexander-Briggs notation with $\nabla(z) = 1+4z^2+2z^4$ and $\delta = 4$). Other examples are provided in Appendix \ref{sec:App conj}.

\begin{figure}[t!]
    \centering
    \subfloat[\centering ABCD-BDAC, $\nabla(z)_2 = \delta = 3$ \label{fig:ABCD-BDAC}]{\includegraphics[scale=0.12]{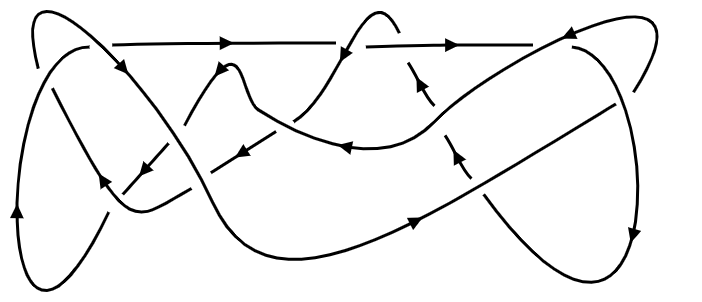}
}
    \subfloat[\centering ABCD-DACB, $\nabla(z)_2 = \delta = 4$ \label{fig:ABCD-DACB}]{\includegraphics[scale=0.12]{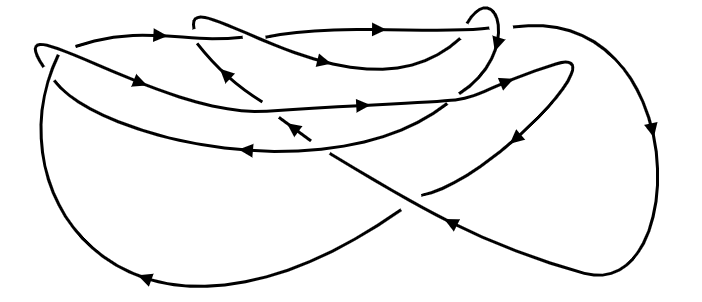}}
    \subfloat[\centering ABCD-CBDA, $\nabla(z)_2 = \delta = 4$ \label{fig:ABCD-CBDA}]{\includegraphics[scale=0.12]{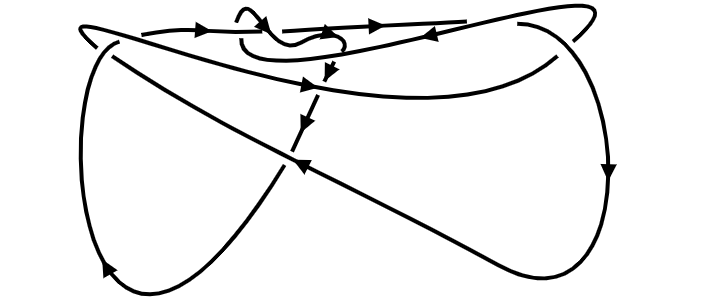}}
    \caption{\justifying \sam{Examples of irreducible causal sequences \acd{that} satisfy Conjecture \ref{Causal indefiniteness prime knot}.}}
    \label{fig:conjexamples}
\end{figure}
}

\begin{theorem}
    \label{Causal indefiniteness and knot theory}
    Given Conjecture \ref{Causal indefiniteness prime knot}, the causal indefiniteness of a superposition of two spacetimes $(\mathcal{M}_\mathcal{A},g_{\mathcal{A}})$ and $(\mathcal{M}_\mathcal{B},g_{\mathcal{B}})$ with $N$ events is given by \begin{equation}
        \nabla(z)_2 = \delta(\mathcal{A},\mathcal{B}),
    \end{equation} where $\nabla(z)_2$ is the quadratic term of the Alexander-Conway polynomial of the knot representation of the causal sequence.
\end{theorem}

\begin{proof}
    Assuming Conjecture \ref{Causal indefiniteness prime knot}, this follows from lemmas \ref{Additivity of quadratic term arbitrary prime knot sum} and \ref{Additivity of indefiniteness} where we have that $\delta(K_1\#K_2\#...\#K_n) = \delta(K_1) + \delta(K_2) + ... + \delta(K_n)$ for knot representations $K_i$ of arbitrary irreducible subsequences of causal orders.
\end{proof}

It then follows that we can read off the indefiniteness of the causal order from knot diagrams. This theorem implies a topological protection of the notion of indefinite causal order: changing quantifiers of the causal order requires topological (discontinuous) changes. In \cite{Kauffman1987,Tsau2016} it was argued that the quadratic term of the Alexander-Conway polynomial of a knot has a topological meaning: it is a measure of the ``self-linking" of the knot.

\sam{Here, there is thus a notion of \enquote{topological protection} which, though not referring to the topology of the spacetime, is related to it. The \enquote{topological protection} here is at the level of the causal sequences themselves and at the level of the associated knots. Quantum diffeomorphisms cannot change the causal ordering in spacetime \cite{Delahamette2022}, and this is trivially highlighted by the fact that two sequences of events with the same ordering have the same knots and consequently the same quantifiers of indefiniteness. At the level of knots, quantum diffeomorphisms can be seen as smooth deformations of the knot, which trivially cannot change the causal ordering of the associated causal sequences. Our result here is more general: certain (though not all) permutations of events, at the level of the causal sequences (i.e.~at the level of sets) preserve the causal ordering of ensembles of events, and such permutations can be implemented as topology-preserving maps. Roughly speaking, the topological properties inherent to the notion of indefinite causal order (including spacetime considerations though not restricted to such) are encoded in the knots. This in turn explains Proposition \ref{Diffeomorphism invariance}, and in this sense we may formulate the following theorem.}

\begin{theorem} \label{thm: topology}
    Given Conjecture \ref{Causal indefiniteness prime knot}, any topology-preserving map cannot change the causal indefiniteness of sequences of events.
\end{theorem}

This is similar to the theorem proved in \cite{CastroRuiz2018} in the context of the process matrix formalism, stating that causal order is invariant under any continuous and reversible transformation. Note that the topological protection of (in-)definite causal order for $M=2$ spacetimes implies \sam{the topological protection} for more generic superpositions of spacetimes, since quantifiers of indefiniteness for such general situations depend on those for $M=2$ spacetimes. Further, note that Theorem \ref{thm: topology} does not strictly imply that the sequences of events \sam{in each branch} cannot change under topology-preserving maps: two different \sam{pairs of} sequences of events with the same causal indefiniteness can be represented by the same knot \sam{(and thus have the same Alexander-Conway polynomials)}. \sam{Thus, topology-preserving maps can in principle change the sequences of events while preserving the causal indefiniteness.}

\section{Operational encoding of the causal order}
\label{sec:operational}

Let us now turn to the operational encoding of the causal order provided in \cite{Delahamette2022} for two events in two spacetimes in superposition. The protocol given therein allows to encode the causal order between two events, $s=\pm1$, in orthogonal qubit states by sending a target system with an internal spin-$1/2$ degree of freedom through a setup, passing in superposition of opposite orders between two agents Alice and Bob, each equipped with a memory register. The set-up is tuned such that the internal degree of freedom starts in an initial factorised state, passing through the first laboratory at proper time $\tau_1^\ast$ and the second laboratory at $\tau^\ast_2$. Upon the target system's passing, each agent measures its state in the orthonormal basis $\{\ket{b_1}_T, \ket{b_2}_T\}$ and encodes $\ket{\tau_i^\ast}$ in their respective memory register. By appropriate post-processing, the branch-wise causal order is thus encoded in the memory register.

This encoding can be generalised to the case of $N$ events -- and thus $N$ agents -- and $M$ spacetimes. The target system is now taken to have an $N$-dimensional internal degree of freedom, such that there exists an orthonormal basis $\{\ket{b_i}_T\}_{i=1,..,N}$ in which each agent measures the test particle upon its passage in their lab. The crossings of the target system's worldline with that of each of the agents' laboratories is tuned in a way that in each of the $M$ branches, the target crosses the $j$-th laboratory at proper time $\tau^\ast_j$. Thus, upon measurement, agent $i$ receives the outcome corresponding to state $\ket{b_{\pi(i)}}_T$, depending on the permutation $\pi$ of the order of events in the respective branch, and encodes $\ket{\tau^\ast_{\pi(i)}}_{R_i}$ in their memory register $R_i$. Thus, given the initial joint state
\begin{equation}
    \ket{\Psi_0}=\bigg(\sum_{\mathcal{X}=1}^M c_\mathcal{X} \ket{\tilde{g}_\mathcal{X}}\bigg) \bigotimes_{i=1}^N \ket{0}_{R_i}\otimes\ket{b_0}_T,
\end{equation}
the state evolves to
\begin{equation}
    \ket{\Psi_{N}}=\bigg(\sum_{\mathcal{X}=1}^M c_\mathcal{X} \ket{\tilde{g}_\mathcal{X}} \bigotimes_{i=1}^N \ket{\tau^\ast_{\pi_\mathcal{X}(i)}}_{R_i}\bigg) \ket{b_{N}}_T,
\end{equation}
where $\pi_\mathcal{X}$ is the permutation corresponding to the order of the $N$ events in branch $\mathcal{X}$ of the superposition. During post-processing, once the target system has evolved to its final state $\ket{b_f}_T$, a referee transforms the state of all memory registers unitarily to 
\begin{equation}
    \begin{aligned}
    \small \ket{\Psi_{f}}=\bigg(\sum_{\mathcal{X}=1}^M c_\mathcal{X} \ket{\tilde{g}_\mathcal{X}} {\sqrt{N}}&\bigotimes_{\substack{i=1\\j=i+1}}^{N-1}  \ket{s^\mathcal{X}_{ij}|\tau^\ast_{{\pi_\chi(j)}}-\tau^\ast_{{\pi_\chi(i)}}|}_{R_i} \bigg) \\
    &\otimes \ket{\textstyle \sum_{k=1}^N \tau^\ast_k}_{R_N}\ket{b_{f}}_T.
\end{aligned}
\end{equation}
{The concrete transformation implementing this change is given by $\hat{U}_{rel} = \mathbb{1}_G\otimes \hat{U}_R \otimes \mathbb{1}_T$ with\footnote{Note that $\hat{U}_R$ is similar to the transformation to relative coordinates $\hat{T}_{rel}^{(M)}$ given in \cite{delaHamette2021falling}.}
\begin{align}
    \hat{U}_R = \int &\left(\prod_{n=1}^Nd\tau_n\right)\sqrt{|\det J|}\left(\bigotimes_{i=1}^{N-1}|\tau_{i+1}-\tau_i\rangle\langle\tau_i|_{R_i} \right)\nonumber\\
    &\otimes |\sum_{k=1}^N\tau_k\rangle\langle \tau_N|_{R_N}\label{eq:UR}
\end{align}
Here, $J$ is the Jacobian associated to the coordinate transformation $\tau_i\to\tau_i' =\tau_{i+1}-\tau_i$ for $i=1,\dots,N-1$ and $\tau_N\to\tau_N' = \sum_{k=1}^N \tau_k$, which simply yields $|\det J| = N$, as shown in Appendix \ref{app:unitarity}. There, we also show that \acd{this} transformation is indeed unitary.}
By tracing out the state of the $N$-th register and the target system and relabelling the state of the $i$-th register to just $\ket{s_{ij}^\mathcal{X}}$, we have
\begin{equation}
    \ket{\Psi} = \sum_{\mathcal{X}=1}^M c_\mathcal{X} \ket{\tilde{g}_\mathcal{X}} \bigotimes_{\substack{i = 1\\j=i+1}}^{N-1} \ket{s_{ij}^\mathcal{X}}.
\end{equation}
Post-selecting the control with the gravitational field on the outcome $\sum_{\mathcal{X}=1}^M \ket{\tilde{g}_\mathcal{X}}$, we get that the test particle is in the state
\begin{equation}
    \ket{\psi} = \sum_{\mathcal{X}=1}^M c_\mathcal{X} \bigotimes_{\substack{i=1\\j=i+1}}^{N-1} \ket{s_{ij}^\mathcal{X}}.\label{eq: operationalEncoding}
\end{equation}
Note that all information about the total causal order of the collection $\mathcal{S}_\mathcal{X}$ of events is contained in the causal order $s_{i(i+1)}$ between neighbouring pairs of events, which is is explicitly encoded in the state \eqref{eq: operationalEncoding}. This concludes the generalisation of the operational encoding of causal order for an arbitrary finite number of events and spacetimes in superposition. 

\section{Measures of quantum coherence in a superposition of causal orders}
\label{sec:quantumcoherence}

In this section, we want to explore quantifiers for indefinite causal order that take into account the \emph{coherence} of a quantum state. Given a pure quantum state $\ket{\psi}$ of the form \eqref{eq: operationalEncoding}, we can associate to it the density operator
\begin{equation}
    \tilde{\rho} = \ket{\psi}\bra{\psi} = \sum_{\mathcal{X,Y}=1}^M c_\mathcal{X} c_\mathcal{Y}^\ast \bigotimes_{\substack{i=1\\j=i+1}}^{N-1} \ket{s_{ij}^\mathcal{X}}\bra{s_{ij}^\mathcal{Y}}.
\end{equation}
From this, we can define the density operator $\rho$ with entries $[\rho]_{\mathcal{X}\mathcal{Y}} = c_\mathcal{X} c_\mathcal{Y}^\ast$.

\begin{definition}[$l_1$ coherence]
    The $l_1$ \textit{coherence} of a quantum state $\rho \in \mathcal{B}(\mathcal{H})$ is defined as
    \begin{equation}
        C_{l_1} := \sum_{m\neq n} \abs{\rho_{mn}}
    \end{equation}
    i.e.~as the sum of off-diagonal components of the density matrix $\rho$.
\end{definition}

We can use the coherence to further distinguish between different \enquote{strengths} of indefinite causal order. Let us start by considering the pairwise causal order between two events in a superposition of two spacetimes. In this case, $C_{l_1} = 2 |\rho_{\mathcal{A}\mathcal{B}}|$, and we can define a weighted version of the pairwise causal order.
\begin{definition}[Pairwise \emph{quantum} causal order]
   For two spacetimes $(\mathcal{M}_\mathcal{A},g_{\mathcal{A}})$ and $(\mathcal{M}_\mathcal{B},g_{\mathcal{B}})$, we define the \textit{pairwise quantum causal order} between two events $\mathcal{E}_a$ and $\mathcal{E}_b$ to be
\begin{equation}
\check{\tsadi}^{\mathcal{A}\mathcal{B}}_{a b} := 2 |\rho_{\mathcal{A}\mathcal{B}}|s^{\mathcal{A}}_{a  b} s^{\mathcal{B}}_{a b}.\label{eq: QPairwiseCO}
\end{equation}
\end{definition}
\noindent If the pairwise causal order is definite, Eq.~\eqref{eq: QPairwiseCO} reduces to the $l_1$ coherence measure. This is not relevant for our purposes. If the pairwise causal order is indefinite, however, it can be interpreted as a measure of the \enquote{strength} of the indefiniteness of the causal order -- the more coherent the superposition, the stronger the indefiniteness. The same reasoning can be applied to the longitudinal causal order.

\begin{definition}[Longitudinal \emph{quantum} causal order]
    The \textit{longitudinal quantum causal order} between two spacetimes $(\mathcal{M}_\mathcal{A},g_{\mathcal{A}})$ and $(\mathcal{M}_\mathcal{B},g_{\mathcal{B}})$ is 
\begin{equation}
    \check{\ayin}^{\mathcal{A}\mathcal{B}} := \sum_{1 \leq i < j}^N 2|\rho_{\mathcal{A}\mathcal{B}}|\tsadi^{\mathcal{A}\mathcal{B}}_{ij}.
\end{equation}    
\end{definition}
When going beyond a superposition of two spacetimes, we also want to be able to take into account differing levels of coherence between different pairs of branches. In this case, it is no longer enough to work with the total coherence $C_{l_1}$. Instead, we want to take into account the coherence $2|\rho_{\mathcal{X}\mathcal{Y}}|$ between any pair of branches $\mathcal{X}$ and $\mathcal{Y}$ separately.

\begin{definition}[Transverse \emph{quantum} causal order]
    The \textit{transverse quantum causal order} between two events $\mathcal{E}_a$ and $\mathcal{E}_b$ in a collection of $M$ spacetimes $\{(\mathcal{M}_\mathcal{X},g_{\mathcal{X}})\}_{\mathcal{X}=1}^M$ is defined as
    \begin{equation}
        \check{\lamed}_{ab} := \sum_{1 \leq \mathcal{X} < \mathcal{Y}}^M 2|\rho_{\mathcal{X}\mathcal{Y}}|\tsadi_{ab}^{\mathcal{X}\mathcal{Y}}
    \end{equation}
    with $\check{\lamed}_{ab} = \check{\lamed}_{ba} = \check{\lamed}_{(ab)}$.
\end{definition}
Naturally, these definitions give rise to a notion of \emph{quantum causal indefiniteness},
\begin{equation}
    \check{\delta}(\mathcal{A},\mathcal{B}):= 2|\rho_{\mathcal{A} \mathcal{B}}| \sum_{1 \leq i < j}^N \abs{s_{ij}^{\mathcal{A}} - s_{ij}^{\mathcal{B}}}^0,
\end{equation}
and \emph{quantum total causal indefiniteness}, 
\begin{equation}
    \check{\Delta} := \sum_{1 \leq \mathcal{X} < \mathcal{Y}}^M 2|\rho_{\mathcal{X}\mathcal{Y}}|\delta(\mathcal{X},\mathcal{Y}),
\end{equation}    
as well as a quantum version of the total causal order.

\begin{definition}[Total \emph{quantum} causal order]
    The \textit{total quantum causal order} between a collection of $M$ spacetimes $\{(\mathcal{M}_\mathcal{X},g_{\mathcal{X}})\}_{\mathcal{X}=1}^M$ is 
\begin{equation}
    \check{\mathfrak{s}}_{\mathrm{tot}} := \sum_{1 \leq \mathcal{X} < \mathcal{Y}}^M \check{\ayin}^{\mathcal{X}\mathcal{Y}} \equiv \sum_{1 \leq i < j}^N \check{\lamed}_{ij}.
\end{equation}
\end{definition}

This set of quantifiers incorporates specific quantum mechanical aspects, in particular the coherence between pairs of spacetimes, in the characterisation of indefinite causal order. Future work could explore the operational meaning of these quantum mechanical quantifiers as well as an extension of the knot theoretic or diagrammatic representation that captures them. The latter might require equipping the diagrams with additional structure, encoding the amplitude of the branches of the superposition.

\section{Discussion} \label{sec:discussion}

In this paper, we explored how the indefiniteness of causal order, for an arbitrary finite number of events and spacetimes in superposition, can be quantified and found a compelling connection to knot theory. In particular, we proposed a way to represent superpositions of different ordered sequences of events in terms of knot diagrams. Our results indicate that scenarios with maximally indefinite causal order correspond to torus knots, those with definite causal order to the simple unknot, and, more generally, irreducible causal sequences to prime knots. Moreover, we demonstrated that the indefiniteness of causal order for a superposition of two spacetimes with $N$ events, whose subsequences are either in definite or maximally indefinite causal order, is related to the quadratic term of the Alexander-Conway polynomial of the corresponding knot. This characterisation facilitates the classification of potential causal orderings by representing them as inequivalent knots and using knot invariants to differentiate them. Additionally, since the quadratic term of the Alexander-Conway polynomial is a topological invariant, it immediately follows from this characterisation that causal definiteness or maximal indefiniteness cannot be altered by any topology-preserving map. This result aligns with previous work showing the invariance of causal order under quantum diffeomorphisms \cite{Delahamette2022} and builds the connection to earlier results from the process matrix formalism that demonstrate its invariance under arbitrary continuous and reversible transformations \cite{CastroRuiz2018}.

An important open conjecture is whether the relationship between the indefiniteness of causal order and knot invariants extends to arbitrary causal sequences. This would require proving the equivalence for arbitrary irreducible causal sequences. This might be reached through an analysis of the Seifert matrix $[A_{ij}]$ associated with these prime knots, which are directly related to the quadratic coefficient of the Alexander-Conway polynomial \cite{Tsau2016}. \sam{Though Conjecture \ref{conj:renaming} does not follow from Conjecture \ref{Causal indefiniteness prime knot}, as two different knots may share the same quadratic term of their respective Alexander-Conway polynomial, proving the former would certainly be a major step in proving the latter.} It would also be interesting to determine under what conditions a given knot represents a causal ordering. Although not all knots seem to be related to causal orderings in our construction -- for instance the figure-eight knot, being a knot with four crossings, cannot be mapped to irreducible sequences -- it may be insightful to determine general properties of knots that are indeed associated to indefinite causal structures. For example, are all knots with an odd number of crossings related to some causal structures?

Another potential avenue for future research would be to connect the present work to various diagrammatic frameworks modelling indefinite causal structures. This includes, for example, the process matrix formalism (e.g.~\cite{Oreshkov_2012, Araujo_2015}), extended circuit formalisms \cite{Vanrietvelde2021routedquantum, Ormrod2023causalstructurein, Vanrietvelde2023} and, more generally, frameworks within quantum causal modelling (e.g.~\cite{Costa_2016,barrett2020quantum, Barrett2021rev}). \sam{Though quantum-controlled indefinite causal order cannot violate causal inequalities \cite{Wechs2021}, it would be interesting to understand whether the framework in this paper can be extended to more general setups, e.g.~to quantised worldlines and non-quasiclassical spacetimes which could go beyond this quantum-controlled formalism, and lead to knots and invariants which could be associated to causal inequality violations.}

\sam{An experimental realisation of the quantum 4-switch, which is a special case of our formalism for $N=4$ events, has been performed in \cite{taddei_computational_2021}. More specifically, the authors looked at \acd{the case of} $N=4,M=4$ with causal sequences $ABCD-BADC-CBDA-DACB$. This corresponds to a total causal indefiniteness of $\Delta = 21$. It is shown in \cite{Chiribella_2012,Araujo_2014,taddei_computational_2021, Renner_2021, Renner_2022} that the quantum $N$-switch process can provide a computational advantage over fixed-gate-order quantum circuits, dependent on both $N$ and $M$, for solving a particular phase-estimation problem. In future work, it would be interesting to understand whether the (total) causal indefiniteness $\delta$ and $\Delta$ can be linked to the complexity of certain protocols, and whether a bound on these quantifiers relate to a possible computational advantage. For example, if a given causal sequence in some $(N,M)$ case has $\Delta \geq \Delta_{\text{crit}}(P)$, then computational advantage could be reached for some protocol $P$.}

Moreover, it would be interesting to explore the group theoretic features of some of the constructions in the present work. This includes, on the one hand, the study of the knot groups associated to, in particular, the unknot and torus knots representing definite and maximally indefinite causal order, respectively. On the other hand, it might be fruitful to explore the link to the braid group. This connection is motivated by a theorem by Alexander \cite{Alexander1923}, which states that every knot can be represented as a closed braid. However, this correspondence is not one-to-one \cite{Birman1975}, in the sense that multiple braids can correspond to the same knot, and would thus be assigned the same causal order invariant. This could be traced back to the fact that the braiding distinguishes between different choices of the reference spacetime, whereas the knot construction does not, as we saw in Sec.~\ref{sec:knots}. Establishing the link to the braid group would open the door to explore further connections to topological quantum field theory, such as Chern-Simons theory and anyons.

Finally, the present framework can be generalised in several directions. Firstly, we may allow for infinitely many events across all spacetimes in superposition, thereby considering the continuum limit $N\to \infty$. Secondly, one may want to allow for a continuum of spacetimes in superposition, that is, take the limit $M\to\infty$ and consider states of the form 
\begin{equation}
    \ket{\Psi} = \int_\Omega c(g^\mathcal{X}) \ket{\Tilde{g}^\mathcal{X}} d\mu(g^\mathcal{X}).
\end{equation}
This, however, would require a meaningful measure $\mu(g^\mathcal{X})$ on a measurable set $\Omega$ of spacetime metrics \footnote{This might be feasible when restricting to a superposition of spacetimes $g^\mathcal{X}$ from a tractable subset of spacetime metrics. Dealing with superpositions of arbitrary spacetime metrics, on the other hand, would require finding a path integral for general (Lorentzian) metrics. This, however, is an open problem, with various approaches to quantum gravity offering different directions for solutions, such as (causal) dynamical triangulations \cite{Loll_2019} or spin foam methods \cite{Perez_2012}. \vspace{1cm}}. In both of these limits, it would be necessary to find an appropriate normalisation of the various quantifiers for indefinite causal order in order to prevent divergences.

Further generalisation could involve considering scenarios with indefinite causal occurrences, where some events may or may not be present in some spacetimes. As a simple example, consider the case of a particle in the presence of a black hole, where in one branch of the superposition, the particle falls into the black hole whereas it escapes into infinity in the other. Here, the event of the particle falling into the black hole only occurs in one branch. Currently, we exclude such scenarios by requiring that the number of events is the same across all spacetimes. An even more intriguing scenario arises when a particle in one branch causes a collapse into a black hole by falling into a massive star while this does not happen in the other branch, thus causing a change in the topology of spacetime in a single branch of the superposition. This would require a generalisation of our framework to account for a superposition of topologically inequivalent spacetimes. Lastly, one could consider other exotic topological scenarios, for instance closed timelike curves. Such cases could be naturally considered through so-called fusion diagrams, often used in the field of topological quantum matter. We believe that these are just a few avenues for the future exploration of the intersection between indefinite causal order, knot theory, and quantum gravity.

\begin{acknowledgements}
SF would like to thank Adrian Kent for interesting discussions. We would like to extend our thanks to Rob Pratt for useful insights into finding the maximal bound of $\Delta$ for general finite $M$ and $N$. SF is funded by a studentship from the Engineering and Physical Sciences Research Council. VK acknowledges support through a DOC Fellowship of the Austrian Academy of Sciences.
This research was funded in whole or in part by the Austrian Science Fund (FWF) [10.55776/F71] and [10.55776/COE1]. For open access purposes, the author has applied a CC BY public copyright license to any author accepted manuscript version arising from this submission. Funded by the European Union - NextGenerationEU. 
This publication was made possible through the financial support of the ID 62312 grant from the John Templeton Foundation, as part of The Quantum Information Structure of Spacetime (QISS) Project (qiss.fr). The opinions expressed in this publication are those of the authors and do not necessarily reflect the views of the John Templeton Foundation.
\end{acknowledgements}

\onecolumngrid 
\begin{appendix}
   \sam{ \section{Examples satisfying Conjectures \ref{conj:renaming} and \ref{Causal indefiniteness prime knot}.}

    \label{sec:App conj}

    In total, there are $2!+3!+4! = 32$ knots for $N \leq 4$ events, and $8$ pairs of knots in those which can be associated upon renaming, the others being their own renamings. For example, setting $A':=B, B':= A, C':= D$ and $D':= C$ shows that ABCD-BADC should really be understood as the same pair of event sets as B'A'D'C'-A'B'C'D', i.e.~the resulting knot for this sequence is, by definition, invariant under the choice of base spacetime. On the other hand, setting $A'':=B, B'':= C$ and $C'':= A$ shows that ABC-BCA and C''A''B''-A''B''C'' are again just renamings of the same pair of event sets, with different choices of base spacetime. We now show that those which are not invariant all satisfy Conjectures \ref{conj:renaming} (and \ref{Causal indefiniteness prime knot}) for $N \leq 4$. We write a DT sequence of the knot construction and provide the associated Alexander-Conway polynomial for each of the knots.

    \begin{enumerate}
        \item ABC-BCA and ABC-CAB.

        \begin{figure}[h!]
            \centering
            \subfloat[ABC-BCA, \newline DT: 6,-10,12,2,-4,-8]{\includegraphics[width=0.3\linewidth]{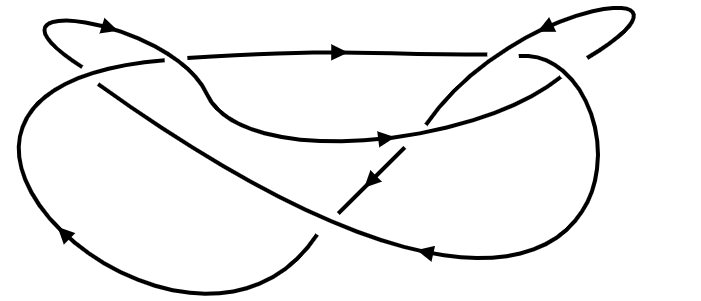}}
            \subfloat[ABC-CAB, \newline DT: -10,-8,12,-4,-2,6]{\includegraphics[width=0.25\linewidth]{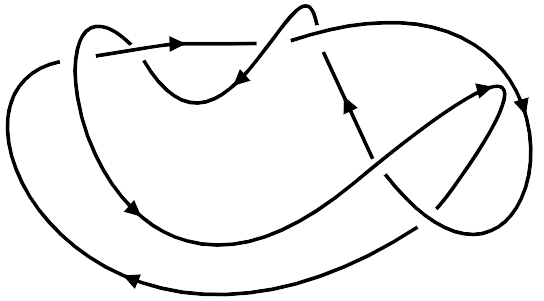}}
            \caption{{1. ABC-BCA and ABC-CAB: $\nabla(z) = 1+2z^2 \Rightarrow \nabla(z)_2 = 2 = \delta$.}}
            \label{fig:ABC-BCA and ABC-CAB}
        \end{figure}

        \item ABCD-BCAD and ABCD-CABD -- clearly these are just the knots of ABC-BCA and ABC-CAB, respectively, since adding D-D at the end does not change the knot (it corresponds to performing two type I Reidermeister moves): $\nabla(z) = 1+2z^2 \Rightarrow \nabla(z)_2 = 2 = \delta$.

        \item ABCD-ACDB and ABCD-ADBC.
        
        \begin{figure}[h!]
            \centering
            \subfloat[ABCD-ACDB, \newline DT: 8,-14,-12,-16,2,-6,-4,10]{\includegraphics[width=0.3\linewidth]{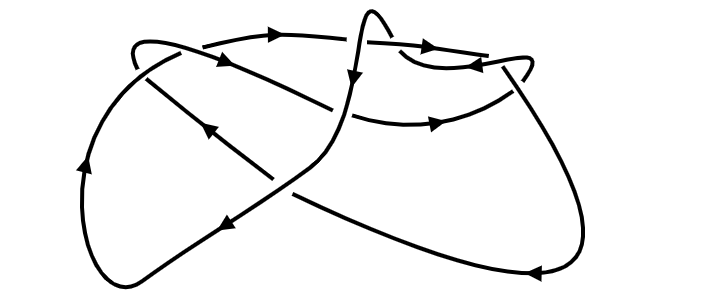}}
            \subfloat[ABCD-ADBC, \newline DT: -12,-10,14,-16,-4,-2,8,6]{\includegraphics[width=0.3\linewidth]{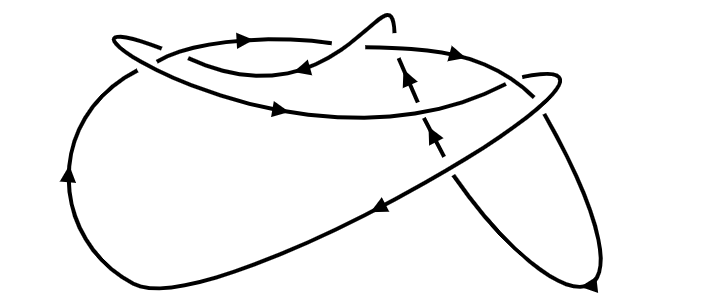}}
            \caption{{3. ABCD-ACDB and ABCD-ADBC: $\nabla(z) = 1+2z^2 \Rightarrow \nabla(z)_2 = 2 = \delta$.}}
            \label{fig:ABCD-ACDB and ABCD-ADBC}
        \end{figure}

        \item ABCD-BCDA and ABCD-DABC.

        \begin{figure}[h!]
            \centering
            \subfloat[ABCD-BCDA, \newline DT: 8, -14, -12, 16, 2, -6, -4, -10]{\includegraphics[width=0.3\linewidth]{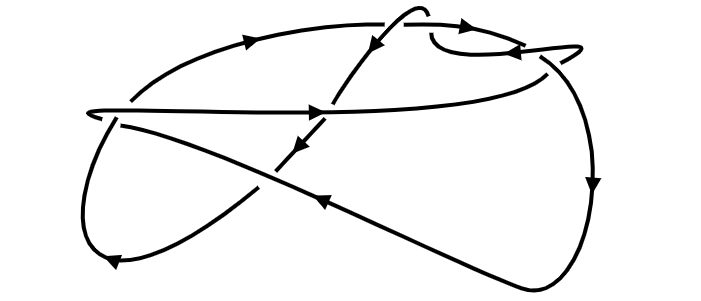}}
            \subfloat[ABCD-DABC, \newline DT: -14, -12, -10, 16, -6, -4, -2, 8]{\includegraphics[width=0.3\linewidth]{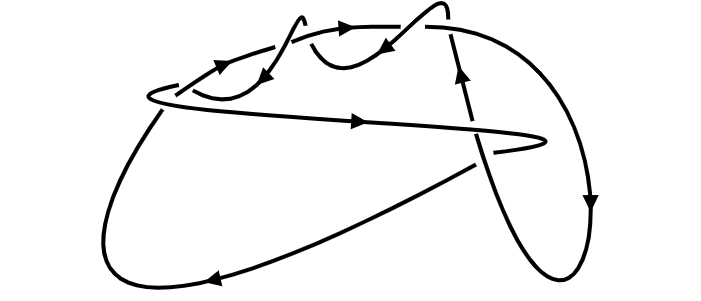}}
            \caption{{4. ABCD-BCDA and ABCD-DABC: $\nabla(z) = 1+3z^2 \Rightarrow \nabla(z)_2 = 3 = \delta$.}}
            \label{fig:ABCD-BCDA and ABCD-DABC}
        \end{figure}

        \item ABCD-BDAC and ABCD-CADB.

        \begin{figure}[h!]
            \centering
            \subfloat[ABCD-BDAC, \newline DT: 14, -10, -20, -18, -4, -22, 24, 2, 12, -6, 8, -16]{\includegraphics[width=0.3\linewidth]{ABCD-BDAC.png}}
            \subfloat[ABCD-CADB, \newline DT: 20, 10, -16, 26, -22, 4, -18, -6, 24, 8, 2, 12, -14]{\includegraphics[width=0.3\linewidth]{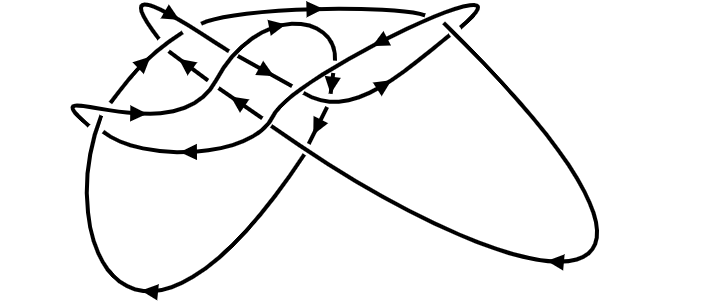}}
            \caption{{5. ABCD-BDAC and ABCD-CADB: $\nabla(z) = 1+3z^2+z^4 \Rightarrow \nabla(z)_2 = 3 = \delta$.}}
            \label{fig:ABCD-BDAC and ABCD-CADB}
        \end{figure}

        \item ABCD-BDCA and ABCD-DACB. 

        \begin{figure}[h!]
            \centering
            \subfloat[ABCD-BDCA, \newline DT: 8,12,14,18,2,-16,4,6,-10]{\includegraphics[width=0.3\linewidth]{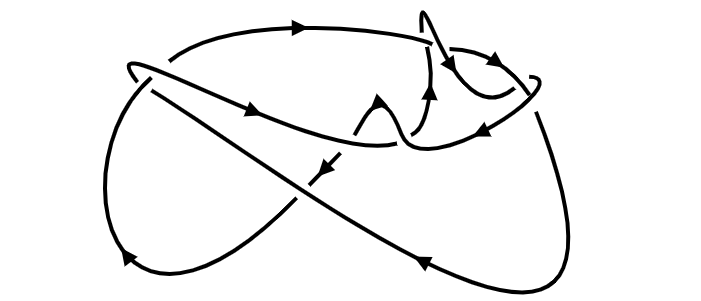}}
            \subfloat[ABCD-DACB, \newline DT: 16, 10, 12, 20, -18, 4, 6, 8, 2, 14]{\includegraphics[width=0.3\linewidth]{ABCD-DACB.png}}
            \caption{{6. ABCD-BDCA and ABCD-DACB: $\nabla(z) = 1+4z^2+2z^4 \Rightarrow \nabla(z)_2 = 4 = \delta$.}}
            \label{fig:ABCD-BDCA and ABCD-CADB}
        \end{figure}

        \item ABCD-CBDA and ABCD-DBAC.

        \begin{figure}[h!]
            \centering
            \subfloat[ABCD-CBDA, \newline DT: 8, 14, -12, 18, 2, -6, 16, 4, -10]{\includegraphics[width=0.3\linewidth]{ABCD-CBDA.png}}
            \subfloat[ABCD-DBAC, \newline DT: 12, 14, -10, 18, -6, -16, 2, 4, 8]{\includegraphics[width=0.3\linewidth]{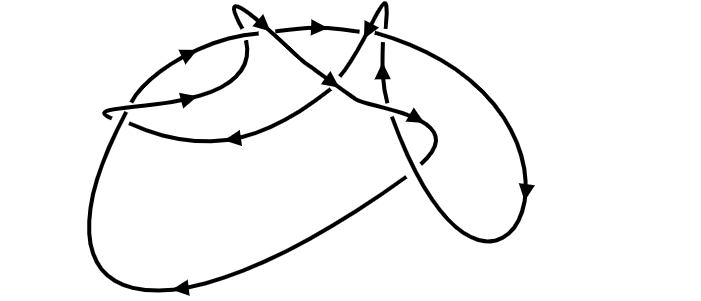}}
            \caption{{7. ABCD-CBDA and ABCD-DBAC: $\nabla(z) =  1+4z^2+2z^4 \Rightarrow \nabla(z)_2 = 4 = \delta$.}}
            \label{fig:ABCD-CBDA and ABCD-DBAC}
        \end{figure}
        
        \item ABCD-CDBA and ABCD-DCAB.
    \end{enumerate}

    \begin{figure}[h!]
            \centering
            \subfloat[ABCD-CDBA, \newline DT: 8, 10, -14, 16, 2, 4, -6, -12]{\includegraphics[width=0.3\linewidth]{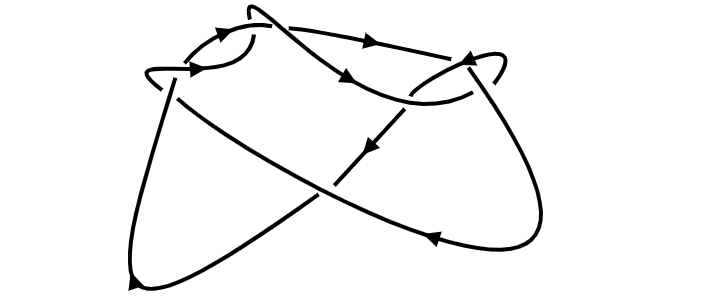}}
            \subfloat[ABCD-DCAB, \newline DT: -12, -10, 14, 16, -4, -2, 8, 6]{\includegraphics[width=0.3\linewidth]{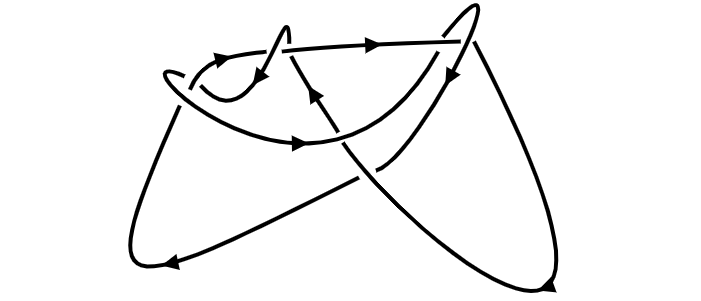}}
            \caption{{8. ABCD-CDBA and ABCD-DCAB: $\nabla(z) = 1+5z^2+2z^4 \Rightarrow \nabla(z)_2 = 5 = \delta$.}}
            \label{fig:ABCD-CDBA and ABCD-DCAB}
        \end{figure}

    The rest of the sequences with $N \leq 4$ events are their own renamings. We not only see that the quadratic coefficients of the Alexander-Conway polynomials are the same, but the knots themselves are the same. We have thus shown that Conjecture \ref{conj:renaming} holds for $N \leq 4$ events, i.e.~the choice of base spacetime is irrelevant for $N \leq 4$. It is conjectured that this holds for any $N \in \mathbb{N}$.
    }
\pagebreak

    \section{Unitarity of Transformation in Operational Encoding}
    \label{app:unitarity}
    {Here, we show that the transformation \eqref{eq:UR} is indeed unitary, that is, $\hat{U}_{\mathrm{R}} \hat{U}_{\mathrm{R}}^\dagger = \hat{U}_{\mathrm{R}}^\dagger \hat{U}_{\mathrm{R}} = \mathbb{1}_{\mathrm{R}}$. This immediately implies the unitarity of the total transformation $\hat{U}_{rel} = \mathbb{1}_G\otimes \hat{U}_R \otimes \mathbb{1}_T$ since the latter acts trivially on the gravitational part of the state and the target system. We start by showing that 	$\hat{U}_{\mathrm{R}} \hat{U}_{\mathrm{R}}^\dagger = \mathbb{1}_R$. 
	\begin{align*}
	\hat{U}_{\mathrm{R}} \hat{U}_{\mathrm{R}}^\dagger &= \int \left(\prod_{n=1}^Nd\tau_n\right)\left(\prod_{m=1}^Nd\tilde{\tau}_m\right)|\det J|\left(\bigotimes_{i=1}^{N-1}|\tau_{i+1}-\tau_i\rangle\underset{\delta(\tau_i-\tilde{\tau}_i)}{\underbrace{\langle\tau_i|\tilde{\tau}_i\rangle}}\langle \tilde{\tau}_{i+1}-\tilde{\tau}_i|_{R_i}\right) 
		\otimes |\sum_{k=1}^N\tau_k\rangle\underset{\delta(\tau_N-\tilde{\tau}_N)}{\underbrace{\langle \tau_N|\tilde{\tau}_N\rangle}}\langle \sum_{l=1}^N\tilde{\tau}_l|_{R_N}\\
		&=\int \left(\prod_{n=1}^Nd\tau_n\right) |\det J | \bigotimes_{i=1}^{N-1}|\tau_{i+1}-\tau_i\rangle\langle \tau_{i+1}-\tau_i|_{R_i} \otimes |\sum_{k=1}^N\tau_k\rangle\langle \sum_{l=1}^N\tau_l|_{R_N}
	\end{align*}
	In the first line we used that, because the coordinate transformation is linear, the Jacobian does not depend explicitly on $\tau$ resp. $\tilde{\tau}$ and thus $|\det J(\tau)| = |\det J(\tilde{\tau})|\equiv |\det J|$ (see below for the concrete expression). To show that this is indeed the identity on $\mathrm{R}$, we change the integration variables from $\tau_{i+1}-\tau_i \to \tau_i$ for $i=1,\dots,N-1$ and $\sum_{k=1}^N \tau_k\to \tau_N$. Note that this is precisely the inverse transformation of that, with respect to which the Jacobian $J$ is defined. Since the Jacobian of the inverse transformation is simply the inverse of the Jacobian, the integration measure is rescaled by $|\det (J^{-1})|=|\det J|^{-1}$. We thus obtain
	\begin{align*}
	\hat{U}_{\mathrm{R}} \hat{U}_{\mathrm{R}}^\dagger =  \int \left(\prod_{n=1}^Nd\tau'_n\right) |\det J |^{-1} |\det J | \bigotimes_{i=1}^{N-1}|\tau'_i\rangle\langle \tau'_i|_{R_i} \otimes |\tau'_N\rangle\langle\tau'_N|_{R_N}
	=	\int \left(\prod_{n=1}^Nd\tau'_n\right) \bigotimes_{i=1}^{N}|\tau'_i\rangle\langle \tau'_i|_{R_i} = \mathbb{1}_{\mathrm{R}}.
	\end{align*}
	Next, let us show $\hat{U}_{\mathrm{R}}^\dagger \hat{U}_{\mathrm{R}} = \mathbb{1}_R$. We have 
	\begin{align*}
		\hat{U}_{\mathrm{R}}^\dagger \hat{U}_{\mathrm{R}} &= \int \left(\prod_{n=1}^Nd\tau_n\right)\left(\prod_{m=1}^Nd\tilde{\tau}_m\right)|\det J|\bigotimes_{i=1}^{N-1} \left(|\tau_i\rangle\underset{\delta(\tilde{\tau}_{i+1}-\tilde{\tau}_i -\tau_{i+1} +\tau_i)}{\underbrace{\langle \tau_{i+1} -\tau_i|\tilde{\tau}_{i+1}-\tilde{\tau}_i\rangle}} \langle \tilde{\tau_i}|_{R_i} \right)\otimes |\tau_N\rangle \underset{\delta(\sum_{k=1}^N(\tilde{\tau}_k-\tau_k))}{\underbrace{\langle \sum_{k=1}^N\tau_k|\sum_{l=1}^N\tilde{\tau}_l\rangle}} \langle \tilde{\tau}_N|_{R_N}.
	\end{align*}
	We can simplify this using the following property of the multivariate delta-distribution:
	\begin{equation*}
		\delta(\vec{g}(\vec{x})) = \frac{\delta(\vec{x}-\vec{x}_0)}{|\det J_{\vec{g}}(\vec{x}_0)|}.
	\end{equation*}
	In our case, $g_i(\vec{\tilde{\tau}}) = \tilde{\tau}_{i+1}-\tilde{\tau}_i -\tau_{i+1} +\tau_i$ for $i=1,\dots, N-1$ and $g_N(\vec{\tilde{\tau}}) = \sum_{k=1}^N(\tilde{\tau}_k-\tau_k)$. When seen as a function of $\vec{\tilde{\tau}}$, the terms proportional to elements of $\vec{\tau}$ do not affect the Jacobian and we simply have $J_{\vec{g}} = J$. The zero of the function inside the delta-distribution is simply given by $\vec{\tilde{\tau}}=\vec{\tau}$. We can thus rewrite the above expression as
	\begin{align*}
		\hat{U}_{\mathrm{R}}^\dagger \hat{U}_{\mathrm{R}} &= \int \left(\prod_{n=1}^Nd\tau_n\right)\left(\prod_{m=1}^Nd\tilde{\tau}_m\right)
			|\det J|\frac{\delta(\vec{\tilde{\tau}} - \vec{\tau})}{|\det J|}\bigotimes_{i=1}^N|\tau_i\rangle\langle \tilde{\tau}_i|_{R_i}= \int \left(\prod_{n=1}^Nd\tau_n\right) \bigotimes_{i=1}^N|\tau_i\rangle\langle \tau_i|_{R_i} = \mathbb{1}_{\mathrm{R}}.
	\end{align*}
	This concludes the proof of unitarity. To simplify the expression for the state after the unitary transformation, we also compute explicitly the constant factor $\sqrt{|\det J|}$ that rescales the wavefunction. To this end, note first that for $N=2$ and $N=3$ the Jacobi determinant is
	\begin{align*}
		\det J_2 \equiv  \det \begin{pmatrix}-1 & 1 \\ 1 & 1\end{pmatrix} = -2, \hspace{0.5cm} \det J_3 \equiv  \det \begin{pmatrix} -1 & 1 & 0 \\ 0 & -1 & 1 \\ 1 & 1 &1 		\end{pmatrix} = -1 \cdot \det \begin{pmatrix}-1 & 1 \\ 1 & 1\end{pmatrix} + 1\cdot \det \begin{pmatrix}1 & 0 \\ -1 & 1\end{pmatrix} = 3.
	\end{align*}
	This generalises to arbitrary $N$ as
	\begin{equation}
		\det J_N = (-1)^N \cdot N,\label{eq:JacobiDet}
	\end{equation}
	which can be proven via induction. In particular, assuming that Eq.~\eqref{eq:JacobiDet} holds for some $N\geq 2$, we can compute the determinant for the $(N+1)$-dimensional Jacobian as follows:
	\begin{align*}
		\det J_{N+1} &= \det \begin{pmatrix}
			-1 & \begin{matrix} 1 & 0 & \dots & 0 \end{matrix} \\
			\begin{matrix} 0 \\ \vdots \\ 0 \\ 1 \end{matrix} & J_N
		\end{pmatrix}		= -1 \cdot \det J_N + (-1)^N\cdot \det \begin{pmatrix} 1 & 0 & 0 & \dots & 0 \\ -1 & 1 & 0 & \dots & 0 \\ \vdots & & & & \vdots \\ 0 & \dots & 0 & -1 & 1 \\ \end{pmatrix}\\ &= (-1)^N
		\cdot N + (-1)^N \cdot 1 = (-1)^{(N+1)-1}(N+1).
\end{align*}			
	where we used the induction hypothesis and the fact that the determinant of a triangular matrix is just the product of its diagonal entries in going from the first to the second line. }
\end{appendix}

\onecolumngrid

\bibliography{library}

\begin{thebibliography}{10}

\bibitem{Hardy_2007}
Lucien Hardy.
\newblock ``Towards quantum gravity: a framework for probabilistic theories with non-fixed causal structure''.
\newblock \href{https://dx.doi.org/10.1088/1751-8113/40/12/S12}{Journal of Physics A: Mathematical and Theoretical {\bf 40}, 3081}~(2007).

\bibitem{Chiribella_2013}
Giulio Chiribella, Giacomo~Mauro D'Ariano, Paolo Perinotti, and Benoit Valiron.
\newblock ``Quantum computations without definite causal structure''.
\newblock \href{https://dx.doi.org/10.1103/PhysRevA.88.022318}{Phys. Rev. A {\bf 88}, 022318}~(2013).

\bibitem{Oreshkov_2012}
Ognyan Oreshkov, Fabio Costa, and {\v C}aslav Brukner.
\newblock ``Quantum correlations with no causal order''.
\newblock \href{https://dx.doi.org/10.1038/ncomms2076}{Nature Communications {\bf 3}, 1092}~(2012).

\bibitem{procopio_experimental_2015}
Lorenzo~M. Procopio, Amir Moqanaki, Mateus Araújo, Fabio Costa, Irati Alonso~Calafell, Emma~G. Dowd, Deny~R. Hamel, Lee~A. Rozema, Časlav Brukner, and Philip Walther.
\newblock ``Experimental superposition of orders of quantum gates''.
\newblock \href{https://dx.doi.org/10.1038/ncomms8913}{Nature Communications {\bf 6}, 7913}~(2015).

\bibitem{Rubino2017}
Giulia Rubino, Lee~A. Rozema, Adrien Feix, Mateus Araújo, Jonas~M. Zeuner, Lorenzo~M. Procopio, Časlav Brukner, and Philip Walther.
\newblock ``Experimental verification of an indefinite causal order''.
\newblock \href{https://dx.doi.org/10.1126/sciadv.1602589}{Science Advances {\bf 3}, e1602589}~(2017).

\bibitem{rubino_experimental_2022}
Giulia Rubino, Lee~A. Rozema, Francesco Massa, Mateus Araújo, Magdalena Zych, Časlav Brukner, and Philip Walther.
\newblock ``Experimental entanglement of temporal order''.
\newblock \href{https://dx.doi.org/10.22331/q-2022-01-11-621}{Quantum {\bf 6}, 621}~(2022).

\bibitem{Goswami2018}
K.~Goswami, C.~Giarmatzi, M.~Kewming, F.~Costa, C.~Branciard, J.~Romero, and A.~G. White.
\newblock ``Indefinite causal order in a quantum switch''.
\newblock \href{https://dx.doi.org/10.1103/PhysRevLett.121.090503}{Phys. Rev. Lett. {\bf 121}, 090503}~(2018).

\bibitem{Araujo_2015}
Mateus Araújo, Cyril Branciard, Fabio Costa, Adrien Feix, Christina Giarmatzi, and {\v{C}}aslav Brukner.
\newblock ``Witnessing causal nonseparability''.
\newblock \href{https://dx.doi.org/10.1088/1367-2630/17/10/102001}{New Journal of Physics {\bf 17}, 102001}~(2015).

\bibitem{Branciard_2016}
Cyril Branciard.
\newblock ``Witnesses of causal nonseparability: an introduction and a few case studies''.
\newblock \href{https://dx.doi.org/10.1038/srep26018}{Scientific Reports {\bf 6}, 26018}~(2016).

\bibitem{Branciard_2015}
Cyril Branciard, Mateus Araújo, Adrien Feix, Fabio Costa, and {\v{C}}aslav Brukner.
\newblock ``The simplest causal inequalities and their violation''.
\newblock \href{https://dx.doi.org/10.1088/1367-2630/18/1/013008}{New Journal of Physics {\bf 18}, 013008}~(2015).

\bibitem{Oreshkov_Giarmatzi_2016}
Ognyan Oreshkov and Christina Giarmatzi.
\newblock ``Causal and causally separable processes''.
\newblock \href{https://dx.doi.org/10.1088/1367-2630/18/9/093020}{New Journal of Physics {\bf 18}, 093020}~(2016).

\bibitem{Abbott_2016}
Alastair~A. Abbott, Christina Giarmatzi, Fabio Costa, and Cyril Branciard.
\newblock ``Multipartite causal correlations: Polytopes and inequalities''.
\newblock \href{https://dx.doi.org/10.1103/PhysRevA.94.032131}{Phys. Rev. A {\bf 94}, 032131}~(2016).

\bibitem{Kissinger_2017a}
Aleks Kissinger and Sander Uijlen.
\newblock ``Picturing indefinite causal structure''.
\newblock \href{https://dx.doi.org/10.4204/eptcs.236.6}{Electronic Proceedings in Theoretical Computer Science {\bf 236}, 87–94}~(2017).

\bibitem{Kissinger_2017b}
Aleks Kissinger and Sander Uijlen.
\newblock ``A categorical semantics for causal structure''.
\newblock In 2017 32nd Annual ACM/IEEE Symposium on Logic in Computer Science (LICS).
\newblock \href{https://dx.doi.org/10.1109/LICS.2017.8005095}{Pages 1--12}.
\newblock ~(2017).

\bibitem{Pinzani_2020}
Nicola Pinzani and Stefano Gogioso.
\newblock ``Giving operational meaning to the superposition of causal orders''.
\newblock In Beno\^it Valiron, Shane Mansfield, Pablo Arrighi, and Prakash Panangaden, editors, {\rm Proceedings 17th International Conference on} Quantum Physics and Logic, {\rm Paris, France, June 2 - 6, 2020}.
\newblock \href{https://dx.doi.org/10.4204/EPTCS.340.13}{Volume 340 of Electronic Proceedings in Theoretical Computer Science, pages 256--278}.
\newblock Open Publishing Association~(2021).

\bibitem{Barrett2021rev}
Jonathan Barrett, Robin Lorenz, and Ognyan Oreshkov.
\newblock ``Cyclic quantum causal models''.
\newblock \href{https://dx.doi.org/10.1038/s41467-020-20456-x}{Nature Communications {\bf 12}, 885}~(2021).

\bibitem{Vanrietvelde2021routedquantum}
Augustin Vanrietvelde, Hl{\'{e}}r Kristj{\'{a}}nsson, and Jonathan Barrett.
\newblock ``Routed quantum circuits''.
\newblock \href{https://dx.doi.org/10.22331/q-2021-07-13-503}{{Quantum} {\bf 5}, 503}~(2021).

\bibitem{Wechs2021}
Julian Wechs, Hippolyte Dourdent, Alastair~A. Abbott, and Cyril Branciard.
\newblock ``Quantum circuits with classical versus quantum control of causal order''.
\newblock \href{https://dx.doi.org/10.1103/PRXQuantum.2.030335}{PRX Quantum {\bf 2}, 030335}~(2021).

\bibitem{Vilasini_2022PRL}
V.~Vilasini and Renato Renner.
\newblock ``Fundamental limits for realizing quantum processes in spacetime''.
\newblock \href{https://dx.doi.org/10.1103/PhysRevLett.133.080201}{Phys. Rev. Lett. {\bf 133}, 080201}~(2024).

\bibitem{Vilasini_2022PRA}
V.~Vilasini and Renato Renner.
\newblock ``Embedding cyclic information-theoretic structures in acyclic space-times: No-go results for indefinite causality''.
\newblock \href{https://dx.doi.org/10.1103/PhysRevA.110.022227}{Phys. Rev. A {\bf 110}, 022227}~(2024).

\bibitem{Ormrod2023causalstructurein}
Nick Ormrod, Augustin Vanrietvelde, and Jonathan Barrett.
\newblock ``Causal structure in the presence of sectorial constraints, with application to the quantum switch''.
\newblock \href{https://dx.doi.org/10.22331/q-2023-06-01-1028}{{Quantum} {\bf 7}, 1028}~(2023).

\bibitem{Vanrietvelde2023}
Augustin Vanrietvelde, Nick Ormrod, Hlér Kristjánsson, and Jonathan Barrett.
\newblock ``Consistent circuits for indefinite causal order''~(2023).
\newblock  \href{http://arxiv.org/abs/2206.10042}{arXiv:2206.10042}.

\bibitem{Pinzani_2023}
Stefano {Gogioso} and Nicola {Pinzani}.
\newblock ``{The Geometry of Causality}''~(2023).
\newblock  \href{http://arxiv.org/abs/2303.09017}{arXiv:2303.09017}.

\bibitem{zych_2019}
Magdalena Zych, Fabio Costa, Igor Pikovski, and {\v{C}}aslav Brukner.
\newblock ``Bell’s theorem for temporal order''.
\newblock \href{https://dx.doi.org/10.1038/s41467-019-11579-x}{Nature Communications {\bf 10}, 3772}~(2019).

\bibitem{SMoller2024gravitational}
Nat{\'{a}}lia S.~M{\'{o}}ller, Bruna Sahdo, and Nelson Yokomizo.
\newblock ``Gravitational quantum switch on a superposition of spherical shells''.
\newblock \href{https://dx.doi.org/10.22331/q-2024-02-12-1248}{{Quantum} {\bf 8}, 1248}~(2024).

\bibitem{Paunkovic_2020}
Nikola Paunkovi{\'{c} } and Marko Vojinovi{\'{c}}.
\newblock ``Causal orders, quantum circuits and spacetime: distinguishing between definite and superposed causal orders''.
\newblock \href{https://dx.doi.org/10.22331/q-2020-05-28-275}{Quantum {\bf 4}, 275}~(2020).

\bibitem{Delahamette2022}
Anne-Catherine de~la Hamette, Viktoria Kabel, Marios Christodoulou, and Časlav Brukner.
\newblock ``Indefinite causal order and quantum coordinates''.
\newblock \href{https://dx.doi.org/10.1103/bnkn-4p3f}{Phys. Rev. Lett.}~(2025).

\bibitem{CastroRuiz2018}
Esteban Castro-Ruiz, Flaminia Giacomini, and {\v{C}}aslav Brukner.
\newblock ``{Dynamics of quantum causal structures}''.
\newblock \href{https://dx.doi.org/10.1103/PhysRevX.8.011047}{Phys. Rev. X {\bf 8}, 011047}~(2018).
\newblock  \href{http://arxiv.org/abs/1710.03139}{arXiv:1710.03139}.

\bibitem{Bose2017}
Sougato Bose, Anupam Mazumdar, Gavin~W. Morley, Hendrik Ulbricht, Marko Toro\ifmmode~\check{s}\else \v{s}\fi{}, Mauro Paternostro, Andrew~A. Geraci, Peter~F. Barker, M.~S. Kim, and Gerard Milburn.
\newblock ``Spin entanglement witness for quantum gravity''.
\newblock \href{https://dx.doi.org/10.1103/PhysRevLett.119.240401}{Phys. Rev. Lett. {\bf 119}, 240401}~(2017).

\bibitem{Marletto2017}
C.~Marletto and V.~Vedral.
\newblock ``Gravitationally induced entanglement between two massive particles is sufficient evidence of quantum effects in gravity''.
\newblock \href{https://dx.doi.org/10.1103/PhysRevLett.119.240402}{Phys. Rev. Lett. {\bf 119}, 240402}~(2017).

\bibitem{Christodoulou_2019}
Marios Christodoulou and Carlo Rovelli.
\newblock ``On the possibility of laboratory evidence for quantum superposition of geometries''.
\newblock \href{https://dx.doi.org/10.1016/j.physletb.2019.03.015}{Physics Letters B {\bf 792}, 64–68}~(2019).

\bibitem{delaHamette2021falling}
Anne-Catherine de~la Hamette, Viktoria Kabel, Esteban Castro-Ruiz, and {\v C}aslav Brukner.
\newblock ``Quantum reference frames for an indefinite metric''.
\newblock \href{https://dx.doi.org/10.1038/s42005-023-01344-4}{Communications Physics {\bf 6}, 231}~(2023).

\bibitem{aspelmeyer2021}
Markus Aspelmeyer.
\newblock ``{When Zeh Meets Feynman: How to Avoid the Appearance of a Classical World in Gravity Experiments}''.
\newblock \href{https://dx.doi.org/10.1007/978-3-030-88781-0_5}{Pages 85--95}.
\newblock Springer International Publishing. Cham~(2022).

\bibitem{Foo_2021}
Joshua Foo, Robert~B Mann, and Magdalena Zych.
\newblock ``{Schrödinger’s cat for de Sitter spacetime}''.
\newblock \href{https://dx.doi.org/10.1088/1361-6382/abf1c4}{Classical and Quantum Gravity {\bf 38}, 115010}~(2021).

\bibitem{sepholearg}
John~D. Norton, Oliver Pooley, and James Read.
\newblock ``{The Hole Argument}''.
\newblock In Edward~N. Zalta and Uri Nodelman, editors, The {Stanford} Encyclopedia of Philosophy.
\newblock Metaphysics Research Lab, Stanford University~(2023).
\newblock {S}ummer 2023 edition.

\bibitem{anandan1997classical}
Jeeva~S. Anandan.
\newblock ``{Classical and Quantum Physical Geometry}''.
\newblock \href{https://dx.doi.org/10.1007/978-94-017-2732-7_3}{Pages 31--52}.
\newblock Springer Netherlands. Dordrecht~(1997).

\bibitem{Hardy2020}
Lucien Hardy.
\newblock ``{Implementation of the Quantum Equivalence Principle}''.
\newblock In Felix Finster, Domenico Giulini, Johannes Kleiner, and J{\"u}rgen Tolksdorf, editors, Progress and Visions in Quantum Theory in View of Gravity.
\newblock \href{https://dx.doi.org/10.1007/978-3-030-38941-3_8}{Pages 189--220}.
\newblock Cham~(2020). Springer International Publishing.

\bibitem{Adlam_2022spacetime}
Emily Adlam, Niels Linnemann, and James Read.
\newblock ``Constructive {Axiomatics} for {Spacetime} {Physics}''.
\newblock \href{https://dx.doi.org/10.1093/9780198922391.001.0001}{Oxford University Press}. ~(2025).

\bibitem{Kabel2024}
Viktoria Kabel, Anne-Catherine de~la Hamette, Luca Apadula, Carlo Cepollaro, Henrique Gomes, Jeremy Butterfield, and {\v C}aslav Brukner.
\newblock ``Quantum coordinates, localisation of events, and the quantum hole argument''.
\newblock \href{https://dx.doi.org/10.1038/s42005-025-02084-3}{Communications Physics {\bf 8}, 185}~(2025).

\bibitem{dowker_classification_1983}
C.H. Dowker and Morwen~B. Thistlethwaite.
\newblock ``Classification of knot projections''.
\newblock \href{https://dx.doi.org/10.1016/0166-8641(83)90004-4}{Topology and its Applications {\bf 16}, 19--31}~(1983).

\bibitem{Simon2023}
Steven~H. Simon.
\newblock ``Topological quantum''.
\newblock Chapter 2. Kauffman Bracket Invariant and Relation to Physics.
\newblock OUP Oxford. ~(2023).

\bibitem{Kauffman2001}
Louis~H Kauffman.
\newblock ``Knots and physics''.
\newblock \href{https://dx.doi.org/10.1142/4256}{World Scientific}. ~(2001).
\newblock 3rd edition.
\newblock  \href{http://arxiv.org/abs/https://www.worldscientific.com/doi/pdf/10.1142/4256}{arXiv:https://www.worldscientific.com/doi/pdf/10.1142/4256}.

\bibitem{Agle2012}
Katherine Ellen~Louise Agle.
\newblock ``Alexander and {C}onway polynomials of {T}orus knots''.
\newblock Master's thesis.
\newblock University of Tennessee.
\newblock ~(2012).

\bibitem{Freyd1985}
P.~Freyd, D.~Yetter, J.~Hoste, W.~B.R. Lickorish, K.~Millett, and A.~Ocneanu.
\newblock ``A new polynomial invariant of knots and links''.
\newblock \href{https://dx.doi.org/10.1090/S0273-0979-1985-15361-3}{Bulletin of the American Mathematical Society{\bf 12}}~(1985).

\bibitem{Norwood1982}
F.~H. Norwood.
\newblock ``Every two-generator knot is prime''.
\newblock \href{https://dx.doi.org/10.1090/s0002-9939-1982-0663884-7}{Proceedings of the American Mathematical Society{\bf 86}}~(1982).

\bibitem{Kauffman1987}
L.~H. Kauffman.
\newblock ``{On Knots}''.
\newblock Volume 115, pages 25--28.
\newblock Princeton University Press. ~(1987).

\bibitem{Tsau2016}
Chichen~M. Tsau.
\newblock ``On the topology of the coefficients of the alexander-conway polynomials of knots''.
\newblock \href{https://dx.doi.org/10.1142/S0218216516500085}{Journal of Knot Theory and its Ramifications{\bf 25}}~(2016).

\bibitem{Costa_2016}
Fabio Costa and Sally Shrapnel.
\newblock ``Quantum causal modelling''.
\newblock \href{https://dx.doi.org/10.1088/1367-2630/18/6/063032}{New Journal of Physics {\bf 18}, 063032}~(2016).

\bibitem{barrett2020quantum}
Jonathan Barrett, Robin Lorenz, and Ognyan Oreshkov.
\newblock ``Quantum causal models''~(2020).
\newblock  \href{http://arxiv.org/abs/1906.10726}{arXiv:1906.10726}.

\bibitem{taddei_computational_2021}
M\'arcio~M. Taddei, Jaime Cari\~ne, Daniel Mart\'{\i}nez, Tania Garc\'{\i}a, Nayda Guerrero, Alastair~A. Abbott, Mateus Ara\'ujo, Cyril Branciard, Esteban~S. G\'omez, Stephen~P. Walborn, Leandro Aolita, and Gustavo Lima.
\newblock ``Computational advantage from the quantum superposition of multiple temporal orders of photonic gates''.
\newblock \href{https://dx.doi.org/10.1103/PRXQuantum.2.010320}{PRX Quantum {\bf 2}, 010320}~(2021).

\bibitem{Chiribella_2012}
Giulio Chiribella.
\newblock ``Perfect discrimination of no-signalling channels via quantum superposition of causal structures''.
\newblock \href{https://dx.doi.org/10.1103/physreva.86.040301}{Physical Review A{\bf 86}}~(2012).

\bibitem{Araujo_2014}
Mateus Araújo, Fabio Costa, and {\v{C}}aslav Brukner.
\newblock ``Computational advantage from quantum-controlled ordering of gates''.
\newblock \href{https://dx.doi.org/10.1103/physrevlett.113.250402}{Physical Review Letters{\bf 113}}~(2014).

\bibitem{Renner_2021}
Martin~J. Renner and \ifmmode \check{C}\else~\v{C}\fi{}aslav Brukner.
\newblock ``Reassessing the computational advantage of quantum-controlled ordering of gates''.
\newblock \href{https://dx.doi.org/10.1103/PhysRevResearch.3.043012}{Phys. Rev. Res. {\bf 3}, 043012}~(2021).

\bibitem{Renner_2022}
Martin~J. Renner and \ifmmode \check{C}\else~\v{C}\fi{}aslav Brukner.
\newblock ``Computational advantage from a quantum superposition of qubit gate orders''.
\newblock \href{https://dx.doi.org/10.1103/PhysRevLett.128.230503}{Phys. Rev. Lett. {\bf 128}, 230503}~(2022).

\bibitem{Alexander1923}
J.~W. Alexander.
\newblock ``A lemma on systems of knotted curves''.
\newblock \href{https://dx.doi.org/10.1073/pnas.9.3.93}{Proceedings of the National Academy of Sciences {\bf 9}, 93--95}~(1923).
\newblock  \href{http://arxiv.org/abs/https://www.pnas.org/doi/pdf/10.1073/pnas.9.3.93}{arXiv:https://www.pnas.org/doi/pdf/10.1073/pnas.9.3.93}.

\bibitem{Birman1975}
Joan~S. Birman.
\newblock ``Braids, links, and mapping class groups. (am-82)''.
\newblock \href{https://dx.doi.org/10.1515/9781400881420}{Princeton University Press}. ~(1975).

\bibitem{Loll_2019}
R.~Loll.
\newblock ``{Quantum Gravity from Causal Dynamical Triangulations: A Review}''.
\newblock \href{https://dx.doi.org/10.1088/1361-6382/ab57c7}{Class. Quant. Grav. {\bf 37}, 013002}~(2020).
\newblock  \href{http://arxiv.org/abs/1905.08669}{arXiv:1905.08669}.

\bibitem{Perez_2012}
Alejandro Perez.
\newblock ``{The Spin Foam Approach to Quantum Gravity}''.
\newblock \href{https://dx.doi.org/10.12942/lrr-2013-3}{Living Rev. Rel. {\bf 16}, 3}~(2013).
\newblock  \href{http://arxiv.org/abs/1205.2019}{arXiv:1205.2019}.

\end{thebibliography}
\bibliographystyle{quantum.bst}

\end{document}